\documentclass[aps,pra,twocolumn,superscriptaddress]{revtex4-2}

\usepackage{amsmath}
\usepackage{physics}
\usepackage[inline]{enumitem}

\usepackage{graphicx}
\usepackage{amsfonts}
\usepackage{amsthm}
\usepackage{booktabs}
\usepackage{xcolor}
\usepackage{upgreek,amssymb}
\definecolor{myblue}{RGB}{50, 100, 200}
\usepackage[colorlinks,allcolors=myblue]{hyperref}

\newtheorem{proposition}{Proposition}
\newtheorem{lemma}{Lemma}

\newcommand{\smallnorm}[1]{\Vert{#1}\Vert}
\newcommand{\smallabs}[1]{ | {#1} |}
\newcommand{\smallket}[1]{ | {#1} \rangle}
\newcommand{\smallbra}[1]{ \langle {#1} |}

\definecolor{darkgreen}{rgb}{.1,.6,.1}

\usepackage[capitalize]{cleveref}

\usepackage{tikz}
\usetikzlibrary{decorations.pathreplacing,calligraphy,decorations.markings}

\definecolor{tensorcolor}{rgb}{0.65,0.77,0.95}
\definecolor{tensormatcolor}{rgb}{0.8,0.8,0.85}
\definecolor{btensorcolor}{rgb}{0.65,0.50,0.69}
\definecolor{whitetensorcolor}{rgb}{0.93,0.93,0.93}
\definecolor{vectorcolor}{rgb}{0.98, 0.91, 0.71}

\newcommand{\gate}[2]{
    \begin{scope}[shift={(#1)}]
        \draw[very thick] (-1.2, 0) -- (1.2, 0);
        \draw[thick, fill=whitetensorcolor, rounded corners=2pt] (-0.8, -0.6) rectangle (0.8, 0.6);
	    \draw (0,0) node {\scriptsize #2};
    \end{scope}
        }

\newcommand{\KTensor}[2]{
	\begin{scope}[shift={(#1)}]
		\draw[very thick] (-1.2,0) -- (1.2,0);
		\draw[very thick] (0,1.2) -- (0,0);
        \draw[thick, fill=tensorcolor, rounded corners=2pt] (-0.6,-0.6) rectangle (0.6,0.6);
		\draw (0,0) node {\scriptsize #2};
	\end{scope}
}

\newcommand{\UT}[2]{
	\begin{scope}[shift={(#1)}]
		\draw[very thick] (-1.2,0) -- (1.2,0);
		\draw[very thick] (0,1.2) -- (0,-1.2);
        \draw[thick, fill=whitetensorcolor, rounded corners=2pt] (-0.6,-0.6) rectangle (0.6,0.6);
		\draw (0,0) node {\scriptsize #2};
	\end{scope}
}

\newcommand{\KGate}[3]{
	\begin{scope}[shift={(#1)}]
		\draw[very thick] (-1.2,0.9) -- (1.2,0.9);
		\draw[very thick] (-1.2,-0.9) -- (1.2,-0.9);
        \draw[thick, fill=tensormatcolor, rounded corners=2pt] (-0.6,-1.5) rectangle (0.6,-0.3);
        \draw[thick, fill=tensormatcolor, rounded corners=2pt] (-0.6,1.5) rectangle (0.6,0.3);
		\draw (0,-0.9) node {\scriptsize #2};
        \draw (0,0.9) node {\scriptsize #3};
	\end{scope}
}

\newcommand{\KDTensor}[2]{
	\begin{scope}[shift={(#1)}]
		\draw[very thick] (-1.2,0) -- (1.2,0);
		\draw[very thick] (0,-1.2) -- (0,0);
        \draw[thick, fill=tensorcolor, rounded corners=2pt] (-0.6,-0.6) rectangle (0.6,0.6);
		\draw (0,0) node {\scriptsize #2};
	\end{scope}
}

\newcommand{\TTensor}[3]{
	\begin{scope}[shift={(#1)}]
        \KTensor{(0,-.9)}{#2};
        \KDTensor{(0,.9)}{#3};
	\end{scope}
}

\newcommand{\CTensor}[2]{
	\begin{scope}[shift={(#1)}]
        \draw[very thick] (-1.2, -.9) -- (1.2, -.9);
        \draw[very thick] (-1.2, .9) -- (1.2, .9);
        \draw[thick, fill=tensorcolor, rounded corners=2pt] (-0.9, -1.5) rectangle (0.9, 1.5);
        \draw (0, 0) node {\scriptsize #2};
	\end{scope}
}

\newcommand{\UDTensor}[2]{
	\begin{scope}[shift={(#1)}]
        \draw[very thick] (-1.2, -.9) -- (1.2, -.9);
        \draw[very thick] (-1.2, .9) -- (1.2, .9);
        \draw[thick, fill=whitetensorcolor, rounded corners=2pt] (-0.9, -1.5) rectangle (0.9, 1.5);
        \draw (0, 0) node {\scriptsize #2};
	\end{scope}
}

\newcommand{\UTensor}[3]{
	\begin{scope}[shift={(#1)}]
        \gate{(0,-.9)}{#2};
        \gate{(0,.9)}{#3};
	\end{scope}
}

\newcommand{\bTensor}[3]{
	\begin{scope}[shift={(#1)}]
    \ifnum#3=-1
	    \draw [very thick] (0,0) to  (1,0);
    \fi

    \ifnum#3=1
        \draw [very thick] (-1,0) to (0,0);
    \fi

		\filldraw[color=black, fill=whitetensorcolor, thick] (0,0) circle (\stradius);
	\draw (0,0) node {#2};
	\end{scope}
}

\newcommand{\BTensor}[4]{
	\begin{scope}[shift={(#1)}]
        \bTensor{(0,0.9)}{#3}{#4};
        \bTensor{(0,-0.9)}{#2}{#4};
	\end{scope}
}

\newcommand{\SingleDots}[2]{
	\begin{scope}[shift={(#1)}]
      \draw [very thick, dotted] (-#2/2,0) to (#2/2,0);
	\end{scope}
}

\newcommand{\DoubleDots}[2]{
	\begin{scope}[shift={(#1)}]
      \SingleDots{0,0.9}{#2};
      \SingleDots{0,-0.9}{#2};
	\end{scope}
}

\newcommand{\sidetensorright}[2]{
    \begin{scope}[shift={(#1)}]
    	\draw [very thick] (-\doubledx+1,0.9) to  [bend left=90] (-\doubledx+1,-0.9);
        \draw [very thick] (-\doubledx+1,0.9) -- (-\doubledx+0.5,0.9);
        \draw [very thick] (-\doubledx+1,-0.9) -- (-\doubledx+0.5,-0.9);
    	\filldraw[color=black, fill=vectorcolor, thick] (-\doubledx+1.4,0) circle (\stradius);
    	\draw (-\doubledx+1.4,0) node {#2};
    \end{scope}
}

\newcommand{\sidetensorleft}[2]{
    \begin{scope}[shift={(#1)}]
        \draw [very thick] (\doubledx-1,0.9) to  [bend right=90] (\doubledx-1,-0.9);
        \draw [very thick] (\doubledx-1,0.9) -- (\doubledx-0.5,0.9);
        \draw [very thick] (\doubledx-1,-0.9) -- (\doubledx-0.5,-0.9);
	    \filldraw[color=black, fill=vectorcolor, thick] (\doubledx-1.4,0) circle (\stradius);
	    \draw (\doubledx-1.4,0) node {#2};
    \end{scope}
}

\newcommand{\sidetensorleftlarger}[2]{
    \begin{scope}[shift={(#1)}]
        \draw [very thick] (\doubledx-1,0.9) to  [bend right=90] (\doubledx-1,-0.9);
        \draw [very thick] (\doubledx-1,0.9) -- (\doubledx-0.5,0.9);
        \draw [very thick] (\doubledx-1,-0.9) -- (\doubledx-0.5,-0.9);
	    \filldraw[color=black, fill=vectorcolor, thick] (\doubledx-1.4,0) ellipse (0.9cm and 0.5cm);
	    \draw (\doubledx-1.4,0) node {#2};
    \end{scope}
}

\newcommand\doubledx{1.6}
\newcommand\singledx{2.2}
\newcommand\dx{1.8}

\newcommand\stradius{0.6}

\newcommand\subsetsim{\mathrel{%
  \ooalign{\raise0.2ex\hbox{$\subset$}\cr\hidewidth\raise-0.8ex\hbox{\scalebox{0.9}{$\sim$}}\hidewidth\cr}}}

\begin{document}

%\preprint{}

\title{Theory of quantum-enhanced interferometry with general Markovian light sources}

\author{Erfan Abbasgholinejad}
\affiliation{Department of Electrical and Computer Engineering, University of Washington, Seattle, WA 98195}
%\thanks{}
\author{Daniel Malz}
\affiliation{Department of Mathematical Sciences, University of Copenhagen, 2100 Copenhagen, Denmark}
\author{Ana Asenjo-Garcia}
\affiliation{Department of Physics, Columbia University, New York, NY 10027}
\author{Rahul Trivedi}
\email{rahul.trivedi@mpq.mpg.de}
\affiliation{Max-Planck-Institut für Quantenoptik, Hans-Kopfermann-Str.~1, 85748 Garching, Germany}
\affiliation{Department of Electrical and Computer Engineering, University of Washington, Seattle, WA 98195}


\date{\today}

\begin{abstract}
Quantum optical systems comprising quantum emitters interacting with engineered optical modes generate non-classical states of light that can be used as resource states for quantum-enhanced interferometry. However, outside of well-controlled systems producing either single-mode states (e.g.~Fock states or squeezed states) or highly symmetric multi-mode states (e.g.~superradiant states), their potential for quantum advantage remains uncharacterized. In this work, we develop a  framework to analyze quantum enhanced interferometry with general Markovian quantum light sources. First, we show how to compute the quantum Fisher Information (QFI) of the photons emitted by a source efficiently by just tracking its internal dynamics and without explicitly computing the state of the emitted photons. We then use this relationship to elucidate the connection between the level structure and spectrum of the source to a potential quantum advantage in interferometry. Finally, we analyze optimal measurement protocols that can be used to achieve this quantum advantage with experimentally available optical elements. In particular, we show that tunable optical elements with Kerr non-linearity can always be harnessed to implement the optimal measurement for any given source. Simultaneously, we also outline general conditions under which linear optics and photodetection is enough to implement the optimal measurement.
\end{abstract}

%\keywords{}

\maketitle
\raggedbottom

% Quantum sensing explores the ultimate precision achievable in sensors by leveraging the principles of quantum mechanics. By exploiting phenomena such as superposition and entanglement, quantum sensors can achieve sensitivities that surpass those of classical sensors. These advanced sensors can be utilized in a wide range of applications, including medical imaging and fundamental physics research. One specific set up for quantum sensing is in interferometry, which uses the interference of light to make precise measurements of physical properties such as refractive index!, and gravitational waves.\\
\section{Introduction}
Leveraging quantum many-body correlations is well known to provide a quadratic advantage in sensing \cite{giovannetti2011advances, toth2014quantum, dowling2015quantum, demkowicz2015quantum}. Using classical sensors with $N$ probes, the precision for estimating an unknown parameter is limited to the standard quantum limit (SQL) which scales as $1/\sqrt{N}$. However, employing probes in a highly correlated state, such as a Greenberger–Horne–Zeilinger (GHZ) state \cite{greenberger1989going}, allows us to surpass this limit and achieve Heisenberg-limited (HL) scaling  of $1/N$ \cite{fox2006quantum, demkowicz2015quantum, sanders1995optimal}. A specific and practically relevant application of quantum-enhanced sensing is interferometry, where the probes are photons emitted from a non-classical source, and the parameter to be sensed is an unknown phase. This phase is typically applied in one arm of a Mach-Zehnder interferometer (MZI). The interferometer is then illuminated with the probe photons and the photons at the MZI output are detected to infer the unknown phase. Such a scheme, using ``classical" photonic states, such as coherent states, senses the unknown phase at the standard quantum limit. In contrast, quantum photonic states, such as NOON states \cite{bollinger1996optimal} or twin Fock states \cite{holland1993interferometric}, can achieve the Heisenberg limit.

However, while several theoretical proposals exist for generating quantum states such as NOON states or Fock states \cite{dowling2008quantum,lee2002quantum, gerry2001generation, mccusker2009efficient, kapale2007bootstrapping, cable2007efficient, pryde2003creation, cirac1993preparation, saxena2024boundary}, they largely require strong photon blockade as well as a high degree of controllability in the quantum emitters \cite{huang2018nonreciprocal, deng2024quantum, bishop2009proposal, chen2010heralded, muller2015coherent, omran2019generation, trivedi2018few}, which is hard to achieve at optical frequencies \cite{tang2015quantum, snijders2018observation}. This has limited the number photons that can be prepared in such quantum states \cite{mitchell2004super, hofheinz2008generation, nagata2007beating, slussarenko2017unconditional}.

Rather than asking which states saturate the Heisenberg limit and then finding ways to prepare them, a practically oriented approach is to ask which quantum-light sources are both simple and able to deliver a quantum advantage in metrology (but not necessarily Heisenberg scaling). Notably, significant experimental efforts have focused on efficiently generating squeezed states \cite{schnabel2017squeezed, vahlbruch2016detection, safavi2013squeezed} that go beyond the standard quantum limit \cite{caves1981quantum, wiseman2009quantum} and are generated by strongly pumping bulk optical non-linearities \cite{andersen201630, orozco1987squeezed, wu1986generation, scully1997quantum}. Another recent line of work has theoretically analyzed the metrological potential of photonic states generated by the collective superradiant decay in the waveguide Dicke limit \cite{dicke1954coherence} and have shown that these states achieve a Heisenberg-limited scaling with the number of photons and are only $\sim$18$\%$ worse than Fock states \cite{paulisch2019quantum}.\\
% As a result, there has been interest in exploring more experimentally feasible quantum states for metrology, created through the emission of multiple emitters into optical modes. Notably, studies have shown that quantum states of light produced through such configurations, such as squeezed states \cite{caves1981quantum, orozco1987squeezed} and super-radiant states \cite{paulisch2019quantum}, can still offer advantages over the SNL.\\

\begin{figure*}[htpb]
	\centering
	\includegraphics[width=0.9\linewidth]{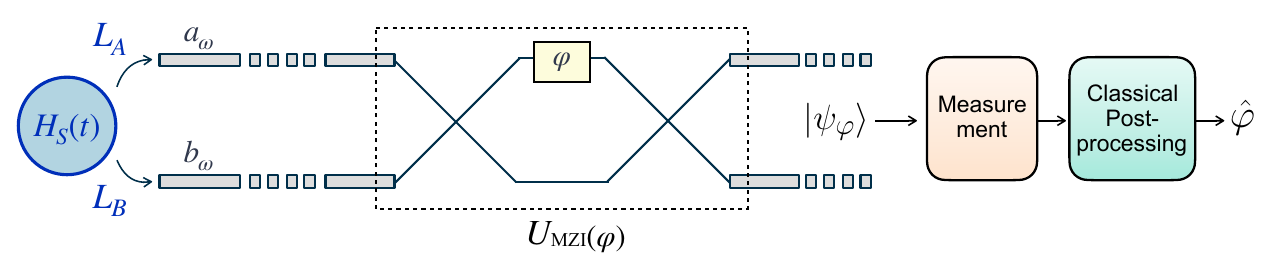}
	\caption{\label{fig:mzi}Schematic depiction of the setup considered in this paper. A source emits photons into input ports of the MZI through the jump operators $L_A$ and $L_B$. The generated photonic state passes though a MZI that includes an unknown sensing parameter, $\varphi$ and the output state $\ket{\psi_\varphi}$ then depends on $\varphi$. This state is then subsequently measured and the measurement outcome is classically postprocessed to output an estimate $\hat{\varphi}$ of $\varphi$.}
\end{figure*}

Going beyond individual examples, it would be desirable to assess more generally the metrological potential of different types of quantum-optical light sources.
An important quantity to calculate is the Quantum Fisher Information (QFI)~\cite{braunstein1994statistical}, which characterizes the maximum sensitivity the interferometer can achieve \cite{helstrom1969quantum, holevo2011probabilistic}. A key obstacle here is that the QFI for photonic states is, in general, challenging to compute since it usually requires being able to calculate the full photonic state, which often has a large number of photons in possibly multiple modes. {While this question has been studied in the context of spectroscopy with continuous-field measurements \cite{yang2023efficient, gammelmark2014fischer}, computing QFI for interferometry has been largely limited to either single-mode states \cite{demkowicz2015quantum} or certain specific highly symmetric multi-mode states \cite{paulisch2019quantum}.} Furthermore, the QFI only provides an upper bound on the sensitivity achievable in the spectroscope. Whether an optically implementable measurement protocol can be devised to achieve this maximal sensitivity also remains an important open question and has been answered only for very specific set of photonic states \cite{pezze2013ultrasensitive, campos2003optical,hofmann2009all, paulisch2019quantum}.

In this paper, we develop a framework to address these problems in a set up where a general source emit photons into each arm of the Mach-Zehnder interferometer to sense an unknown phase $\varphi$ [Fig.~\ref{fig:mzi}]. The only assumption that we make on the source is that its coupling with the output photon field is well described within the Markov approximation. We explicitly connect the QFI of the emitted photons to their two-time correlation functions, which can be computed by only tracking the internal dynamics of the source. Unlike previous works \cite{paulisch2019quantum}, we do not require an explicit computation of the output many-body photonic state, which allows us to numerically analyze a much larger class of photon sources \cite{aoki2006observation, akimov2007generation, lund2008experimental, wallraff2004strong}. Furthermore, using this result for the QFI, we then analyze how the source spectrum impacts the possibility of quantum metrological advantage and demonstrate that the photons generated by driving a source with a single ground state will always exhibit a standard quantum limited (SQL) scaling. However, sources with multiple ground states can generate photons with Heisenberg limited (HL) scaling. For instance, we show that a simple $\pi-$level system with a continous-wave drive can generate photons with quantum metrological advantage.

Finally, we explore the optimal measurement protocol that is needed to achieve the QFI-limited sensitivity. We show that, irrespective of the source, such a measurement can always be implemented with a dynamically driven non-linear resonator and photodetectors. Furthermore, since it is experimentally much easier to implement linear optics, we also study the optimality of measurements that can be implemented with only dynamically controllable linear optical elements and photodetectors and outline conditions when such measurements can be optimal.\\

\section{QFI for general light sources}
\subsection{Setup} 
The interferometric setup that we consider throughout this paper is shown in Fig.~\ref{fig:mzi} --- a source of photons emits into the two input ports (labeled as port A and port B) of the MZI, which has the unknown phase $\varphi$ to be sensed. After passing through the MZI, the photons are in a $\varphi-$dependent state $\ket{\psi_\varphi}$ which is then measured and the measurement outcome is classically post-processed to produce an estimate of the unknown phase $\varphi$. We model the input ports of the MZI as one-dimensional propagating fields (e.g.,~fields in a waveguide mode or a collimated free space beam) with annihilation operators $a_\omega$ and $b_\omega$ corresponding to the frequency $\omega$. Within the Markov approximation \cite{breuer2002theory}, the source-port Hamiltonian can be expressed as:
\begin{subequations}\label{eq:model}
\begin{align}
H(t) = H_S(t) + H_P + V_{SP},
\end{align}
where $H_S(t)$ is a possibly time-dependent Hamiltonian acting on the source, $H_P$ is the Hamiltonian describing photon propagation in the ports, i.e.,
\begin{align}
H_P = \sum_{x \in \{a, b\}}\int_{-\infty}^\infty \omega \tilde{x}_\omega^\dagger \tilde{x}_\omega d\omega,
\end{align}
and
\begin{align}
V_{SP} = \sum_{x\in \{a, b\}} \int_{-\infty}^\infty \big(\tilde{x}_\omega^\dagger L_X + \text{h.c.}\big)\frac{d\omega}{\sqrt{2\pi}},
\end{align}
\end{subequations}
is the interaction Hamiltonian between the source and port. Here $L_A, L_B$ are the operators, acting only on the Hilbert space of the source, through which the source couples to the two output ports. Equation (\ref{eq:model}) follows the Markov approximation for the model by (i) assuming that the frequencies in the port extend from $-\infty$ to $\infty$, with the addition of fictitious negative frequencies being justified when the natural resonant frequency of the source is much larger than other frequency scales in the problem and (ii) assuming that the source interacts identically with each frequency $\omega$ in the output ports.

It is convenient to introduce the time-domain annihilation operator for the port, ${x}_\tau = \int_{\mathbb{R}}\tilde{x}_\omega e^{-i\omega \tau} d\omega/\sqrt{2\pi}$, which can be interpreted as destroying excitations at the time $\tau$ in the port. Working in the interaction picture with respect to the port Hamiltonian $H_P$, we obtain the source-port Hamiltonian $H_I(t)$:
\[
H_I(t) = H_S(t) + \sum_{x\in \{a, b\}} \big({x}_{t} L_X^\dagger +\text{h.c.}\big).
\]
Assuming that the source is initially in state $\ket{\phi_{S, i}}$ and the ports are initially in a vacuum state, and finally projecting the source onto the state $\ket{\phi_{S, f}}$, the state of the light channel after time $T$ is expressed as:
\begin{subequations}\label{eq:psi}
\begin{equation}
    \ket{\psi} = \frac{1}{\mathcal{N}}\bra{\phi_{S, f}}U_I(T, 0) \ket{\phi_{S, i} , \text{vac}},
\end{equation}
where
\begin{align}
    U(T, 0) = \overrightarrow{\mathcal{T}} \text{exp}\left({\int_0^T H_{I}(s) ds}\right),
\end{align}
\end{subequations}
with $\overrightarrow{\mathcal{T}}$ being the time-ordering operator that orders times in decreasing order and $\mathcal{N}$ chosen to ensure that $\ket{\psi}$ is normalized. The projection onto a final state $\ket{\phi_{S, f}}$ maybe performed explicitly via a projective measurement on the source and post-selecting on the measurement outcome. More realistically, in an experiment, we would often wait for a sufficiently long time $T$ for the source to have decayed into its ground state which, if unique, would disentangle from the photons and thus effectively projecting the source onto its ground state. Finally, while we consider the more general setting of the same source emitting into both the ports, in many experimental setups (such as twin-Fock state interferometry \cite{holland1993interferometric}), it is more common to consider two independent sources emitting separately into the two ports. In this case, $H_S(t) = H_A(t) \otimes I_B + I_A\otimes  H_B(t)$ and the operator $L_A$ will act only on the Hilbert space of the source coupling to port A and the operator $L_B$ will act only on the Hilbert space of the source coupling to port B. Consequently, the state of photons $\ket{\psi}$ will be a separable state in between the two ports i.e. $\ket{\psi} = \ket{\psi_A}\otimes \ket{\psi_B}$.

The ports then pass through an MZI with the which implements a unitary transformation dependent on the phase $\varphi$ that is to be sensed and can be expressed as
\begin{align}\label{eq:MZI_transformation}
    U^\dag_{\text{MZI}}(\varphi) \begin{bmatrix} {a}_t \\ {b}_t \end{bmatrix} U_{\text{MZI}}(\varphi) &= \begin{bmatrix} \cos \varphi & \sin \varphi \\
    -\sin \varphi & \cos \varphi\end{bmatrix}\begin{bmatrix} {a}_t \\ {b}_t \end{bmatrix}.
\end{align}
Consequently, the final photonic state at the output of the MZI is given by
\begin{equation}\label{eq:mzi_unitary}
    \ket{\psi_\varphi} = U_{\text{MZI}}(\varphi) \ket{\psi},
\end{equation}
which depends on the phase $\varphi$. In the next subsection, we start from this expression for the state $\ket{\psi_\varphi}$ and compute its QFI with respect to $\varphi$.

\subsection{Computing QFI efficiently}
The QFI of the $\varphi-$dependent state $\ket{\psi_\varphi}$ evaluated at $\varphi = \varphi_0$ is formally given by \cite{liu2020quantum}
\begin{align}\label{eq:QFI_def}
    \text{QFI} = 4\bigg(\norm{\frac{\partial}{\partial \varphi} \ket{\psi_\varphi}}^2 - \bra{\psi_\varphi}\frac{\partial}{\partial \varphi}\ket{\psi_\varphi}\bigg)_{\varphi =\varphi_0} .
\end{align}
As is typical while studying quantum enhanced interferometry, we will consider the QFI at $\varphi = 0$ since the measurement of a non-zero $\varphi$ can always be done in two steps --- first making an approximate estimate of $\varphi$ and then biasing the MZI around $\varphi = 0$ using the estimated phase \cite{barndorff2000fisher, hayashi2011comparison, yang2019attaining}.

One approach to developing an explicit expression for $\text{QFI}$, which has been adopted in Ref.~\cite{paulisch2019quantum}, is to first calculate the wave-functions
\[\psi_\varphi(t_1, t_2 \dots t_n; s_1, s_2 \dots s_m) = \bra{\text{vac}}\prod_{i = 1}^n {a}_{t_i} \prod_{i = 1}^m {b}_{s_i} \ket{\psi_\varphi},\]
associated with $\ket{\psi_\varphi}$, which can be obtained entirely in terms of the source's effective Hamiltonian $H_\text{eff}(t) = H_S(t) - i\sum_{x \in \{a, b\}} L_x^\dagger L_x / 2$, using input-output formalism \cite{trivedi2018few, xu2015input}, and then evaluate $\text{QFI}$. The complexity of evaluating $\psi_\varphi(t_1, t_2 \dots t_n; s_1, s_2 \dots s_m)$ grows exponentially with $n, m$ in the absence of any special symmetries in the source, making it hard to employ this approach beyond very special cases \cite{paulisch2019quantum}.

We take an alternative route and first express $\text{QFI}$ in terms of the correlation functions of $\ket{\psi_\varphi}$. We begin by noting that the MZI unitary in Eq.~(\ref{eq:mzi_unitary}) can be considered as being generated by the time dependent Hamiltonian \cite{trivedi2019point}:
\begin{align}
    H_\text{MZI}(t, \varphi) = 2 \tan \bigg(\frac{\varphi}{2}\bigg) h(t) \text{ with } h(t) = i({a}_{ t}^\dagger {b}_{t} - {a}_{t}{b}_{t}^\dagger).
\end{align}
Since $\ket{\psi_\varphi} = \overrightarrow{\mathcal{T}}\text{exp}(-i\int_0^T H_\text{MZI}(s, \varphi) ds) \ket{\psi}$, we can now use Duhamel's formula \cite{horn2012matrix} to evaluate $\partial \ket{\psi_\varphi}/ \partial \varphi$ at $\varphi = 0$:
\begin{align}\label{eq:deriv_psi_phi}
    \frac{\partial}{\partial \varphi} \ket{\psi_\varphi} \bigg |_{\varphi = 0} = -i H_d \ket{\psi} \text{ where }H_d = \int_0^T h(s) ds .
\end{align}
{It then follows from Eq.~(\ref{eq:QFI_def}) that $\text{QFI} = 4(\langle H_{d}^2 \rangle- \langle H_{d} \rangle^2)$ or equivalently:
\begin{widetext}
    \begin{align}\label{eq:QFI_corr_fun}
    \text{QFI} = 8 \int_{0}^{T} \int_{0}^{T} \text{Re}\big( \mathcal{C}^{(2)}_{a, b; b, a}(t, s) - \mathcal{C}^{(2)}_{b, b; a, a}(t, s) \big) d t d s  +  4\int_{0}^{T} \big(\mathcal{C}^{(1)}_{a; a}(t) + \mathcal{C}^{(1)}_{b; b}(t) \big) d t + 16\abs{ \int_{0}^{T} \text{Im}\big(\mathcal{C}^{(1)}_{a; b}(t)\big) d t }^2,
\end{align}
\end{widetext}
where $\mathcal{C}_{x^1 \dots x^n; y^1\dots y^n}^{(n)}(t_1 \dots t_n)$, for $x^i, y^i \in \{a, b\}$ is an $n-$point correlation function given by
\[
\mathcal{C}^{(n)}_{x^1 \dots x^n; y^1 \dots y^n}(t_1 \dots t_n) = \left\langle \prod_{i = 1}^n x^{i\dagger}_{t_i} \prod_{i = 1}^n y^i_{t_i}  \right\rangle.
\]
Equation (\ref{eq:QFI_corr_fun}) expresses the QFI of the emitted photonic state in terms of its  correlation functions --- as we detail in Appendix \ref{app:qfi_measurement}, this also yields a method to extract the QFI from two-photon correlation functions without necessarily having to perform the full sensing measurement to measure the unknown phase.

From the perspective of numerical calculations of QFI it is convenient to express it in terms of the internal dynamics of the source, which in turn is captured by the Lindbladian $\mathcal{L}_S(t)$:
\begin{align}
    \mathcal{L}_S(t) = -i[H_S(t),\cdot]+\sum_{X\in\{A,B\}} \mathcal{D}_{L_X},
\end{align}
where $\mathcal{D}_L(X) = L X L^\dagger-\{L^\dagger L, X\}/2$, and the associated channel $\mathcal{E}_S(t, s) = \overrightarrow{\mathcal{T}}\text{exp}(\int_s^t \mathcal{L}_S(\tau)d\tau)$. In Appendix \ref{app:QFI_comp}, we show that this can be accomplished by using the input-output formalism \cite{gardiner2004quantum} and the quantum regression theorem \cite{carmichael2009open, carmichael2013statistical}. In particular, we show that the correlator $\mathcal{C}^{(n)}_{x^1 \dots x^n; y^1 \dots y^n}(t_1 \dots t_n)$ can be expressed as
\begin{align}\label{eq:qfi_system_dynamics}
    &\mathcal{C}^{(n)}_{x^1\dots x^n; y^1\dots y^n}(t_1\dots t_n) =\nonumber\\
    &\quad \frac{1}{\mathcal{N}^2}\bra{\phi_f}\bigg[\prod_{i=1}^{n}\mathcal{E}_S(t_{i -1 },t_{i})\mathcal{M}_i \bigg]\mathcal{E}_S(t_n,0)(\rho_{S, i})\ket{\phi_f},
\end{align}
where we have assumed $T = t_0 \geq t_1\geq t_2\geq \dots \geq t_n \geq 0$, $\rho_{S,i} = \ket{\phi_{S, i}}\!\bra{\phi_{S, i}}$ and $\mathcal{M}_i(\cdot) = L_{Y^i}(\cdot)L_{X^i}^\dagger$. 
Furthermore, the normalization constant $\mathcal{N}$ can also be expressed as $\mathcal{N} = (\bra{\phi_{S, f}}\mathcal{E}_S(T,0)(\rho_{S,i}) \ket{\phi_{S, f}})^{1/2} $.}

We emphasize that, for a source with $D$ levels, the computation of the QFI from Eqs.~(\ref{eq:QFI_corr_fun}) and (\ref{eq:qfi_system_dynamics}) requires time that scales as $D^2 T^2$, whereas computing the QFI via first computing the output photon wave-packet would typically require time that scales exponentially with $T$. Appendix \ref{app:qfi_examples} demonstrates an application of Eqs.~(\ref{eq:QFI_corr_fun}) and (\ref{eq:qfi_system_dynamics}) to numerically analyze the QFI of the photons emitted from some paradigmatic multi-emitter quantum optical systems (such as the Dicke model and Tavis-Cumming model).

Finally, we remark that thus far, we have only considered sources which are coupling only to the ports of the MZI and do not lose photons into any additional channels. In many experimental settings, this may not exactly be the case --- the source might emit into multiple output channels which may correspond to different propagation directions, or even channels corresponding to decohering processes (such as non-radiative photon absorption from the source, or dephasing). When the source couples to additional loss channels, its dynamics is described instead by the Lindbladian
\begin{align}\label{eq:source_lindblad_with_losses}
    \mathcal{L}_S(t) = -i[H_S(t), \cdot] + \sum_{X \in \{A, B\}} \mathcal{D}_{L_X} + \sum_{\alpha} \mathcal{D}_{N_\alpha},
\end{align}
where the jump operators $N_1, N_2 \dots$ model photon emission into additional channels. When the photons emitted into these channels are not discarded (for e.g., if the different channels correspond to photon emission in different propagation directions and these photons are optically collected), then the expression for the QFI in Eqs.~\eqref{eq:QFI_corr_fun} and \eqref{eq:qfi_system_dynamics} continue to hold. However, when the photons in the additional channels are discarded, the result given in Eqs.~(\ref{eq:QFI_corr_fun}) and (\ref{eq:qfi_system_dynamics}) no longer accurately captures the QFI of the MZI output, since it is derived under the assumption that the photons emitted are in a pure-state. In the presence of the loss channels, however, the photons in the output ports would be in a mixed state for which computing the QFI is a significantly harder task. However, Eqs.~(\ref{eq:QFI_corr_fun}) and (\ref{eq:qfi_system_dynamics}) still provide the QFI for a purification of the mixed state of the photons in the output ports, and consequently still sets an upper bound on the QFI of the photons in the output ports \cite{braunstein1994statistical}.

{We would also like to highlight that while we have analyzed the interferometric setup, Refs.~\cite{yang2023efficient, gammelmark2014fischer} used a similar approach to compute the QFI for a spectroscopic setup where the goal is to optimally measure a parameter of the source Hamiltonian by performing a measurement on the emitted photons. There, the authors used a matrix product state representation of the emitted photons to relate the QFI to the internal dynamics of the source, which is formally similar to the approach that we take in this subsection.}
\subsection{Source spectrum and QFI} \label{sec:source_spectrum_qfi}
\indent In this section, we show that the result of the previous section can also be used to understand the constraints put by the level structure and the spectrum of the source on the possible quantum advantage. Throughout this section, we will restrict ourselves to the case where both the arms of the MZI have independent and indentical sources. In this case, the system Hamiltonian has the form $H_S(t) = H(t) \otimes I + I \otimes H(t)$, the jump operators are of the form $L_A = L\otimes I$, $L_B = I \otimes L$ and the initial and final source states are of the form $\ket{\phi_{S, i}} = \ket{\phi_{i}} \otimes \ket{\phi_i}$, $\ket{\phi_{S, f}} = \ket{\phi_f}\otimes \ket{\phi_f}$. We will denote by $\mathcal{E}(t, s) = \mathcal{T}\textnormal{exp}(\int_s^t \mathcal{L}(\tau) d\tau)$, where $\mathcal{L}(\tau) = -i[H(\tau), \cdot] + \mathcal{D}_L$, to be the channel describing the dynamics of the individual sources, and by $\rho_i(t) = \mathcal{E}(t, 0)(\ket{\phi_i}\!\bra{\phi_i}), \mathcal{P}_f(t) = \mathcal{E}^\dagger(T, t)(\ket{\phi_f}\!\bra{\phi_f})$. The expression for the QFI in Eqs.~\eqref{eq:QFI_corr_fun} and~\eqref{eq:qfi_system_dynamics} can then be reduced to
\begin{subequations}\label{eq:QFI_identical_sources}
\begin{align}
    \text{QFI} = 8\bigg(\text{Q}^{(2)} + \frac{1}{\mathcal{N}^2}\int_0^T n(t) dt\bigg),
\end{align}
where $n(t) = \text{Tr}(\rho_f(T, t)L\rho(t) L^\dagger)$, $\mathcal{N}^2 = \bra{\phi_f} \rho(T)\ket{\phi_f}$ and $\text{Q}^{(2)}$ is the contribution to the QFI from the two-point correlation functions:
\begin{align} \label{eq:Q_2}
    \text{Q}^{(2)} &= 2 \int_0^T \int_0^t \bigg( \smallabs{C^{(g)}( t, s)}^2 -  \smallabs{C^{(\chi)}(t, s)}^2\bigg)dt ds.
\end{align}
with
\begin{align}
    C^{(g)}( t, s) &= \bra{\psi} a_t^\dagger a_s \ket{\psi}=\frac{1}{\mathcal{N}^2}\text{Tr}( L^\dagger \mathcal{P}_f(t) \mathcal{E}(t,s)(L\rho_i(s) ), \nonumber\\
    C^{(\chi)}(t, s) &= \bra{\psi} a_t a_s \ket{\psi} = \frac{1}{\mathcal{N}^2}\text{Tr}(\mathcal{P}_f(t) L \mathcal{E}(t, s)(L\rho_i(s)),
\end{align}
\end{subequations}
where $\ket{\psi}$ is the state of the photons emitted into any one of the ports.
$C^{(g)}(t, s)$ and $C^{(\chi)}(t, s)$ are the two-point correlation functions contributing to the QFI. 

We will assume each source to have $D$ levels and driven for time $T$ during which it emits photons into the output port. The number of photons emitted by this source will typically grow as $\sim T$: Even when the source has only a few levels (i.e.,~$D$ is small), it can still be used to emit a large number of photons in a possibly entangled state. It can then be asked what kind of a source is required to emit photons in a state which has a QFI that scales as $\sim (\text{Number of photons})^2 \sim T^2$, and under what conditions would the QFI scale only as $\sim T$, thus forbidding a scaling quantum advantage with $T$.
\begin{figure}
    \centering
    \includegraphics[width=1.0\linewidth]{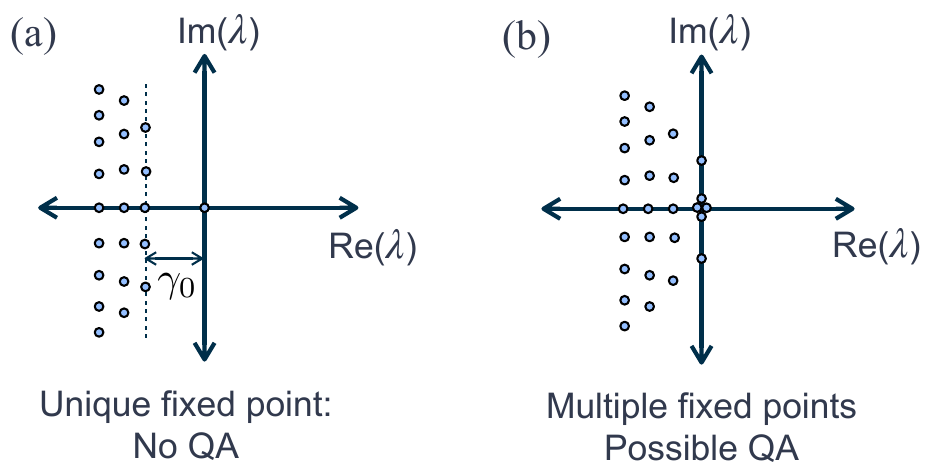}
    \caption{Schematic depiction of eigenvalues \( \lambda \) of a time-independent Lindbladian \( \mathcal{L} \) for the two cases considered in our analysis. (a) If \( \mathcal{L} \) has a unique fixed point $\uptau$ (i.e., \( \mathcal{L}(\rho) = 0 \implies \rho = \uptau\)) and all other eigenvalues have a strictly negative real part, QFI scales at most linearly with time \( T \), preventing any quantum advantage. (b) If the Lindbladian has multiple fixed points, then a quadratic quantum advantage may be possible with respect to the emission time $T$.}
    \label{fig:spectrum}
\end{figure}

While analyzing the scaling of QFI as a function of $T$, we will make the assumption that the normalization constant $\mathcal{N}$, which in general will depend on the total time $T$, is asymptotically lower bounded by a $T-$independent constant $p_0 > 0$, i.e., $\mathcal{N}^2 \geq p_0 > 0$ as $T \to \infty$. Physically, this corresponds to the assumption that the projection of the source onto the final state $\ket{\phi_f}$ at time $T$ succeeds with a probability at least $p_0$ which does not vanish as $T \to \infty$. With this assumption, we note that a possible quadratic scaling in QFI 
[i.e.,~Eq.~(\ref{eq:QFI_identical_sources})] can only be due to a quadratic scaling in $\text{Q}^{(2)}$ since
\begin{align}
\text{QFI} &= 8\bigg(\text{Q}^{(2)} + \frac{1}{\mathcal{N}^2}\int_0^T \text{Tr}(\mathcal{P}_f(t)L\rho(t) L^\dagger) dt\bigg) \nonumber\\
&\leq 8\bigg(\text{Q}^{(2)} + \frac{T\norm{L}^2}{p_0}\bigg).
\end{align}
Therefore, to assess quantum advantage, we aim to understand how $\text{Q}^{(2)}$ scales with $T$ in the large $T$ limit.

\subsubsection{Time-independent sources}
For a time-independent source (i.e.,~$\mathcal{L}(t) = \mathcal{L}$ independent of $t$), whether $\text{Q}^{(2)}$ scales as $T^2$ or only as $T$ can be related to the spectrum of the source Lindbladian $\mathcal{L}$. Consider first the case where $\mathcal{L}$ has a unique fixed point $\uptau$ (i.e., $\mathcal{L}(\uptau) = 0$) and its decay rates have a gap $\gamma_0$ (i.e., all other eigen-values of $\mathcal{L}$ have a negative real part $\leq -\gamma_0$ [Fig.~\ref{fig:spectrum}(a)]). In this case, using the Jordan decomposition of $\mathcal{L}$, we can decompose $\mathcal{E}(t, s)$ as
\begin{subequations}\label{eq:lindbladian_decomp}
\begin{align}
\mathcal{E}(t, s) = \text{Tr}(\cdot) \uptau + \Delta(t - s),
\end{align}
with
\begin{align}
    \norm{\Delta(\tau)}_\diamond \leq D_0(\tau) e^{-\gamma_0 \tau},
\end{align}
\end{subequations}
where $\smallnorm{\cdot}_\diamond$ is the superoperator diamond norm, $\gamma_0 > 0$ is a constant that depends on the spectrum of $\mathcal{L}$ and $D_0(\tau) \leq O(\tau^{m_0})$ for some $m_0$ as $\tau \to \infty$. Returning to Eq.~(\ref{eq:QFI_identical_sources}), Eq.~(\ref{eq:lindbladian_decomp}) implies that the correlation functions $C^{(g)}(t, s), C^{(\chi)}(t, s)$ approximately factorize into products of a function of $t$ and another function of $s$ when $\abs{t- s}$ is large:
\begin{align}\label{eq:connected_corr_time_ind}
    &\smallabs{C^{(g)}(t, s) - \alpha^*(t) \beta(s)} \leq  \frac{\norm{L}^2}{p_0} D_0(\abs{t - s}) e^{-\gamma_0 \abs{t - s}}, \nonumber\\
    &\smallabs{C^{(\chi)}(t, s) - \alpha(t) \beta(s)} \leq \frac{\norm{L}^2}{p_0} D_0(\abs{t - s}) e^{-\gamma_0 \abs{t - s}},
\end{align}
where $\alpha(t) = \textnormal{Tr}(\mathcal{P}_f(t) L \uptau)$ and $\beta(s) = \textnormal{Tr}(L\rho_i(s))$. This factorization of the two-point correlation function suggests that the photon emission at a time $t$ becomes uncorrelated with a previous photon emission at a time $s \ll t$. Using Eq.~(\ref{eq:connected_corr_time_ind}) together with  $\smallabs{C^{(g)}(t, s)}, \smallabs{C^{(\chi)}(t, s)} \leq \norm{L}^2$ and $\smallabs{\alpha(t)}, \smallabs{\beta(t)} \leq \norm{L}$, we obtain that
\begin{align}
&\text{Q}^{(2)} \leq \frac{4\norm{L}^4 T}{p_0^2}\int_0^T \big(2D_0(\tau) + D_0^2(\tau)\big)e^{-\gamma_0 \tau} d\tau \leq O(T).
\end{align}
For such sources, it thus follows that QFI $\leq O(T)$ as $T \to\infty$, thus forbidding a quadratic quantum advantage in interferometry. This can be seen as a consequence of such sources only emitting photons with short-time temporal correlations [as shown in Eq.~(\ref{eq:connected_corr_time_ind})], while obtaining a quadratic quantum advantage in an interferometry task requires long-time temporal correlations. We remark that this conclusion does not contradict the Heisenberg limited scalings obtained in Ref.~\cite{paulisch2019quantum} for the photons emitted in by the Dicke model, even though the Dicke model has a unqiue fixed point within the subspace of permutationally invariant states. This is because our analysis applies to sources with a fixed number of levels, with the number of emitted photons being increased by increasing $T$. On the other hand,  Ref.~\cite{paulisch2019quantum} increased the number of emitted photons by increasing the number of emitters in the Dicke model, which corresponds to increasing the number of levels in the source.

On the other hand, when the source Lindbladian does not have a unique fixed point [Fig.~\ref{fig:spectrum}(b)], then there is the possibility of quantum advantage. The simplest and physically relevant example of a source that can produce a state with QFI $\sim T^2$ is a $\pi-$level system with two excited states $\ket{e_1}, \ket{e_2}$ and two ground states $\ket{g_1}, \ket{g_2}$, with both the transitions ($\ket{e_1}\to \ket{g_1}$ and $\ket{e_2} \to \ket{g_2}$) emitting collectively through the jump operator $L = \ket{g_1}\!\bra{e_1} + \ket{g_2} \!\bra{e_2}$. Applying a Hamiltonian
\begin{align}
    H = \Omega_0 (\sigma_1 + \text{h.c.}) + \Omega_0 (\sigma_2 e^{i \alpha} + \text{h.c.}),
\end{align}
where $\sigma_i = \ket{g_i}\!\bra{e_i}$ and taking the initial and final states for the source to be $(\ket{g_1} + \ket{g_2})/\sqrt{2}$, it can be shown that when $\alpha \neq 0$, this source exhibits $\text{Q}^{(2)} \sim T^2 \sin^2 \alpha $ as $T \to \infty$. This can physically be understood as follows: for the $\pi-$level source, if initialized in and projected on the state $(\ket{g_1} + \ket{g_2})/\sqrt{2}$, the emitted photonic state $\ket{\psi}$ is a macroscopic superposition (i.e.,~``cat" like state) of two photonic states $\ket{\psi_1}, \ket{\psi_2}$ which would have been emitted by a two-level system when driven by a laser $\Omega_0$ and $\Omega_0 e^{i\alpha}$ respectively. Since $\alpha \neq 0$, $\ket{\psi_1}, \ket{\psi_2}$ become asymptotically orthogonal to each other as $T \to \infty$, and their macroscopic superposition inherits long-range correlations that make them useful for quantum-enhanced interferometry.

\subsubsection{Time-dependent sources}
We next consider sources where the source Hamiltonian itself is dependent on time $t$. In many experimental setups, while it is often possible to have a completely controllable source Hamiltonian (i.e.,~the source Hamiltonian can be designed as a function of time to apply any desired unitary on the source Hilbert space), it is typically hard to modulate this interaction between the source and the output port (i.e.,~the jump operator $L$). Consequently, we will assume that the jump operator $L$ is time-independent. While it might be physically expected that the ability to apply arbitrary unitaries on the source provides a huge flexibility in designing the output wave-packet, we provide evidence below that the form of the jump operator $L$ places severe restrictions on the achievable scalings of QFI with $T$.

Similar to the time-independent sources with a unique fixed point and a decay rate gap, time-dependent sources with asymptotically strictly contractive dynamics will have QFI $\sim T$. More specifically, a time-dependent source will be asymptotically strictly contractive if for all states $\rho_1, \rho_2$ with $X = \rho_1 - \rho_2$,
\begin{align}\label{eq:source_dynamics}
\norm{\mathcal{E}(t, s)(X)}_1 \leq C_0 e^{-\gamma_0 \abs{t - s}} \norm{X}_1 \ \ \forall \abs{t - s} \geq \tau_0
\end{align}
for some $C_0, \gamma_0, \tau_0 \geq 0$. Physically, this conditions implies that the dynamics of the source at long-times becomes independent of the initial state of the source. Given that the source dynamics satisfies Eq.~(\ref{eq:source_dynamics}), we show in Appendix \ref{app:strictly_contr} that, similar to the case of time-independent sources with a unique fixed point (Eq.~\ref{eq:connected_corr_time_ind}), $C^{(g)}(t, s), C^{(\chi)}(t, s)$ also factorize when $t \gg s$ and consequently the QFI $\sim T$. While it is generally a difficult task to show strict contractivity of a Lindbladian with a time-dependent Hamiltonian and a set of time-independent jump operators, it is generally expected to hold for sources with a single ground state and no possible dark states. In Appendix~\ref{app:strictly_contr}, we rigorously show that this condition holds for a two-level source decaying into the output port with $L = \sqrt{\gamma}\ket{g}\! \bra{e}$, as well as a multi-level source where the jump operators have a full Kraus rank, i.e., they form a complete basis for the space of source operators. 

We next consider sources where the dynamics is not strictly contractive and thus a quantum advantage is not forbidden. The simplest and experimentally available sources with non strictly contractive dynamics are those that emit via a transition from an excited state $\ket{e}$ to a ground state $\ket{g}$ and additionally have some number of dark states $\{\ket{m_1}, \ket{m_2} \dots \ket{m_k}\}$. The jump operator corresponding to the source is $L = \ket{g}\!\bra{e}$, which leaves the dark states unchanged, but the Hamiltonian $H(t)$ can in general couple the dark states with the excited and ground state. Examples of such sources include $\Lambda$ and V-level systems with both the transitions coupling to the same output port, as well as multi-emitter systems described by the Dicke model \cite{dicke1954coherence}. Several protocols have been devised that leverage the dark states in such sources to create long-time correlated photonic states (such as the photonic GHZ state) \cite{pichler2017universal, wei2021generation, wei2022generation, rubies2024deterministic}. As a simple example, consider a source with an excited state $\ket{e}$, ground state $\ket{g}$ and a dark state $\ket{m}$: By initializing the source in $(\ket{g} + \ket{m}) / \sqrt{2}$ and then driving the $\ket{e} - \ket{g}$ transition with $H(t) = \Omega (\ket{e}\!\bra{g} + \text{h.c.})$ and finally projecting the source on $(\ket{g} + \ket{m}) / \sqrt{2}$, we can generate a photonic state 
\begin{align}\label{eq:state_form_macro_super}
\ket{\psi} \propto \ket{\psi_\text{TLS}} + \ket{\text{vac}}, 
\end{align}
where $\ket{\psi_\text{TLS}}$ is the state of photons emitted by a driven two-level state. Since $\ket{\psi}$ is a coherent superposition between $\ket{\text{vac}}$ and a state $\ket{\psi_\text{TLS}}$ with $O(T)$ photons emitted in the time interval $[0,T]$, it is a state with long-time correlations.
\begin{figure*}
    \centering
    \includegraphics[width=1.0\linewidth]{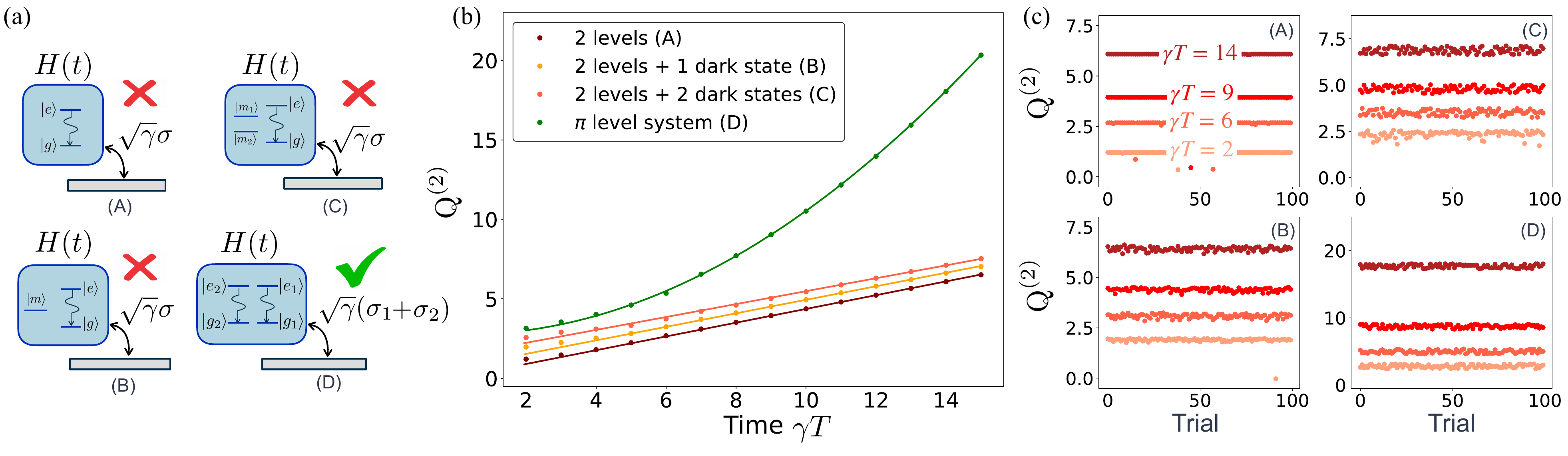}
    \caption{(a) Schematic illustration of different level structures analyzed for possible quantum advantage. (A) a two-level system, (B) a two-level system with a single dark state, (C) a two-level system with two dark states, and (D) a \( \pi \)-level system. (b) Numerically optimized \( \text{Q}^{(2)} \), with the optimization variable being the time-dependent source Hamiltonian $H(t)$, as a function of emission time $T$ for different sources. When the source has only one emission pathway ($\ket{e}\to \ket{g}$), with or without dark states, \( \text{Q}^{(2)} \) scales linearly with $T$ thus providing no quantum advantage. Since a \( \pi \)-level source has two emission pathways $\ket{e_1}\to\ket{g_1}, \ket{e_2}\to \ket{g_2}$, $\text{Q}^{(2)}$ scales quadratically with $T$ thus providing a quadratic quantum advantage at large $T$. (c) Distribution of the optimized \( \text{Q}^{(2)} \) over 100 trials with random initial conditions for \( H(t) \). The data plotted in (b) is the maximum $\text{Q}^{(2)}$ obtained over the 100 optimization trials to ensure that it is a good estimate of the ``global optimum" of the optimization problem.}
    \label{fig:qfi_optimized}
\end{figure*}

Despite the possibility of generating long-time correlations, we find that the QFI of the photons emitted from such sources $\sim T$. In Fig.~\ref{fig:qfi_optimized}, we numerically investigate the photonic state emitted by a two-level system, a two-level system with upto two dark states and a $\pi$-level system [Fig.~\ref{fig:qfi_optimized}(a)]. We use a gradient-based optimization algorithm to numerically maximize $\text{Q}^{(2)}$ of the emitted photons with respect to the time-dependent source Hamiltonian $H(t)$. To compute the gradient of $\text{Q}^{(2)}$ with respect to $H(t)$ for large $T$, we develop an adjoint-variable method \cite{giles2000introduction, givoli2021tutorial, white2022enhancing} that significantly speeds up the gradient computation (Appendix \ref{app:adj_vm}). As shown in Fig.~\ref{fig:qfi_optimized}(b), we find that the optimized $\text{Q}^{(2)} \sim T$ for the two-level system, which is expected since the source dynamics of the two-level systems are provably strictly contractive. Furthermore, we also find that for a two-level system with additional dark states, while having a $\text{Q}^{(2)}$ higher than that of a two-level system without dark states, still exhibits $\text{Q}^{(2)}\sim T$. On the other hand, for a $\pi$-level system, the optimized $\text{Q}^{(2)}$ scales as $T^2$. We remark that the optimized $\text{Q}^{(2)}$ shown in Fig.~\ref{fig:qfi_optimized}(b) are the largest $\text{Q}^{(2)}$ obtained from individual optimization trials with randomly chosen guess for the initial $H(t)$ --- Fig.~\ref{fig:qfi_optimized}(c) shows the distribution of the optimized $\text{Q}^{(2)}$ obtained from these different trials. 

The scaling of $\text{Q}^{(2)} \sim T$ for the two-level system with dark states, despite their potential for generating long-time correlated states, can be physically attributed to the correlation functions $C^{(g)}(t, s)$ and $C^{(\chi)}(t, s)$ that determine the QFI [Eq.~(\ref{eq:QFI_identical_sources})] remaining short-range correlated. Specifically, for states of the form as in Eq.~(\ref{eq:state_form_macro_super}), $C^{(g)}(t, s) \propto \bra{\psi_\text{TLS}} a_t^\dagger a_s \ket{\psi_\text{TLS}}, C^{(\chi)}(t, s) \propto \bra{\psi_\text{TLS}} a_t a_s \ket{\psi_\text{TLS}}$, i.e.,~the correlation functions $C^{(g)}(t, s)$ and $C^{(\chi)}(t, s)$ inherit the correlations of the state $\ket{\psi_\text{TLS}}$ which, being a photonic state emitted by a simple two-level system, only has short time correlations thus yielding a QFI $\propto T$. Even with more complicated protocols, we expect that the emitted photonic state that can possibly be generated is of the form of Eq.~(\ref{eq:state_form_macro_super}), since such a source can only create correlations by either emitting a photon (when it is in $\ket{e}$) or not emitting a photon (when it is in a dark state). In order to generate a long-time correlated state with QFI $\sim T^2$, we need a source that can generate photons via two distinct and distinguishable emission paths. For instance, as discussed previously in this section, using a $\pi-$level system where photons can be generated via either the transition $\ket{e_1} \to \ket{g_1}$ or $\ket{e_2}\to \ket{g_2}$, we can obtain a state $\ket{\psi}$ that is the coherent superposition of two macroscopic photonic states for which even $C^{(g)}(t, s)$ and $C^{(\chi)}(t, s)$ are long-time correlated thus yielding a QFI $\sim T^2$. Indeed, this is what we observe in Fig.~\ref{fig:qfi_optimized}(b), where a time-dependent $\pi$-level source with the source Hamiltonian $H(t)$ optimized to maximize the QFI, exhibits a quadratic improvement over the QFI obtained from source with one emitting transition and a dark state subspace.

\section{Optimal measurements}

In the previous section, we characterized the QFI that can be achieved with a Markovian light source. However, even when the QFI has Heisenberg-limited scaling, to achieve the corresponding phase-sensitivity in an actual spectroscopic setup, it is necessary to carefully design the measurement that extracts the unknown phase $\varphi$ from the state at the output of the MZI. More formally, the measurement is captured by a positive-operator valued measurement $E_a$ corresponding to the outcome $a$. When the measurement is applied on the state $\ket{\psi_\varphi}$ at the output of the MZI, it yields the outcome $x$ with probability $p_\varphi(a) = \bra{\psi_\varphi} E_a \ket{\psi_\varphi}$. The measurement is optimal if the Classical Fisher information (CFI) of $p_\varphi(a)$, given by
\[
\text{CFI} = \int  \bigg(\frac{d}{d\varphi}\ln p_\varphi(a)\bigg)^2 p_\varphi(a) da,
\]
is equal to the quantum Fisher information (QFI) of $\ket{\psi_\varphi}$. In general, the optimal measurement is not unique and there can be multiple measurements for which the corresponding CFI is equal to the QFI of $\ket{\psi_\varphi}$.

In the following, we analyze different strategies that can be used to optimally measure the phase $\varphi$ from the photonic state $\ket{\psi_\varphi}$ obtained at the output of the MZI. Our goal here is to understand what experimental resources are required and how they depend on the photonic state. We first describe a general strategy to optimally obtain $\varphi$ but which requires a high quality factor Kerr resonator. Then, we go onto understand when linear optics and photodetection are sufficient to perform an optimal measurement, deriving a set of sufficient conditions under which a simple photodetection measurement is optimal. Furthermore, in the case of independent sources emitting into the two input ports of the MZI, we show that if $\varphi$ can be measured optimally with linear optics and photodetectors, then only photodetection is also an optimal measurement.

\subsection{Optimal measurement with non-linear optics}\label{sec:pd_nlo}

\begin{figure}
	\centering
	\includegraphics[width=1\linewidth]{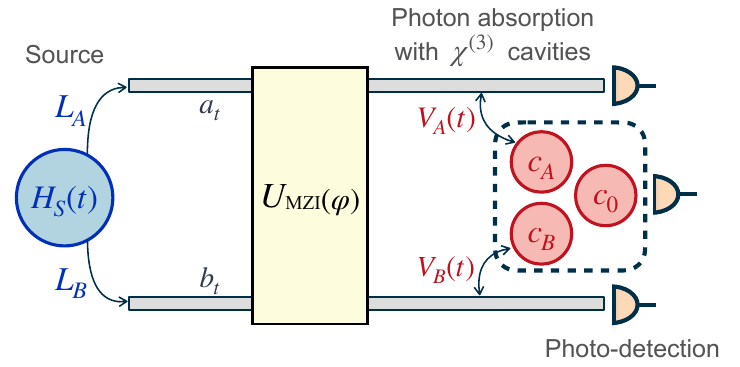}
	\caption{\label{fig:mzi_reabsorb} Schematic depiction of the optimal measurement around $\varphi = 0$ using non-linear optical elements. A   $\chi^{(3)}$ cavity with three optical modes ($c_A,c_B, c_0$), with tunable couplings $V_A(t), V_B(t)$ to the output ports of the MZI, is used to re-absorb the photons emitted by the source, and then a photodetection is performed on both the output ports as well as on the $\chi^{(3)}$ cavity.}
\end{figure}
% The ultimate QFI bound in Eq.~(\ref{eq:QFI_MZI}) can be achieved by employing time reversal, reabsorbing the photons into an auxiliary controllable system of dimension $D+2$, and performing photodetection on the output port along with a measurement on the auxiliary system as shown in Fig.~(\ref{fig:mzi_reabsorb}).\\ 
\indent We first consider the most general setting in which a time-dependent $D$-level source emits photons into the two ports of the MZI interferometer. When the source has no further symmetries, the optimal measurement essentially has to undo the source-port dynamics. To illustrate this in a simple setting, we consider the simpler setup of $n$ qubits initially in $\ket{0}^{\otimes n}$ and a unitary $U_z$, dependent on the unkown parameter $z$, is applied on them to obtain that state $\ket{\psi_z} = U_z \ket{0}^{\otimes n}$. An optimal measurement to extract the parameter $z$, assuming it to be in the neighbourhood of $z_0$, would be to implement the projective measurement given by $\{P_0, P_1, P_2 \dots \}$
\begin{align}\label{eq:projective_measurement}
P_0 = \ket{\psi_{z_0}}\!\bra{\psi_{z_0}} \text{ and }\sum_{i \geq 1} P_i = I  - \ket{\psi_{z_0}}\!\bra{\psi_{z_0}}.
\end{align}
It is easy to check that this measurement is optimal, i.e.,~the CFI of the probability distribution  over the measurement outcomes $\{0, 1, 2 \dots \}$, at $z_0$, is equal to the QFI of the state $\ket{\psi_z}$ at $z = z_0$ (Ref.~\cite{braunstein1994statistical,zhou2020saturating}, and reviewed in Appendix \ref{app:opt_measurement_nlo}). The simplest strategy to implement this measurement on $\ket{\psi_z}$ is to first undo the unitary transformation $U_{z_0}$ (i.e.,~apply the unitary $U_{z_0}^\dagger$) followed by a computational basis measurement on the $n$ qubits. Measuring zero on all qubits is then equivalent to measuring outcome $0$ in the projective measurement in Eq.~(\ref{eq:projective_measurement}). 

To apply this measurement strategy to the interferometric setup, where we want to measure the MZI phase $\varphi$ in the neighbourhood of $\varphi = 0$, we need to undo the photon emission from the source, i.e.,~coherently re-absorb the photons emitted by the source. {A similar strategy has been used to implement an optimal measurement for spectroscopy using measurements on the emitted photon fields \cite{yang2023efficient} --- our focus in this section is to show that this measurement can always be implemented with nonlinear optical elements. In particular, we show that a perfectly re-absorbing system can in principle be implemented using a multi-mode optical cavity with a $\chi^{(3)}$ non-linearity interacting with the output ports of the MZI (Fig.~\ref{fig:mzi_reabsorb}).} Specifically, we use 3 modes of this optical cavity, with corresponding annihilation operators $c_A, c_B$ and $c_0$. We assume that the mode $c_A$ is coupled to port A, $c_B$ is coupled to port $B$ and all three modes couple to each other via the nonlinear self- and cross-phase modulation induced by the $\chi^{(3)}$ non-linearity. In addition, the three modes will be driven by a tunable time-dependent coherent laser. The Hamiltonian of this reabsorbing system is then given by
\begin{align}
H_R(t) &= \sum_{k \in \{0, A, B\}} \big(\Omega_k(t) c_k^\dagger + \text{h.c.}\big) \nonumber\\
&+\frac{\chi}{2}\sum_{k, k' \in \{0, A, B\}}c_k^\dagger c_{k'}^\dagger c_k c_k',
\end{align}
with $\Omega_k(t)$ being the coherent field applied on the $k^\text{th}$ mode. By building upon Ref.~\cite{yuan2023universal}, we show in Appendix \ref{app:opt_measurement_nlo} that any unitary on the joint Hilbert space of these cavity modes can be applied by designing the laser field $\Omega_k(t)$. As we show, the speed at which unitaries can be applied is \emph{not} limited by the non-linear strength $\chi$, but only by the magnitude and the rate of change of the coherent drives $\Omega_k(t)$. Thus, gates can applied very fast on the Hilbert space of the cavity modes. We also assume that the modes $c_A$ and $c_B$ couple linearly with the output ports of the MZI via
\begin{align}
    H_{R, P} = V_A(t) a_t^\dagger c_A + V_B(t) b_t^\dagger c_B + \text{h.c.},
\end{align}
which gives us the ability to sequentially transfer photons from the output ports to the cavity modes during the reabsorption process. We show in Appendix \ref{app:opt_measurement_nlo} that, by appropriately tuning the parameters $V_A(t), V_B(t), \Omega_A(t), \Omega_B(t), \Omega_0(t)$ as a function of time $t$, we can coherently absorb the photons emitted from the source into the non-linear cavity. This effectively inverts the emission of photons from the source, and a subsequent photodetection on both the optical cavity as well as the output ports allows us to implement the target optimal measurement.

To implement the coherent photon reabsorption required to apply this measurement strategy, one approach is to time-reverse the source. However, this may not always be possible.  \emph{First}, we may not have access to a replica of the physical system used as the source and thus would like to be able to design the re-absorbing system with components that are independent of the source. \emph{Second}, in several sources exhibiting quantum advantage, the source might not disentangle from the emitted photons on its own and we might have to project the source onto a final target state (e.g.,~see the example of the driven $\pi$-level system discussed in section \ref{sec:source_spectrum_qfi}).

We instead design a perfectly reabsorbing system using a completely controllable multi-mode non-linear optical cavity. The key idea that we use to implement this photon absorption is to note that, when expressed in the time-bin basis, the photonic state of the port is approximately a matrix product state (MPS) with a bond-dimension equal to the dimensionality of the Hilbert space of the source \cite{fischer2018particle}. Consequently, by using the canonical form of this MPS, we can compute a sequence of unitaries applied on the qudits in this MPS and an ancilla to map it to the vacuum state \cite{schon2005sequential, schon2007sequential, cirac2021matrix}. To physically implement the reabsorption, we treat the cavity mode $c_0$ as the ancilla and perform two operations when each time-bin in the output port is incident on the reabsorbing $\chi^{(3)}$ cavity --- we first swap the photons in this time-bin from the ports $A$ and $B$ into the cavities $c_A$ and $c_B$ respectively. This operation can be performed by only tuning $V_A(t), V_B(t)$. Next, since control over the parameters $\Omega_A(t), \Omega_B(t), \Omega_0(t)$ allows us to implement any unitary on the joint Hilbert space the three cavity modes, which now contains both the qudit of the MPS describing the photonic state in the cavities $c_A, c_B$ and the ancilla in $c_0$, we design them to implement the gates in the previously computed unitary circuit. The details of this protocol are explicitly laid out in Appendix \ref{app:opt_measurement_nlo}.

We remark that this protocol can be considered to be a generalization of well known state-transfer protocols \cite{cirac1997quantum}, which considered the problem of coherently transferring the quantum state of a source, that couples to an output photon field and emits a single photon entangled with the source, to another quantum system also coupling to the photon field. Finally, we remark that this protocol guarantees that there is always a choice of the parameters of the $\chi^{(3)}$ cavity so as to reabsorb the emitted photons, but might not be the most practical implementation of the reabsorption process. In practice, depending on the constraints of the specific experimental system at hand, heuristic gradient-based control design algorithms such as GRAPE \cite{khaneja2005optimal} can also be used to find the parameters of the $\chi^{(3)}$ cavity that accomplish this reabsorption.

\begin{figure*}
	\centering
	\includegraphics[width=0.95\linewidth]{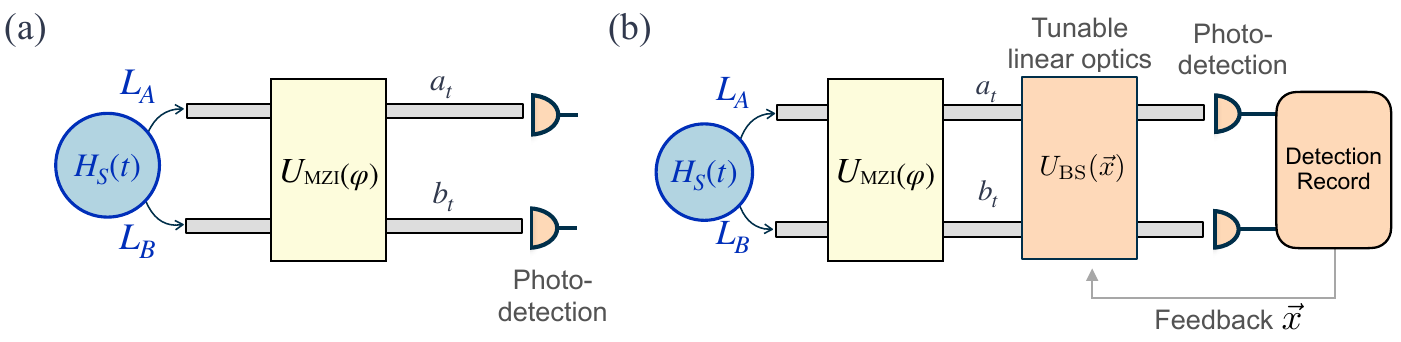}
	\caption{\label{fig:beam_splitter} Schematic depiction of two measurements implementable with linear optical elements and photodetectors studied in Section \ref{sec:pdlo}. (a) A setup where photodetection is performed at the MZI output, and (b) a setup where a  linear-optical element, tuned based on the photodetection record, is applied before the photodetectors.}
\end{figure*}

\subsection{Optimal measurement with photodetection and linear optics}\label{sec:pdlo}
Although the previous subsection described a general strategy to implement an optimal measurement irrespective of the source, having a completely controllable nonlinear optical cavity is still experimentally challenging. In this subsection, we study under what conditions photodetectors and linear-optical elements alone are enough to implement an optimal measurement.

The simplest measurement that can be performed is a measurement of the total number of photons in the output ports. We first provide a general sufficient condition under which photon-number measurement at the two ports (Fig.~\ref{fig:beam_splitter}) is optimal. Suppose $\Pi_{n_A, n_B}$ is the projector on the subspace where there $n_A$ photons in the port $A$ and $n_B$ photons in the port $B$. Let $\mathcal{I} = \{(n_A, n_B): n_A, n_B \in \{0, 1, 2 \dots\}\}$ be a set of photon numbers in the two ports such that for all $(n_A, n_B) \in \mathcal{I}$, $(n_A \pm 1, n_B \pm 1) \notin \mathcal{I}$. Then, if the photonic state incident on the MZI $\ket{\psi}$ satisfies 
\begin{align}
    \Pi_{\mathcal{I}} \ket{\psi} = \ket{\psi},
\end{align}
where $\Pi_{\mathcal{I}} = \sum_{(n_A, n_B) \in \mathcal{I}}\Pi_{n_A, n_B}$ is a projector, then photon-number measurement at the output of the MZI is the optimal measurement in the neighbourhood of $\varphi = 0$. Physically, this condition implies that addition or subtraction of a single photon in both the output ports to $\ket{\psi}$ makes the resulting state orthogonal to the subspace with projector $\Pi_{\mathcal{I}}$ which contains the state $\ket{\psi}$. Furthermore, since this subspace can be identified by just the number of photons in the two output ports, photon number measurement becomes an optimal measurement for sensing the phase $\varphi$.

Examples of states that satisfy this condition are
\begin{enumerate}
    \item[(a)] $\ket{\psi} = \smallket{\psi_A^{(N_A)}}\otimes \smallket{\psi_B^{(N_B)}}$ where $\smallket{\psi_A^{(N_A)}}$ has $N_A$ photons and $\smallket{\psi_B^{(N_B)}}$ has $N_B$ photons. Here $\mathcal{I} = \{(N_A, N_B)\}$ and it is clear that adding or removing a photon from any of the output ports results in state orthogonal to $\Pi_{\mathcal{I}}$. For such states, it is already known that just photon-number measurement is optimal even if the state is not a Fock state \cite{paulisch2019quantum}.
    \item[(b)] $\ket{\psi} = \smallket{\psi_A^{(N_A)}}\otimes \ket{\psi_B}$ where $\smallket{\psi_A^{(N_A)}}$ has $N_A$ photons while $\ket{\psi_B}$ can be arbitrary. In this case, $\mathcal{I} = \{(N_A, n_B) : n_B \in \{0, 1, 2 \dots \}\}$ and therefore adding or removing a photon from port $A$ also results in a state orthogonal to $\Pi_{\mathcal{I}}$. Consequently, photon number measurement at the output ports remains optimal if only one of the ports has a state with a definite number of photons.
    \item[(c)] $\ket{\psi} = \smallket{\psi_A^{(k_A, r_A)}}\otimes \ket{\psi_B}$ where $\smallket{\psi_A^{(N_A)}}$ has either $r_A, k_A + r_A, 2k_A + r_A, 3k_A + r_A \dots$ photons and $\ket{\psi_B}$ can be arbitrary. When $k_A = 2$, the state $\ket{\psi_A}$ is a state with a definite photon number parity (i.e.,~odd parity if $r_A = 1$ and even parity if $r_A = 0$). Here, again photon number measurement at the output ports is optimal when $k_A \geq 2$ since $\mathcal{I} = \{(r_A + m_A k_A, n_B) : m_A, n_B \in \{0, 1, 2 \dots\}\}$ and adding or removing a photon from port $A$ results in an orthogonal state.
\end{enumerate}

To show the optimality of the photon-number measurement, we consider the projective measurement with outcomes $0, 1$ given by $P_0 = \Pi_\mathcal{I}, P_1 = I - \Pi_\mathcal{I}$. This measurement can clearly be performed with just photodetectors by first measuring the photon numbers $(n_A, n_B)$ at the output port and then recording $0$ if $(n_A, n_B) \in \mathcal{I}$ else recording 1. The measurement outcome follows the probability distribution $p_\varphi(0) = \bra{\psi_\varphi} \Pi_{\mathcal{I}} \ket{\psi_\varphi}, p_\varphi(1) = 1 - p_\varphi(0)$. We first note that $p_{\varphi = 0}(0) = 1, p_{\varphi = 0}(1) = 0$ and from Eq.~(\ref{eq:deriv_psi_phi}),
\begin{align*}
    \frac{d}{d\varphi} p_{\varphi = 0}(0)  = -\frac{d}{d\varphi} p_{\varphi= 0}(1)  = 2\ \text{Im}\big(\bra{\psi}\Pi_{\mathcal{I}} H_d\ket{\psi}\big) = 0,
\end{align*}
where we remind the reader that $H_d = i\int_0^T (a_t^\dagger b_t - a_t b_t^\dagger)dt$. Furthermore, we have used $\Pi_{\mathcal{I}} H_d \ket{\psi} = 0$ since $H_d \ket{\psi}$ is a state in which one photon has either been added or subtracted from both the modes and hence, by assumption, is in the subspace orthogonal to $\Pi_{\mathcal{I}}$. Consequently, the CFI of the probability distribution $p_\varphi$ at $\varphi = 0$ is given by
\begin{align}
    \text{CFI} &= \lim_{\varphi \to 0} \frac{1}{p_\varphi(1)}\bigg(\frac{d}{d\varphi}p_\varphi(1)\bigg)^2 =2\frac{d^2}{d\varphi^2} p_{\varphi = 0}(1) \nonumber \\
    &=4\bra{\psi}H_d (I - \Pi_{\mathcal{I}}) H_d \ket{\psi} = 4 \bra{\psi} H_d^2 \ket{\psi},
\end{align}
where we have again used $\Pi_{\mathcal{I}}H_d\ket{\psi} = 0$. Finally, noting that $\bra{\psi} H_d \ket{\psi} = \bra{\psi} \Pi_{\mathcal{I}} H_d\ket{\psi} = 0$, it follows from Eq.~(\ref{eq:QFI_corr_fun}) that CFI = QFI, thus establishing that photon-number measurement is the optimal measurement.

{Next, we analyze if adding linear-optical elements before the photo detectors can allow us to construct optimal measurements for a larger class of quantum photonic states [Fig.~\ref{fig:beam_splitter}(b)]. To keep the measurement protocol as general as possible, we allow the linear optical elements before the photodetector to be modulated as a function of time depending on the result of the photodetection. More specifically, the output $\vec{x}$ of the photodetectors is a set of times $\tau_1, \tau_2, \tau_3\dots $ at which a photon has been detected as well indices $\sigma_1, \sigma_2, \sigma_3 \dots \in\{a, b\}$ where $\sigma_i$ indicates the port in which the photon at time $\tau_i$ is detected. As depicted in Fig.~\ref{fig:beam_splitter}(b), we allow the beam splitter $U_\text{BS}(\vec{x})$, that is applied before the photodetection, to change depending on the photodetector output $\vec{x}$ recorded so far. 

We first analyze this measurement protocol for the case of independent sources emitting into the two input ports of the MZI, i.e.,~the input photonic state $\ket{\psi} = \ket{\psi_A}\otimes \ket{\psi_B}$, where $\ket{\psi_A}$ and $\ket{\psi_B}$ are the photonic states emitted into the ports A and B, respectively. We make an additional assumption that the wave-functions corresponding to the states $\ket{\psi_A}, \ket{\psi_B}$ are non-zero, i.e.,
\[
\bra{\text{vac}} \prod_{i = 1}^n x_{t_i} \ket{\psi_X} \neq 0,
\]
for all $x \in \{a, b\}, \ n \in \{0, 1, 2\dots \},\ t_1, t_2 \dots t_n \geq 0$. For sources satisfying this assumption, we show in Appendix \ref{app:opt_measurement_pdlo} that optimality of a measurement protocol with a tunable beam-splitter and photodetector also implies optimality of a measurement protocol with only photodetectors without using any linear-optical elements. Stated differently, when the sources emitting photons in the two input ports of the MZI are independent, if photodetection is a sub-optimal measurement then adding linear-optics to it alone cannot make it optimal. On the other hand, when the photons emitted in the two input ports of the MZI are entangled, then it is possible that photodetection alone is sub-optimal as a measurement, but becomes optimal when supplemented with linear optical elements. We provide an explicit example of a photonic state illustrating this fact in Appendix \ref{app:opt_measurement_pdlo}.}

\section{Conclusion}
\indent In this paper, we provide a general framework for evaluating the quantum advantage of light sources for quantum interferometry. Within the Markov approximation, we first show how to compute the Quantum Fisher Information (QFI) of the emitted photons in a Mach-Zehnder Interferometer (MZI). In particular, we show that to compute QFI, it is enough to be able to simulate the internal dynamics of the source and we do not need to calculate the entire photon wave-packet. We then use this result to analytically and numerically elucidate the level structure and spectral properties of the photon source needed to obtain a possible quantum advantage in interferometry. Finally, we turn to the question of how to implement optimal measurements for photons emitted from general light sources and with experimentally available optical elements. We show that a controllable system of coupled $\chi^{(3)}$ cavities together with a photodetector can always be used to implement the optimal measurement irrespective of the light source. We additionally elucidate conditions when photodetectors and linear optics alone are enough to implement the optimal measurement. 

An immediate next step is to apply the framework developed in our work to numerically and analytically study the potential of experimentally realistic quantum photonic system to generate metrologically useful photonic states. It would be particularly interesting to understand the impact of experimental non-idealities, such as position and spectral inhomogeneities in the quantum emitters as well as losses, on the possible quantum advantage in the emitted photons. On the more theoretical side, making rigorous the limitations to quantum advantage imposed by the level structure of the quantum emitters driven by a possible time-dependent Hamiltonian is also an open question --- while we have partly addressed this question through analytical calculations and numerical simulations, a fully rigorous treatment of this question remains open. Finally, extending our framework to cases where the light source is not Markovian (e.g., has a time-delay and feedback \cite{pichler2017universal, grimsmo2015time, calajo2019exciting, trivedi2021optimal}) could also allow us to study if and how non-Markovianity can be used as a resource to generate quantum advantage in interferometry.

\section{Acknowledgements}
R.T.~acknowledges support from Center for Integration of Modern Optoelectronic Materials on Demand (IMOD) seed grant (DMR-2019444). R.T.~also acknowledges support from QuPIDC, an Energy Frontier Research Center, funded by the US Department of Energy (DOE), Office of Science, Basic Energy Sciences (BES), under the award number DE-SC0025620.  A.A.G. acknowledges support by the National Science Foundation through the CAREER Award (No. 2047380), the Air Force Office of Scientific Research through their Young Investigator Prize (grant No. 21RT0751), as well as by the David and Lucile Packard Foundation. DM acknowledges support from Novo Nordisk Fonden under grant numbers NNF22OC0071934 and NNF20OC0059939. This research was also gusupported in part by grant NSF PHY-2309135 to the Kavli Institute for Theoretical Physics (KITP).

\bibliography{ref_biblo}
\ \ 
\newpage

\appendix
\onecolumngrid
\section{Quantum Fisher Information of the light source}
In this section, we provide details on the measurement protocol for QFI of a light source using its correlation functions and Eq.~(\ref{eq:QFI_corr_fun}), the derivation of Eq.~(\ref{eq:qfi_system_dynamics}), which expresses the QFI in terms of the internal dynamics of the system, as well as an application of Eq.~(\ref{eq:QFI_corr_fun}) to study the QFI of some paradigmatic quantum optical systems.

\subsection{Measurement of the QFI from correlation functions}\label{app:qfi_measurement}

To measure the QFI of a light source, it is enough to measure the correlation functions in Eq.~(\ref{eq:QFI_corr_fun}). In this section, {we show that this can be done without a full interferometric measurement of the unknown phase,} for both the cases where there is a single source simultaneously emitting into the two input ports of the MZI, as well as when independent and identical sources emit into these input ports and we only have access to one of these sources.
\begin{figure}[b]
	\centering
	\includegraphics[width=1.0\linewidth]{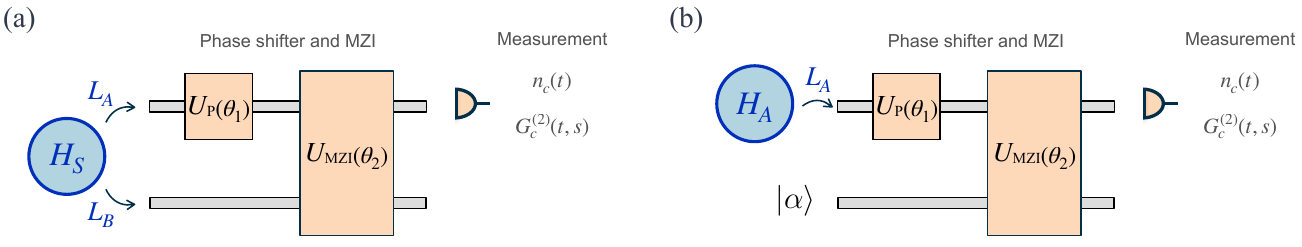}
	\caption{\label{fig:correlation_measurement} (a) Schematic representation of the measurement setup for extracting the 2-point and 4-point correlation functions. We measure the photon count, $n_c(t)$, and the second-order correlation function, $G^{(2)}_c(t,s)$, after applying a phase shifter and an MZI unitary for different known values of $\theta_1$, and $\theta_2$. This measurement enables the extraction of the correlation functions required for computing QFI. (b) If the source emits into a single port, we can use the same measurement, but with a coherent state as the input to the second port. }
\end{figure}

\subsubsection{Single source emitting into both ports} First, we apply a phase shifter $U_P(\theta_1)$ to one of the output arms of the source (here $a_t$), and then pass it through an MZI with the known phase $\theta_2$ (Fig.~\ref{fig:correlation_measurement}). This transforms the annihilation operator $a_t$ to
\begin{equation}
    c_t = U_{\text{MZI}}^\dag(\theta_2) U_P^\dag(\theta_1) a_t U_P(\theta_1)U_{\text{MZI}}(\theta_2) = a_t e^{i\theta_1} \cos(\theta_2) + b_t \sin(\theta_2).
\end{equation}
If we measure the photon flux at the output we get,
\begin{equation}
    n_c(t) = \langle c_t^\dag c_t \rangle = \cos^2(\theta_2) \langle a_t^\dag a_t \rangle + \sin^2(\theta_2) \langle b_t^\dag b_t \rangle + \sin(2\theta_2) \Re \big(\langle a_t^\dag b_t\rangle e^{-i\theta_1} \big).
\end{equation}
Thus, if we choose $\theta_1 = \pi/2$ and measure $n_c(t)$ as a function of $\theta_2$, we can extract the 1-point correlations $\mathcal{C}^{(1)}_{a;a}(t)$, $\mathcal{C}^{(1)}_{b;b}(t)$, and $\mathcal{C}^{(1)}_{a;b}(t)$. Next, we perform a two-photon correlation measurement at the output port to obtain $G^{(2)}_c(t, s) = \langle c_t^\dagger c_s^\dagger c_t c_s\rangle$. This yields
\begin{subequations}\label{eq:g2_measure}
\begin{align}
    G^{(2)}_c(t,s) = \cos^4(\theta_2) g_{40} + \sin^4(\theta_2)g_{04} +\cos^2(\theta_2)\sin^2(\theta_2) g_{22} +
    2\cos^3(\theta_2)\sin(\theta_2) g_{31} +  2\cos(\theta_2)\sin^3(\theta_2) g_{13},
\end{align}
where
\begin{align}
    g_{40} &= \langle a^\dag_t a^\dag_s a_t a_s \rangle , \\
    g_{31} &= \Re\big(\langle a^\dag_t a^\dag_t a_s b_s^\dag\rangle e^{i\theta_1(s)}\big) + \Re\big(\langle a_t b^\dag_t a_s^\dag a_s\rangle e^{i\theta_1(t)}\big),\\ \label{eq:g_22}
    g_{22} & = \langle a^\dag_t b^\dag_s a_t b_s \rangle+ \langle b^\dag_t a^\dag_s b_t a_s \rangle + 2\Re\big(\langle a^\dag_t a^\dag_s b_t b_s \rangle  e^{-i(\theta_1(t)+\theta_1(s))}\big)+ 2\Re\big(\langle a^\dag_t b^\dag_s b_t a_s \rangle  e^{-i(\theta_1(t)-\theta_1(s))}\big), \\
    g_{13} &= \Re\big(\langle a^\dag_t b_t b_s^\dag b_s\rangle e^{-i\theta_1(t)} \big) + \Re\big(\langle b_t^\dag a^\dag_s b_t b_s\rangle e^{-i\theta_1(s)}\big), \text{ and}\\ 
    g_{04} &= \langle b^\dag_t b^\dag_s b_t b_s \rangle,
\end{align}
\end{subequations}
where we assume that the angle of the phase shifter $\theta_1$ can be time-dependent.
Choosing a linear function for $\theta_1(t)$, i.e. $\theta_1(t) = \theta^*_1 t$ and measuring the value of $G^{(2)}_c(t,s)$ as a function of $t$, $s$, $\theta^*_1$, and $\theta_2$, we can find the 2-point correlations $\mathcal{C}^{(2)}_{a, b; b, a}(t, s) = \Re\big(\langle a_{t}^\dagger b^\dagger_{s} b_{t}a_{s}\rangle\big)$, $\mathcal{C}^{(2)}_{b, b; a, a}(t, s) = \Re\big(\langle a^\dagger_{t}a^\dagger_{s} b_{t}b_{s}\big)\rangle$. 

\subsubsection{Single source emitting into one port} If the source emits into one port, we can use the same measurement setup, except that we introduce a coherent state $\ket\alpha$ into the second port. Then Eq.~(\ref{eq:g_22}) becomes
\begin{align}
    g_{22} = |\alpha|^2\langle a^\dag_t a_t \rangle+ |\alpha|^2\langle a^\dag_s a_s \rangle + 2\Re\big(\langle a^\dag_t a^\dag_s \rangle  \alpha^2 e^{-i(\theta_1(t)+\theta_1(s))}\big)+ 2\Re\big(\langle a^\dag_t a_s \rangle  |\alpha|^2 e^{-i(\theta_1(t)-\theta_1(s))}\big) \big).
\end{align} 
We again choose $\theta_1(t) = \theta_1^* t$ and measure \( G_c^{(2)}(t, s) \) for different values of \( \theta_2 \) and \( \theta_1^* \). Thus, we can determine \( \langle a_t a_s \rangle \) and \( \langle a_t^\dag a_s \rangle \). Alternatively, we can fix \( \theta_2 \) and vary \( |\alpha|^2 \) instead.
\subsection{Quantum Fisher Information in terms of the light source dynamics} \label{app:QFI_comp}
In this subsection, we show how to compute the functions $\mathcal{C}_{x^1 \dots x^n; y^1\dots y^n}^{(n)}(t_1 \dots t_n)$ by only tracking the internal dynamics of the source. A result that we will use for this analysis is the Quantum regression theorem \cite{carmichael2013statistical}, which we review below.

\begin{proposition}[Quantum regression theorem]\label{prop:qreg_theorem}
    Suppose $H(t)$ is the Hamiltonian of a system with $M$ bosonic ports
    \[
    H(t) = H_S(t) + \sum_{\alpha = 1}^M (L_\alpha a_{\alpha, t}^\dagger + \text{h.c.}),
    \]
    where $H_S(t), L_\alpha$ are operators on the system Hilbert space, and $a_{\alpha, t}$ are annihilation operators on the bosonic ports with $[a_{\alpha, t}, a^\dagger_{\beta, s}] = \delta_{\alpha, \beta} \delta(t - s)$. Then, for any system operators $O_1, O_2 \dots O_k, Q_1, Q_2 \dots Q_k, P$, system state $\ket{\phi}$ and $0 < t_1 < t_2 \dots < t_{k - 1} < t_k$,
    \begin{align}
    &\bra{\phi, \textnormal{vac}}  \bigg(\prod_{j = 1}^{k - 1} O_j(t_j)\bigg)P(t_k) \bigg(\prod_{j = k - 1}^{1}  Q_j(t_j)\bigg) \ket{\phi, \textnormal{vac}} \nonumber\\
    &\qquad \qquad \qquad  = \textnormal{Tr}\bigg(P \bigg(\prod_{j = k - 1}^{1}\mathcal{E}(t_{j + 1}, t_j)  \mathcal{M}_j \bigg) \mathcal{E}(t_1, t_0)(\ket{\phi}\!\bra{\phi})\bigg),
    \end{align}
    where for system operator $X$, $X(t) = U(0, t) X U(t, 0)$ with $U(t, 0) = \overrightarrow{\mathcal{T}}\textnormal{exp}(-i\int_0^t H(s) ds)$, $\mathcal{E}(t, s) = \overrightarrow{\mathcal{T}}\textnormal{exp}(\int_s^t \mathcal{L}(\tau) d\tau)$ with $\mathcal{L}(\tau) = -i[H_S(t), \cdot] + \sum_{\alpha = 1}^M \mathcal{D}_{L_\alpha}$ and $\mathcal{M}_j(X) = Q_j X O_j$.
\end{proposition}
\begin{proof}
    To prove this statement, we note that in time basis, the Hilbert space of any the ports $\mathcal{H}$ can be expressed as $\mathcal{H} = \mathcal{H}_{(-\infty, 0] \cup (t_k, \infty)} \otimes \mathcal{H}_{(0, t_1]} \otimes \mathcal{H}_{(t_1, t_2]} \otimes \dots \otimes \mathcal{H}_{(t_{k - 1}, t_k]}$, where for an interval $I$, $\mathcal{H}_I$ is the Hilbert space of states with photons created by $a_{\alpha, t}$ with $t \in I$, i.e.
    \[
    \ket{\psi} \in \mathcal{H}_I: \ \ket{\psi} = \smallket{\text{vac}_I} + \sum_{n = 1}^\infty \int_{t_1, t_2 \dots t_n \in I} \psi_{\alpha_1, \alpha_2 \dots \alpha_n}(t_1, t_2 \dots t_n) a_{t_1}^\dagger a_{t_2}^\dagger \dots a_{t_n}^\dagger \smallket{\text{vac}_I} dt_1 dt_2 \dots dt_n,
    \]
    where $\smallket{\text{vac}_I}$ is the vacuum corresponding to the Hilbert space $\mathcal{H}_I$. We note that 
    \begin{align}\label{eq:vacuum_decomp}
    \ket{\text{vac}} = \smallket{\text{vac}_{(-\infty, 0] \cup (t_k, \infty)}} \otimes \smallket{\text{vac}_{(t_0, t_1]}} \otimes \smallket{\text{vac}_{(t_1, t_2]}} \otimes \dots \otimes \smallket{\text{vac}_{(t_{k - 1}, t_k]}}
    \end{align}
    Furthermore, for any system operator $Q$ and for $s < t$, we note that
    \begin{align}\label{eq:application_qreg}
    q(t, s) = \smallbra{\text{vac}_{(s, t]}}U(s, t) Q U(t, s) \smallket{\text{vac}_{(s, t]}} = \mathcal{E}^\dagger(t, s)(Q),
    \end{align}
    where $\mathcal{E}^\dagger(t, s)$ is the adjoint of $\mathcal{E}(t, s)$, i.e.,~it is a super-operator satisfying $\text{Tr}(\mathcal{E}^\dagger(t, s)(X) Y) = \text{Tr}(X \mathcal{E}(t, s)(Y))$ for all $X$ and $Y$. It is also easy to show that 
    \[
    \frac{d}{dt}\mathcal{E}^\dagger(t, s) = \mathcal{E}^\dagger(t, s) \mathcal{L}^\dagger(t) \text{ where }\mathcal{L}^\dagger(t) = -i[\cdot, H(t)] + \sum_{\alpha}\mathcal{D}_{L_\alpha}^\dagger,
    \]
    with $\mathcal{D}_L^\dagger(X) = L^\dagger X L - \{L^\dagger L, X\}/2$.
    
    To establish Eq.~(\ref{eq:application_qreg}), we begin by partitioning the interval $(s, t] = (\tau_0, \tau_1] \cup (\tau_1, \tau_2] \cup \dots \cup (\tau_{m - 1}, \tau_m]$, where $\tau_j = s + j \varepsilon $ with $\varepsilon = (t - s) / m$. Then, we note that for any system operator $Q'$,
    \begin{align*}
    &\smallbra{\text{vac}_{(\tau_{k - 1}, \tau_k]}} U(\tau_{k - 1}, \tau_k) Q' U(\tau_k, \tau_{k - 1})\smallket{\text{vac}_{(\tau_{k - 1}, \tau_k]}}\nonumber\\
    & \qquad = Q' -i\int_{\tau_{k - 1}}^{\tau_k}  \smallbra{\text{vac}_{(\tau_{k - 1}, \tau_k]}}[Q', H(s)] \smallket{\text{vac}_{(\tau_{k - 1}, \tau_k]}} ds - \nonumber\\
    &\qquad \qquad \qquad \qquad \qquad \int_{\tau_{k - 1}}^{\tau_k} \int_{\tau_{k - 1}}^s \smallbra{\text{vac}_{(\tau_{k - 1}, \tau_k]}} [[Q', H(s')], H(s)] \smallket{\text{vac}_{(\tau_{k - 1}, \tau_k]}} ds dt + O(\varepsilon^2) \nonumber\\
    &\qquad = Q' - i\int_{\tau_{k - 1}}^{\tau_k} [Q', H_S(s)]ds + \varepsilon \sum_{\alpha}\mathcal{D}_{L_\alpha}^\dagger(Q') + O(\varepsilon^2) = \mathcal{E}^\dagger(\tau_{k}, \tau_{k - 1})(Q') + O(\varepsilon^2).
    \end{align*}
    Therefore,
    \begin{align}
        q(t, s) &= \bigg(\prod_{k = 0}^{m - 1}\smallbra{\text{vac}_{(\tau_k, \tau_{k + 1}]}}\bigg) \bigg(\prod_{k = 0}^{m - 1}U(\tau_{k}, \tau_{k + 1})\bigg) Q \bigg(\prod_{k = m - 1}^0 U(\tau_{k + 1}, \tau_k)\bigg) \bigg(\prod_{k = m - 1}^{0}\smallket{\text{vac}_{(\tau_k, \tau_{k + 1}}}\bigg)
        \nonumber\\
        &= \bigg(\prod_{k = 0}^{m - 1}\smallbra{\text{vac}_{(\tau_k, \tau_{k + 1}]}}U(\tau_k, \tau_{k + 1}) \bigg) Q\bigg(\prod_{k = m - 1}^{0} U(\tau_{k + 1}, \tau_{k })\smallket{\text{vac}_{(\tau_k, \tau_{k + 1}]}} \bigg)\nonumber\\
        &= \bigg(\prod_{k = 0}^{m - 1}\mathcal{E}^\dagger(\tau_{k + 1}, \tau_k)\bigg)Q + O(m \varepsilon^2) \nonumber\\
        &= \mathcal{E}^\dagger(t, s)(Q) + O(\abs{t - s}\varepsilon).
    \end{align}
    Taking $\varepsilon \to 0$ establishes Eq.~\eqref{eq:application_qreg}. Next, using Eqs.~\eqref{eq:vacuum_decomp} and \ref{eq:application_qreg}, we obtain that
    \begin{align}
        &\bra{\phi, \textnormal{vac}}  \bigg(\prod_{j = 1}^{k - 1} O_j(t_j)\bigg)P(t_k) \bigg(\prod_{j = k - 1}^{1}  Q_j(t_j)\bigg) \ket{\phi, \textnormal{vac}} \nonumber\\
        &=\bra{\phi, \textnormal{vac}} U(t_0, t_1) \bigg(\prod_{j = 1}^{k - 1} O_j U(t_j, t_{j + 1})\bigg)P \bigg(\prod_{j = k - 1}^{1} U(t_{j + 1}, t_j) Q_j\bigg) U(t_1, t_0)\ket{\phi, \textnormal{vac}} \nonumber \\
        &=\smallbra{\phi}\bigg[ \smallbra{\text{vac}_{(t_0, t_1]}}U(t_0, t_1)  \bigg(\prod_{j = 1}^{k - 1}O_j \smallbra{\text{vac}_{(t_j,t_{j + 1}]}} U(t_{j}, t_{j + 1})\bigg) P \times \nonumber\\
        &\qquad \qquad \qquad \qquad \qquad  \qquad \qquad \bigg(\prod_{j = k - 1}^1 U(t_{j + 1}, t_j)\smallket{\text{vac}_{(t_j, t_{j + 1}]}} Q_j\bigg) U(t_1, t_0) \smallket{\text{vac}_{(t_0, t_1]}}\bigg] \ket{\phi} \nonumber \\
        & = \smallbra{\phi}\mathcal{E}^\dagger(t_1, t_0) \big(O_1 \mathcal{E}^\dagger(t_2, t_1)\big(O_2\mathcal{E}^\dagger(t_3, t_2)(O_3 \dots \mathcal{E}^\dagger(t_{k - 1}, t_{k - 2})(O_{k - 1}\mathcal{E}^\dagger(t_{k}, t_{k - 1}))(P)Q_{k - 1})\dots Q_3 \big) Q_2\big)Q_1\big)\ket{\phi} \nonumber\\
        &=\textnormal{Tr}\bigg(P \bigg(\prod_{j = k - 1}^{1}\mathcal{E}(t_{j + 1}, t_j)  \mathcal{M}_j \bigg) \mathcal{E}(t_1, t_0)(\ket{\phi}\!\bra{\phi})\bigg),
    \end{align}
    which proves the proposition statement.
\end{proof}

Next, we consider the correlator $\mathcal{C}^{(n)}_{x^1 \dots x^n; y^1 \dots y^n}(t_1 \dots t_n)$ defined in the main text:
\[
\mathcal{C}^{(n)}_{x^1 \dots x^n; y^1 \dots y^n}(t_1 \dots t_n) = \left\langle \prod_{i = 1}^n x^{i\dagger}_{t_i} \prod_{i = 1}^n y^i_{t_i}  \right\rangle,
\]
where $x^i, y^i \in \{a, b\}$. Using Eq.~(\ref{eq:psi}), this correlator can be expressed in the Heisenberg picture as
\begin{align}\label{eq:g_1_heisenberg_pic}
    \mathcal{C}^{(n)}_{x^1 \dots x^n; y^1 \dots y^n}(t_1 \dots t_n) &= \frac{1}{\mathcal{N}^2}\bra{\text{vac}, \phi_{S, i}} \Phi_{S, f}(T)\prod_{i = 1}^n x^{i\dagger}_{t_i}(T) \prod_{i = 1}^n y^i_{t_i}(T)\ket{\text{vac}, \phi_{S, i}}, \nonumber\\
&= \frac{1}{\mathcal{N}^2} \bra{\text{vac}, \phi_{S, i}}  \overleftarrow{\mathcal{T}}\left[\prod_{i = 1}^n x^{i\dagger}_{t_i}(T)\right]\Phi_{S, f}(T) \overrightarrow{\mathcal{T}}\left [\prod_{i = 1}^n y^{i}_{t_i}(T)\right]\ket{\text{vac}, \phi_{S, i}},
\end{align}
where $O(T) = U^\dagger(T, 0) O U(T, 0)$, $\Phi_{S, f} = \ket{\phi_{S, f}}\!\bra{\phi_{S, f}}$ and $\overrightarrow{\mathcal{T}}$ orders Heisenberg-picture operators in the decreasing order of time indices while $\overleftarrow{\mathcal{T}}$ performs this time-ordering in increasing order of time-indices. Note that in Eq.~(\ref{eq:g_1_heisenberg_pic}), the time ordering operators can be inserted fictitiously since all the Heisenberg-picture operators are evaluated at time $T$. Next, from the input-output formalism \cite{collett1984squeezing}, it follows that
\begin{align}\label{eq:input_output}
    x_t(T) = x_t(0) - iL_X(t) \Theta(0 \leq t \leq T),
\end{align}
where the indicator function $\Theta(0\leq t \leq T) = 1$ if $t \in (0, T)$, $1/2$ if $t \in \{0, T\}$ and $0$ otherwise. We now substitute Eq.~(\ref{eq:input_output}) into Eq.~(\ref{eq:g_1_heisenberg_pic}): Due to the time-ordering operator, we can then move $y_t^i(0)$ to the right and $x_t^{i\dagger}(0)$ to the left and apply them on the vacuum state in the port. Since $y_t^i(0)\ket{\text{vac}} = 0$ and $\bra{\text{vac}}x_t^{i\dagger}(0) = 0$, we obtain that 
\begin{align}\label{eq:qfi_corr_fun_heis_full}
    \mathcal{C}^{(n)}_{x^1 \dots x^n; y^1\dots y^n}(t_1 \dots t_n) = \frac{1}{\mathcal{N}^2} \bra{\textnormal{vac}, \phi_{S, i}} \overleftarrow{\mathcal{T}}\left [ \prod_{i = 1}^n L_{X^i}^\dagger(t_i) \right ] \Phi_{S, f}(T) \overrightarrow{\mathcal{T}}\left [\prod_{i = 1}^n L_{Y^i}(t_i) \right ]\ket{\textnormal{vac}, \phi_{S, i}}.
\end{align}
Now, applying the quantum regresssion theorem (\cref{prop:qreg_theorem}) to Eq.~(\ref{eq:qfi_corr_fun_heis_full}), we obtain Eq.~(\ref{eq:qfi_system_dynamics}) from the main text.

\subsection{Examples}\label{app:qfi_examples}
\begin{figure}
	\centering
	\includegraphics[width=0.85\linewidth]{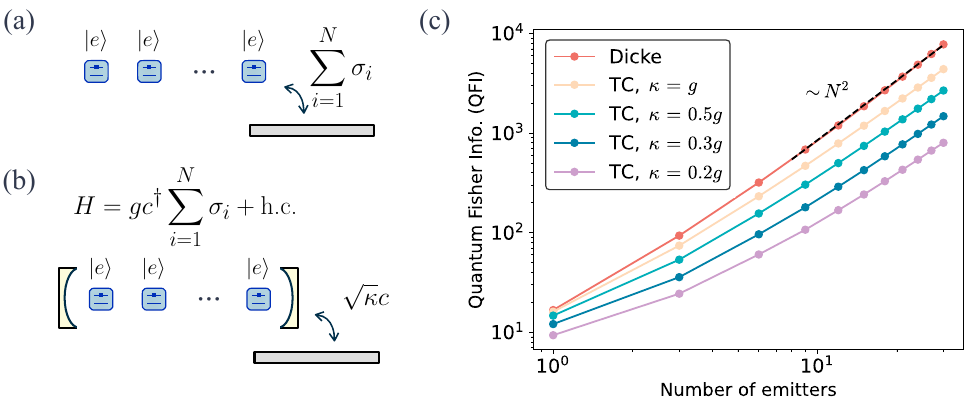}
	\caption{\label{fig:QFI_N} (a) Dicke model setup. All atoms are initially excited and collectively emit \( N \) photons into the port. The jump operator governing this process is given by $\sum_i\sigma_i$. (b) Tavis-Cumming model setup. All atoms are initially excited and interact with a cavity mode with an interaction strength \( g \). The annihilation operator of the cavity mode is denoted as \( c \). Photons are emitted from the cavity into the port at a decay rate $\kappa$.  (c) Scaling of the QFI with the number of emitters for the Dicke and Tavis-Cumming setup.}
\end{figure}
As examples of applying \cref{eq:QFI_corr_fun,eq:qfi_system_dynamics}, we consider some paradigmatic systems used as quantum light sources and study the QFI of the emitted photons. In all of our examples, we will consider the setting of two identical and independent sources emitting into the two input ports of the MZI.

\subsubsection{Harmonic Oscillator}
Consider first the simplest setting of an optical cavity with annihilation operator $a$ and frequency $\omega_c$. We consider the source Hamiltonian to be $H = \omega_c a^\dagger a$ and jump operator $L = a$. We will initialize the cavity in an initial state $\ket{\phi_i}$, and set the total emission time $T \to \infty$. The cavity would then have decayed to the vacuum state $\ket{\text{vac}}$ which we will use as the final state (i.e., $\ket{\phi_f} = \ket{\text{vac}}$).

For simplicity we will choose the reference frequency to be $\omega_c = 0$. For this setup, it is useful to note that $\mathcal{D}_a^\dag(X) = -X/2$ for $X \in \{a, a^\dagger\}$ and $\mathcal{D}_a^\dagger(X) = -X$ for $X \in \{a^2,a^{\dagger 2}, a^\dagger a\}$ which implies that
\begin{align}
&\mathcal{E}^\dagger(t, s)(X) = e^{-({t - s})/2} X \text{ if }X \in \{a, a^\dagger\}, \nonumber\\
&\mathcal{E}^\dagger(t, s)(X) = e^{-(t - s)} X \text{ if } X \in \{a^2, a^{\dagger2}, a^\dagger a\},
\end{align}
where, for a channel $\mathcal{E}(X) = \sum_i K_i X K_i^\dagger$, its adjoint is given by $\mathcal{E}^\dagger(X) = \sum_i K_i X K_i^\dagger$. From Eq.~(\ref{eq:QFI_identical_sources}), we then obtain that
\[
\textnormal{QFI} = 8 \bigg(\abs{\bra{\phi_i}a^\dagger a\ket{\phi_i}}^2 - \abs{\bra{\phi_i} a^2 \ket{\phi_i}}^2 + \bra{\phi_i}a^\dagger a \ket{\phi_i}\bigg).
\]
\begin{enumerate}
    \item[(a)] Suppose $\ket{\phi_A} = \ket{\alpha}$, i.e., the initial state is a coherent state with (mean) photon number $N = \abs{\alpha}^2$. We then obtain
    \[
    \textnormal{QFI} = 8\abs{\alpha}^2 \sim N \text{ for large }N,
    \]
    and matches the expected scaling.
    \item[(b)] Suppose $\ket{\phi_A} = \ket{N}$, i.e., the initial state is Fock state, then
    \[
    \textnormal{QFI} = 8(N^2 + N) \sim N^2 \text{ for large }N,
    \]
    and matches the expected scaling.
\end{enumerate}

\subsubsection{Dicke and Tavis-Cumming model}
Next, using the QFI formula in Eq.~(\ref{eq:qfi_system_dynamics}) to compute QFI, we simulate and reproduce the result from Ref.~\cite{paulisch2019quantum} for a source described by the Dicke Model, where $N$ excited 2-level systems emit collectively into the optical port [Fig.~\ref{fig:QFI_N}(a)]. In this case, the source Hamiltonian $H = 0$ and the jump operator $L = \sum_{i=1}^N \sigma_i$, where $\sigma_i = \ket{g_i}\!\bra{e_i}$. In Fig.~\ref{fig:QFI_N}(c), we plot the QFI as a function of the number of emitters, $N$. The QFI scales quadratically with $N$ as expected.

Next, we consider a source described by the Tavis-Cumming model \cite{trivedi2019photon}, which comprises of an optical cavity coupling to $N$ emitters [Fig.~\ref{fig:QFI_N}(b)]. Here, the source Hamiltonian is described by
\begin{equation}
    H = g c^\dag \sum_i(\sigma_i + \sigma_i^\dagger)+ \text{h.c.} ,
\end{equation}
where \( g \) is the coupling strength between the two-level systems and the cavity, \( c \) is the annihilation operator of the cavity, and \( \sigma_i = \ket{g_i}\!\bra{e_i} \). Initially, all \( N \) emitters are excited and subsequently emit into the cavity, which then emits into the optical port. The photon emission from the cavity into the output port is described by the jump operator \(L =  \sqrt{\kappa} c \). We simulate the QFI of the emitted photons for different values of the cavity-to-port decay rates, $\kappa$ relative to the emitter-cavity interaction strength $g$ [Fig.~\ref{fig:QFI_N}(c)]. In all cases, we observe that the QFI scales quadratically with the number of emitters $N$ (which is equal to the number of emitted photons), demonstrating Heisenberg-limited scaling. Furthermore, for larger values of $\kappa$, the QFI for the Tavis-Cumming setup approaches the Dicke case, as expected.

\begin{figure}
	\centering
	\includegraphics[width=0.9\linewidth]{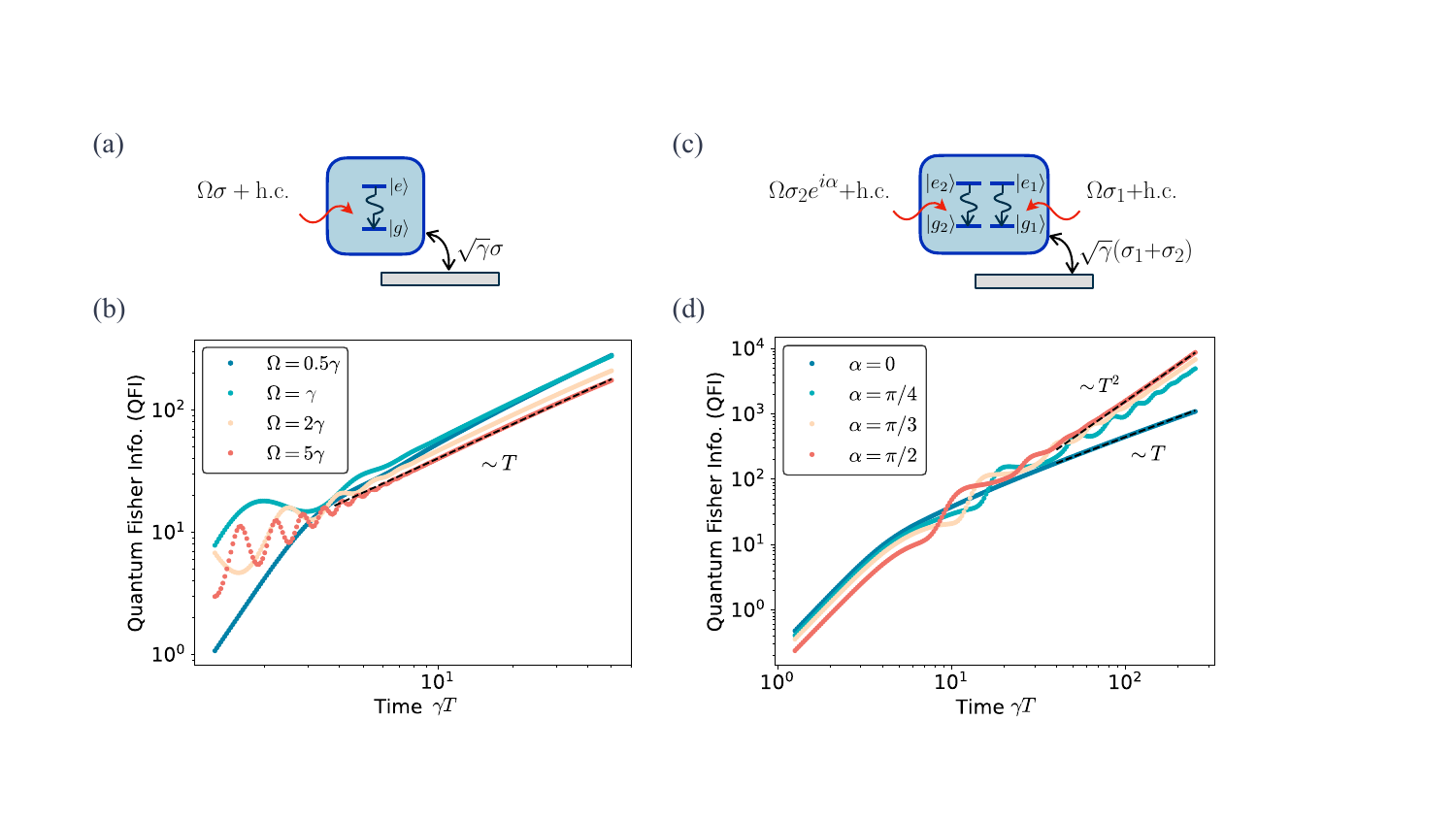}
	\caption{\label{fig:2level_QFI_T} (a) Schematic of the 2-level system setup, where the system is driven by a laser with amplitude \( \Omega \) and emits photons into the port at a decay rate \( \gamma \). (b) Scaling of the QFI with time for the 2-level system. The QFI scales linearly with time for different values of $\Omega$. (c) Schematic of the $\pi$-level system setup, where the system is driven by two lasers with amplitude \( \Omega \) and phase difference \( \alpha \), emitting photons into the port at a decay rate \( \gamma \). (d) Scaling of the QFI with time for the $\pi$-level system with $\Omega = \gamma /\sqrt{8}$. The QFI scales quadratically with time for non-zero $\alpha$, while for $\alpha =0$ it scales linearly.}
\end{figure}

\subsubsection{Driven 2-level and $\pi-$level source}
Next, we consider a continuously driven two level system [Fig.~\ref{fig:2level_QFI_T}(a)]. The source Hamiltonian for this system is given by $H = (\Omega \sigma + \text{h.c.})$ and the jump operator $L = \sqrt{\gamma}\sigma$ where $\ket{g}\!\bra{e}$. The Lindbladian specified by $H$ and $L$ has a unique fixed point $\uptau$ given by
\begin{align}\label{eq:2ls_fp}
\uptau = \frac{1}{1 + 8 \abs{\Omega}^2}\big((1 + 4\abs{\Omega}^2)\ket{g}\!\bra{g} +4\Omega^2\ket{e}\!\bra{e} + 2i\Omega\ket{g}\!\bra{e} - 2i\Omega^* \ket{e}\!\bra{g}\big).
\end{align}
Since this fixed point is unique, as per the discussion in Section \ref{sec:source_spectrum_qfi}, we expect the QFI to scale linearly with $T$ asymptotically. This is numerically demonstrated in Fig.~\ref{fig:2level_QFI_T}(b) where we plot the QFI as a function of time for $\Omega\in\{0.5\gamma, \gamma, 2\gamma, 5\gamma\}$.

Next we consider a $\pi-$level system with two excited states $\ket{e_1}, \ket{e_2}$ and two ground states $\ket{g_1}, \ket{g_2}$ [Fig.~\ref{fig:2level_QFI_T}(c)]. The source Hamiltonian $H$ has a coherent drive applied on the two transitions ($\ket{e_1}\to \ket{g_1}$ and $\ket{e_2}\to \ket{g_2}$) but with different phases, i.e., $H = \Omega_0 (\sigma_1 + \sigma_2 e^{i\alpha} + \text{h.c.})$, where $\ket{\sigma_i} = \ket{g_i}\!\bra{e_i}$. The source is coupled to the output port with the jump operator $L = \sigma_1 + \sigma_2$. The Lindbladian specified by $H$ and $L$ has two fixed points, $\uptau_1, \uptau_2$ where $\uptau_1$ is given by Eq.~(\ref{eq:2ls_fp}) with $\ket{e}\to \ket{e_1}, \ket{g}\to \ket{g_1}, \Omega \to \Omega_0 $ and $\uptau_2$ is given by Eq.~(\ref{eq:2ls_fp}) with $\ket{e}\to \ket{e_2}, \ket{g} \to \ket{g_2}, \Omega \to \Omega_0 e^{i\alpha}$. To analyze the QFI of the emitted photons, we note that $\mathcal{E}(t, s)$ can be expressed as
\begin{equation} \label{eq:two_FP}
    \mathcal{E}(t, s) \rho = \tr(P_1\rho) \uptau_1 + \tr(P_2\rho) \uptau_2 + \mathcal{M}(t, s) \rho,
\end{equation}
where $P_i = \ket{e_i}\!\bra{e_i} + \ket{g_i}\!\bra{g_i}$ and $\smallnorm{\mathcal{M}(t, s)}_\diamond \leq O(e^{-\gamma_0(t - s)})$ as $\abs{t - s} \to \infty$ for some $\gamma_0 > 0$. We will consider the initial source state to be $\ket{\phi_i} = (\ket{g_1} + \ket{g_2})/\sqrt{2}$ and the final source state to be $\ket{\phi_f} = (\ket{g_1} + \ket{g_2})/\sqrt{2}$. We now estimate $\text{Q}^{(2)}$ defined in Eq.~(\ref{eq:QFI_identical_sources}). In the limit of $\gamma_0 T \gg 1$, we obtain that
\begin{equation}
    \textnormal{Q}^{(2)} = {2T^2} \frac{ \big( \smallabs{\sum_j \abs{\Tr(L\uptau_j)}^2 \Tr(\uptau_j \rho_f) \Tr(P_j \rho_i)}^2 - \smallabs{\sum_j {\Tr(L\uptau_j)}^2 \Tr(\uptau_j \rho_f) \Tr(P_j \rho_i)}^2 \big) }{\big( {\sum_j \Tr(\uptau_j \rho_f) \Tr(P_j \rho_i)} \big)^2} + O(T),
\end{equation}
where $\rho_i = \ket{\phi_i}\!\bra{\phi_i}$ and $\rho_f = \ket{\phi_f}\!\bra{\phi_f}$. Using the explicit expressions for $\uptau_1, \uptau_2, P_1, P_2$ and $\ket{\phi_i}, \ket{\phi_f}$, we obtain that 
\begin{equation}\label{eq:qfi_pi_ls}
    \textnormal{Q}^{(2)} = {2T^2} \frac{32^2{\Omega_0}^4 - 16^2{\Omega_0}^4\abs{1 + e^{-2i\alpha}}^2}{(1 + 8 {\Omega_0}^2)^4} + O(T) =  \frac{2^{11}{\Omega_0}^4}{(1 + 8 {\Omega_0}^2)^4} T^2 \sin^2\alpha + O(T).
\end{equation}
Figure \ref{fig:2level_QFI_T}(d) shows numerically simulated QFI for the $\pi-$level system as a function $\alpha$. As expected from Eq.~(\ref{eq:qfi_pi_ls}), the QFI shows Heisenberg limited scaling ($\sim T^2$) when $\alpha \neq 0$, while when $\alpha = 0$, the QFI scales as $T$. Physically, this can be attributed to the fact that when $\alpha = 0$, the $\pi-$level system becomes indistiguishable from the 2-level system, and thus the $\text{QFI}\sim T$. On the other hand, when $\alpha \neq 0$, the emitted photons are a superposition of two macroscopic photon states and hence the state exhibits QFI $\sim T^2$.

\section{Quantum Fisher Information scaling with time}
\subsection{Strictly contractive source dynamics implies no quantum advantage in interferometry}\label{app:strictly_contr}

\subsubsection{Strictly contractive dynamics forbid quantum advantage}
We recall from the main text that a Lindbladian $\mathcal{L}(t)$ is said to generate strictly contractive dynamics if $\exists C_0, \tau_0, \gamma_0 > 0$ such that $\forall $ states $\rho_1, \rho_2$, 
\begin{align}\label{eq:strictly_contr_app}
    \norm{\mathcal{E}(t, s)(\rho_1 - \rho_2)}_1 \leq C_0 e^{-\gamma_0\abs{t - s}}\norm{\rho_1 - \rho_2}_1 \qquad \forall  \ \abs{t - s} \geq \tau_0,
\end{align}
where $\mathcal{E}(t, s) = \mathcal{T}\text{exp}(\int_s^t \mathcal{L}(\tau)d\tau)$. We first establish that this condition implies that the QFI introduced in Eqs.~\eqref{eq:QFI_identical_sources} scales as $T$, and thus exhibits standard quantum limit (SQL) scaling. We begin by noting the following simple fact.
\begin{lemma}
    Suppose $\mathcal{E}(t, s)$ is strictly contractive [Eq.~(\ref{eq:strictly_contr_app})], then for any traceless (possibly non-Hermitian) operator $B$,
    \begin{align}\label{eq:strict_contr_non_herm}
\norm{\mathcal{E}(t, s)(B)}_1 \leq 2C_0 e^{-\gamma_0{\abs{t - s}}}\norm{B}_1 \ \forall \abs{t - s} \geq \tau_0.
\end{align}
\end{lemma}
\begin{proof}
   Consider first the case where $B = B^\dagger$. Suppose $B$ has eigenvalues and eigenvectors  $(\lambda_\alpha, \ket{\alpha})$ --- then $B = B^{(+)} - B^{(-)}$
   \[
   B^{(+)} = \sum_{\alpha: \lambda_\alpha \geq 0}\lambda_\alpha \ket{\alpha}\!\bra{\alpha} \text{ and }B^{(-)} = -\sum_{\alpha: \lambda_\alpha < 0}\lambda_\alpha \ket{\alpha}\!\bra{\alpha}.
   \]
   Note that $B^{(+)}, B^{(-)} \succeq 0$ and $\text{Tr}(B^{(+)}) = \text{Tr}(B^{(-)}) = \norm{B}_1/2$. Therefore, using Eq.~(\ref{eq:strictly_contr_app}) with $\rho_1 = 2B^{(+)} / \norm{B}_1, \rho_2 = 2B^{(-)}/\norm{B}_1$, we obtain that for $\abs{t - s} \geq \tau_0$,
\begin{align}\label{eq:strict_contr_herm}
\norm{\mathcal{E}(t, s)(B)}_1 \leq C_0 e^{-\gamma_0 \abs{t - s}} \norm{B}_1.
\end{align}
Finally, any traceless (but possibly non-Hermitian) operator $B$ can be expressed as $B = B_h + i B_{ah}$ where $B_h = (B + B^\dagger)/2$, $B_{ah} = i(B^\dagger - B)/2$ are both Hermitian operators. Using Eq.~(\ref{eq:strict_contr_herm}), we then obtain that for $\abs{t - s}\geq \tau_0$
\begin{align}\label{eq:strict_contr_non_herm}
\norm{\mathcal{E}(t, s)(B)}_1 \leq \norm{\mathcal{E}(t, s)(B_h)}_1 + \norm{\mathcal{E}(t, s)(B_{ah})}_1 \leq C_0 e^{-\gamma_0\abs{t - s}}(\norm{B_h}_1 + \norm{B_{ah}}_1) \leq 2C_0 e^{-\gamma \abs{t - s}}\norm{B}_1,
\end{align}
which proves the Lemma.
\end{proof}
\noindent Next, we show that strict contractivity of the dynamics implies that expected values of time-domain annihilation operators in the emitted state approximately factorize. We consider here the setting introduced in the main text of the source coupling to a set of output ports described by the time-domain annihilation operator $a_{\alpha,  t}$,
\begin{align}\label{eq:source_port_hamiltonian}
H(t) = H_S(t) + \sum_{\alpha} \big(a_{\alpha, t}^\dagger L_\alpha +\text{h.c.}\big).
\end{align}
The source dynamics will be captured by the Lindbladian $\mathcal{L}(t) = -i[H_S(t), \cdot] + \sum_{\alpha}\mathcal{D}_{L_\alpha}$. We will now consider the state $\ket{\psi}$ in the output ports obtained on evolving an initial state $\ket{\phi_i, \text{vac}}$ (i.e.,~the source being in the state $\ket{\phi_i}$ and the output ports being in the vacuum state) for time $T$ and projecting the source onto the state $\ket{\phi_f}$:
\[
\ket{\psi} = \frac{1}{\mathcal{N}} \bra{\phi_f}U(T, 0) \ket{\phi_i, \text{vac}} \text{ where } U(t, s) = \mathcal{T}\exp\bigg(-\int_{s}^t H(\tau) d\tau\bigg) \text{ and } \mathcal{N} = \smallabs{\bra{\phi_f}U(T, 0) \ket{\phi_i, \text{vac}}}^2.
\]
We now establish the fact that, when the source dynamics is strictly contractive, the connected correlators in the emitted photonic state at two different time instants decay exponentially with their time-difference.
\begin{lemma}[Decay of connected correlators]\label{lemma:conn_corr_decay}
    Suppose $\mathcal{E}(t, s) = \overrightarrow{\mathcal{T}}\textnormal{exp}(\int_s^t \mathcal{L}(\tau)d\tau)$, where $\mathcal{L}(\tau)$ corresponding to the source-port Hamiltonian in Eq.~(\ref{eq:source_port_hamiltonian}), is strictly contractive [Eq.~(\ref{eq:strictly_contr_app})], then ,
    \begin{align}\label{eq:decay_connected_corr_ann}
     \smallabs{\bra{\psi} a_{\alpha, t} a_{\beta, s}\ket{\psi} - \bra{\psi} a_{\alpha, t}\ket{\psi} \bra{\psi}a_{\beta, s}\ket{\psi}} \leq \frac{8C_0}{\mathcal{N}^4}\smallnorm{L_\alpha}\smallnorm{L_\beta}e^{-\gamma_0(t - s)},
    \end{align}
    and
    \begin{align}\label{eq:decay_connected_corr_crea}
     \smallabs{\bra{\psi} a_{\alpha, t}^\dagger a_{\beta, s}\ket{\psi} - \bra{\psi} a_{\alpha, t}^\dagger\ket{\psi} \bra{\psi}a_{\beta, s}\ket{\psi}} \leq \frac{8C_0}{\mathcal{N}^4}\smallnorm{L_\alpha}\smallnorm{L_\beta}e^{-\gamma_0(t - s)},
    \end{align}
    where $0 \leq s < t\leq T$ and $t - s \geq \tau_0$.
\end{lemma}
\begin{proof}
    We will establish Eq.~(\ref{eq:decay_connected_corr_ann}), the proof of Eq.~(\ref{eq:decay_connected_corr_crea}) will follow similarly. Following the calculation in Section \ref{app:QFI_comp}, $\bra{\psi}a_{\alpha, t} a_{\beta, s}\ket{\psi}$ and $\bra{\psi}a_{\nu, \tau}\ket{\psi}$ can be expressed as
    \[
    \bra{\psi} a_{\alpha, t} a_{\alpha, s} \ket{\psi} = \frac{1}{\mathcal{N}^2} \text{Tr}(P_f(t) L_\alpha \mathcal{E}(t, s) (L_\beta \rho_i(s))) \text{ and }\bra{\psi}a_{\nu, \tau}\ket{\psi} = \frac{1}{\mathcal{N}^2} \text{Tr}(P_f(\tau)L_\nu \rho_i(\tau)),
    \]
    where $\rho_i(\tau) = \mathcal{E}(\tau, 0)(\rho_i)$ with $\rho_i = \ket{\phi_i}\!\bra{\phi_i}$ and $\mathcal{P}_f(\tau) = \mathcal{E}^\dagger(T, \tau)(\ket{\phi_f}\!\bra{\phi_f})$. Using these expressions, we then obtain
    \begin{align}\label{eq:inequality_1}
        &\smallabs{\bra{\psi} a_{\alpha, t} a_{\beta, s}\ket{\psi} - \bra{\psi} a_{\alpha, t}\ket{\psi} \bra{\psi}a_{\alpha, s}\ket{\psi}}\nonumber\\
        &\qquad = \frac{1}{\mathcal{N}^4}\abs{\text{Tr}(\mathcal{P}_f(t)L_\alpha \mathcal{E}(t, s)(\mathcal{N}^2L_\beta \rho_i(s) - \text{Tr}(\mathcal{P}_f(s)L_\beta \rho_i(s))\rho_i(s))}, \nonumber\\
        &\qquad \leq \frac{\norm{L_\alpha}}{\mathcal{N}^2}\smallnorm{\mathcal{E}(t, s)( L_\beta \rho_i(s) -\text{Tr}(L_\beta \rho_i(s)) \rho_i(s))}_1 + \frac{\norm{L_\alpha}}{\mathcal{N}^4}\smallabs{\text{Tr}(\mathcal{P}_f(s) L_\beta \rho_i(s)) - \mathcal{N}^2\text{Tr}(L_\beta \rho_i(s))}.
    \end{align}
    Similarly, using $\mathcal{N}^2 = \text{Tr}(\mathcal{P}_f\mathcal{E}(T, 0)(\rho_i)) = \text{Tr}(\mathcal{P}_f\mathcal{E}(T, s)\rho_i(s))$, we obtain that
    \begin{align}\label{eq:inequality_2}
        \smallabs{\text{Tr}(\mathcal{P}_f(s) L_\beta \rho_i(s)) - \mathcal{N}^2\text{Tr}(L_\beta \rho_i(s))} &= \smallabs{\text{Tr}(\mathcal{P}_f \mathcal{E}(T, s)(L_\beta \rho_i(s) - \text{Tr}(L_\beta \rho_i(s))\rho_i(s))}, \nonumber\\
        &\leq \smallnorm{\mathcal{E}(T, s)(L_\beta \rho_i(s) - \text{Tr}(L_\beta \rho_i(s)) \rho_i(s))}_1.
    \end{align}
    Combining Eqs.~\eqref{eq:inequality_1} and \eqref{eq:inequality_2}, and using the strict contractivity of $\mathcal{E}(t, s)$, we then obtain that
    \begin{align*}
    \smallabs{\bra{\psi}a_{\alpha, t}a_{\beta, s}\ket{\psi} - \bra{\psi}a_{\alpha, t}\ket{\psi}\bra{\psi}a_{\beta, s}\ket{\psi}} &\leq 2C_0 \norm{L_\alpha} \bigg(\frac{e^{-\gamma_0(t - s)}}{\mathcal{N}^2} + \frac{e^{-\gamma_0(T - s)}}{\mathcal{N}^4}\bigg) \smallnorm{L_\beta \rho_i(s) - \text{Tr}(L_\beta \rho_i(s)) \rho_i(s)}_1, \nonumber \\
    &\leq \frac{8C_0}{\mathcal{N}^4} \norm{L_\alpha}\norm{L_\beta} e^{-\gamma_0(t - s)}.
    \end{align*}
    This establishes Eq.~(\ref{eq:decay_connected_corr_ann}). A similar analysis can be performed to establish Eq.~(\ref{eq:decay_connected_corr_crea}).
\end{proof}

Using lemma \ref{lemma:conn_corr_decay}, we now show that the QFI [given by Eq.~(\ref{eq:QFI_identical_sources})] scales as $\sim T$ for large $T$. Consider the two correlation functions $C^{(g)}(t, s) = \langle a_t^\dagger a_s\rangle$ and $C^{(\chi)}(t, s) = \langle a_t a_s\rangle$. Applying lemma \ref{lemma:conn_corr_decay} together with the fact that $\smallabs{\langle a_t \rangle \langle a_s \rangle} = \smallabs{\langle a_t^\dagger \rangle \langle a_s \rangle}$ and $\smallabs{\langle a_t a_s\rangle}, \smallabs{\langle a_t^\dagger a_s \rangle} \leq \norm{L}^2, \smallabs{a_t}, \smallabs{a_t^\dagger}\leq \norm{L}$, we obtain that
\begin{align}
    \smallabs{C^{(\chi)}(t, s)}^2 - \smallabs{C^{(g)}(t, s)}^2 \leq \frac{64 C_0^2}{\mathcal{N}^8}  \norm{L}^4 e^{-2\gamma_0\abs{t - s}} + \frac{32C_0}{\mathcal{N}^4} \norm{L}^4 e^{-\gamma_0\abs{t - s}} \leq \frac{32 \norm{L}^4}{p_0^4}(2C_0^2 + C_0) e^{-\gamma_0\abs{t- s}},
\end{align}
where $p_0$ is a $T-$independent constant such that $p_0 \leq \mathcal{N}^2$.
Finally, using this together with the expression for $\textnormal{Q}^{(2)}$ [Eq.~(\ref{eq:QFI_identical_sources}b)] we obtain that
\begin{align}
    \text{Q}^{(2)} &= 2\int_0^T \int_0^t \big(\smallabs{C^{(\chi)}(t, s)}^2 - \smallabs{C^{(g)}(t, s)}^2\big) ds dt \nonumber \\
    &= 2\int_0^T \int_{\tau}^T  \big(\smallabs{C^{(\chi)}(t, t - \tau)}^2 - \smallabs{C^{(g)}(t, t - \tau)}^2\big)dt d\tau \nonumber \\
    &\leq \frac{2}{p_0^4}\int_{0}^{\tau_0} \int_{\tau}^T \norm{L}^4 dt d\tau + \frac{64}{p_0^4}\int_{\tau_0}^T \int_{\tau}^T \norm{L}^4 (2C_0^2 + C_0) e^{-\gamma_0 \tau} dt d\tau \nonumber \\
    &\leq \frac{2\norm{L}^4}{p_0^4}\bigg(\tau_0 + \frac{32}{\gamma_0} (2C_0^2 + C_0)  \bigg) T,
\end{align}
which establishes that $\text{Q}^{(2)} \leq O(T)$ as $T \to \infty$ and consequently from Eq.~(\ref{eq:QFI_identical_sources}a), QFI $\leq O(T)$.

\subsubsection{Provable strict contractivity for  Lindbladians with full Kraus rank}\label{app:provable_str_contr_fkr}
Next, we will consider two settings where we are able to establish that the source dynamics are provably strictly contractive irrespective of the time-dependent Hamiltonian applied on the source. The first setting that we will consider is when the source dynamics are modelled by a Lindbladian with a time dependent Hamiltonian as well as a set of jump operators which span the whole space of traceless operators operators. More specifically, for a source with $d$ levels, the source Lindbladian $\mathcal{L}(t)$ is given by
\begin{align}\label{eq:lindblad_str_contr}
    \mathcal{L}(t) = -i[H(t), \cdot] + \sum_{\alpha = 1}^{d^2 - 1} \mathcal{D}_{L_\alpha} = -i[H(t), \cdot] + \sum_{\alpha = 0}^{d^2} \mathcal{D}_{L_\alpha},
\end{align}
where $\{L_\alpha\}_{\alpha \in \{1, 2 \dots d^2 - 1\}}$ span the space of traceless system operators. Furthermore, we have artificially added the dissipator $\mathcal{D}_{L_0}$ with $L_0 = I$ for convenience without changing the master equation: The set of jump operators $\{L_\alpha\}_{\alpha \in \{0, 1 \dots d^2 - 1\}}$ form a complete basis for the space of system operators. Due this completeness of $\{L_\alpha\}_{\alpha \in \{0, 1, 2 \dots d^2 - 1\}}$, $\exists \lambda_0 > 0$ such that
\begin{align}\label{eq:jump_op_full_rank}
\forall \ \text{all system operators } X: \sum_{\alpha} \smallabs{\text{Tr}(L_\alpha^\dagger X)}^2 \geq \lambda_0 \text{Tr}(X^\dagger X).
\end{align}
We first show that the Lindbladian [Eq.~(\ref{eq:lindblad_str_contr})] with the jump operators satisfying Eq.~(\ref{eq:jump_op_full_rank}) generates strictly contractive dynamics.

\begin{proposition}\label{prop:contractivity_full_kraus}
    Suppose $\mathcal{E}(t, s) = \overrightarrow{\mathcal{T}}\textnormal{exp}(\int_s^t \mathcal{L}(\tau) d\tau)$ is the channel generated by a Lindbladian [Eq.~(\ref{eq:lindblad_str_contr})] whose jump operators satisfy Eq.~(\ref{eq:jump_op_full_rank}), then $\forall t > s$ and any two states $\rho_1, \rho_2$,
    \[
    \smallnorm{\mathcal{E}(t, s)(\rho_1 - \rho_2)}_1 \leq e^{-\lambda_0 \abs{t - s}}\smallnorm{\rho_1 - \rho_2}_1.
    \]
\end{proposition}
\begin{proof}
    We will denote by $\mathcal{D}$ the Lindbladian $\mathcal{D} = \sum_{\alpha = 1}^{d^2} \mathcal{D}_{L_\alpha}$.
    $\mathcal{D}$ contains contribution from all the jump operators in system Lindbladian which are assumed to be time independent. It will also be convenient for us to work in a discretization of this dynamics: We will discretize the time interval $(t, s] = (t_T, t_{T - 1}] \cup (t_{T - 1}, t_{T - 2}] \cup \dots \cup (t_{1}, t_0]$ where $t_m = s + m \varepsilon$ with $\varepsilon = (t - s)/T$ and approximate the channel $\mathcal{E}(t, s)$ as
    \begin{align}\label{eq:discretization}
    \tilde{\mathcal{E}} = \mathcal{K} \mathcal{U}_{T} \mathcal{K} \mathcal{U}_{T - 1} \dots \mathcal{K} \mathcal{U}_1,
    \end{align}
    where $\mathcal{U}_\tau(X) = U_\tau X U_\tau^\dagger $, with $U_\tau = \overrightarrow{\mathcal{T}}\textnormal{exp}(-i\int_{t_{\tau - 1}}^{t_\tau} H_S(s) ds)$ and $\mathcal{K}(X) = \sum_{\alpha = 0}^{d^2} K_\alpha X K_\alpha^\dagger$ where
    \[
    K_\alpha = \sqrt{\varepsilon} L_\alpha \text{ for }i \in \{1, 2 \dots d^2\}\text{ and }K_0 = \bigg({I - \varepsilon\sum_{\alpha = 1}^{d^2 } L_\alpha^\dagger L_\alpha}\bigg)^{1/2},
    \]
    We note that
    \begin{align}\label{eq:appx_channel_D}
    \smallnorm{e^{\varepsilon \mathcal{D}} - \mathcal{K}}_\diamond \leq O(\varepsilon^2),
    \end{align}
    and consequently 
    \begin{align}\label{eq:discretization_error}
        \smallnorm{\mathcal{E}(t, s) - \tilde{\mathcal{E}}}_\diamond \leq O(T \varepsilon^2) \leq O(\abs{t - s}\varepsilon),
    \end{align}
     and thus the exact dynamics is recovered in the limit of $\varepsilon \to 0$. We denote the fixed point of the Lindbladian $\mathcal{D}$ by $\uptau$. Since the jump operators $\{L_\alpha\}_{\alpha \in\{0, 1, 2\dots d^2 - 1\}}$ have full Kraus rank, $\uptau$ is both the unique fixed point of $\mathcal{D}$ and also positive definite ($\uptau \succ 0$). Furthermore, using Eq.~(\ref{eq:appx_channel_D}) and the fact that $\mathcal{D}(\uptau) = 0$, we obtain that $\smallnorm{\mathcal{K}(\uptau) - \uptau}_1 = \smallnorm{\mathcal{K}(\uptau) - e^{\varepsilon\mathcal{D}}(\uptau)}_1 \leq O(\varepsilon^2)$. Next, we decompose the channel $\mathcal{K}$ into the convex combination of a superoperator $\tilde{\mathcal{K}}$ and the replacement channel $\mathcal{R}_\uptau = \text{Tr}(\cdot)\uptau$:
    \begin{align}\label{eq:decomp_K}
    \mathcal{K} = (1 - \varepsilon \lambda_0)\tilde{\mathcal{K}} + \varepsilon \lambda_0 \mathcal{R}_{\uptau}.
    \end{align}
    We first show that $\tilde{\mathcal{K}}$ defined this way is a valid channel. First, we note that $\tilde{\mathcal{K}}$ is trace-preserving since both $\mathcal{K}$ and $\mathcal{R}_\uptau$ are trace preserving. Second, to ensure that is also completely positive, by the Choi-Jamiolkowski isomorphism, we need to show that $\Phi_{\tilde{\mathcal{K}}} = (\Phi_{\mathcal{K}} - \varepsilon \lambda_0  \Phi_{\mathcal{R}_\uptau}) / (1 - \varepsilon \lambda_0) \succeq 0$, where $\Phi_\mathcal{S} = (\mathcal{S}\otimes \text{id})(\ket{\Phi^+}\!\bra{\Phi^+})$ is the Choi state of $\mathcal{S}$. We note that for a state $\ket{\psi} = \sum_{i, j}\psi_{i, j}\ket{i, j} \in \mathbb{C}^d \otimes \mathbb{C}^d$, and its corresponding operator $\ket{\Psi} = \sum_{i, j} \psi_{i, j}\ket{i}\! \bra{j}$,
    \begin{align}\label{eq:choi_state_K_action}
    \bra{\psi} \Phi_{\mathcal{K}}\ket{\psi} &= \frac{1}{d}\sum_{\alpha} \sum_{i, i'} \sum_{p, p'} \psi_{p, i}^* \psi_{p', i'} \bra{p}K_\alpha \ket{i} \bra{i'}K_\alpha^\dagger \ket{p'} \nonumber\\
    &=\frac{1}{d} \sum_{\alpha} \abs{\text{Tr}(K_\alpha^\dagger \Psi)}^2 \nonumber \\
    &\geq \frac{\varepsilon}{d}\sum_{\alpha \geq 1}\abs{\text{Tr}(L_\alpha^\dagger \Psi)}^2 \geq \frac{\varepsilon \lambda_0}{d}\norm{\psi}^2.
    \end{align}
    Similarly,
    \begin{align}\label{eq:choi_state_R_action}
    \bra{\psi}\Phi_{\mathcal{R}_\uptau}\ket{\psi} = \frac{1}{d}\bra{\psi}(I \otimes \uptau) \ket{\psi} \leq \frac{1}{d} \norm{\psi}^2 \norm{\uptau} \leq \frac{1}{d}\norm{\psi}^2.
    \end{align}
    Therefore, from Eqs.~\eqref{eq:choi_state_K_action} and \eqref{eq:choi_state_R_action}, we find that for every $\ket{\psi} \in \mathbb{C}^d \otimes \mathbb{C}^d$, 
    \begin{align*}
        \bra{\psi} \Phi_{\tilde{\mathcal{K}}} \ket{\psi} = \frac{1}{1 - \varepsilon \lambda_0}\bigg(\bra{\psi} \mathcal{K}\ket{\psi} - \varepsilon \lambda_0 \bra{\psi}\mathcal{R}_\uptau \ket{\psi} \bigg) \geq 0,
    \end{align*}
    and thus $\tilde{\mathcal{K}}$ is a valid channel. Using Eq.~(\ref{eq:decomp_K}), we can then show that $\mathcal{K}$ is strictly contractive since for any two states $\rho_1, \rho_2$,
    \[
    \norm{\mathcal{K}(\rho_1 - \rho_2)}_1 = (1 - \varepsilon \lambda_0) \smallnorm{\tilde{\mathcal{K}}(\rho_1 - \rho_2)}_1 \leq (1 - \varepsilon \lambda_0) \norm{\rho_1 - \rho_2}_1.
    \]
    Using this together with Eq.~(\ref{eq:discretization}), we obtain that
    \[
    \smallnorm{\mathcal{\tilde{E}}(\rho_1 - \rho_2)}_1 \leq (1 - \varepsilon \lambda_0)^T \norm{\rho_2 - \rho_1}.
    \]
    Using Eq.~(\ref{eq:discretization_error}), we obtain that 
    \[
    \smallnorm{\mathcal{E}(t, s)(\rho_1 - \rho_2)}_1 \leq \smallnorm{\mathcal{\tilde{E}}(\rho_1 - \rho_2)}_1 + O(\abs{t - s}\varepsilon) \leq (1 - \varepsilon\lambda_0)^T \norm{\rho_1 - \rho_2} + O(\abs{t - s}\varepsilon).
    \]
    Finally, taking the limit $\varepsilon \to 0$ in this inequality, and setting $(1 - \varepsilon \lambda_0)^T \to e^{-\lambda_0 \abs{t - s}}$, we obtain the proposition.
\end{proof}

\subsubsection{Strict contractivity of the two-level system}\label{app:provable_str_contr_tls}
We next consider the Lindbladian describing a two-level system with states $\ket{e}, \ket{g}$ with a jump operator $L =\sqrt{\gamma}\sigma$ where $\sigma = \ket{g}\!\bra{e}$ and any time-dependent Hamiltonian $H(t)$:
\begin{align}\label{eq:master_equation_2ls_time_dep}
    \mathcal{L}(t) = -i[H(t), \cdot] + \gamma \mathcal{D}_\sigma.
\end{align}
We show that for this system as well, the dynamics generated by $\mathcal{L}(t)$ is strictly contractive.
\begin{proposition}
   Suppose $\mathcal{E}(t, s) = \overrightarrow{\mathcal{T}}\textnormal{exp}(\int_s^t \mathcal{L}(\tau) d\tau)$ is the channel generated by the two-level system Lindbladian in Eq.~(\ref{eq:master_equation_2ls_time_dep}) then $\forall t > s$ and any two states $\rho_1, \rho_2$,
    \[
    \smallnorm{\mathcal{E}(t, s)(\rho_1 - \rho_2)}_1 \leq e^{-\gamma \abs{t - s}/2}\smallnorm{\rho_1 - \rho_2}_1.
    \]
\end{proposition}
\begin{proof}
    Similar to the proof of proposition \ref{prop:contractivity_full_kraus}, we work in a trotterization of the dynamics: We will discretize the interval $(t, s] = (\tau_T, \tau_{T - 1}] \cup (\tau_{T - 1}, \tau_{T - 2}] \cup \dots \cup (\tau_1, \tau_0]$, where $\tau_m = s + m\varepsilon$ with $\varepsilon = (t - s)/T$. The channel $\mathcal{E}(t, s)$ will be approximated by $\hat{\mathcal{E}}$ where
    \begin{align}\label{eq:trotterization_tls}
        \hat{\mathcal{E}} = \mathcal{K} \mathcal{U}_{T} \mathcal{K}\mathcal{U}_{T - 1} \dots \mathcal{K}\mathcal{U}_1,
    \end{align}
    where $\mathcal{U}_k(X) = U_k X U_k^\dagger$ with $U_k = \overrightarrow{\mathcal{T}}\textnormal{exp}(-i\int_{\tau_{k - 1}}^{\tau_k} H(s) ds)$ and $\mathcal{K}(X) = K_0 X K_0^\dagger + K_1 X K_1^\dagger$ where
    \[
    K_1 = \sqrt{\gamma \varepsilon}\sigma = \sqrt{\gamma \varepsilon} \ket{g}\!\bra{e} \text{ and }K_0 = (1 - \gamma \varepsilon \sigma^\dagger \sigma)^{1/2} = \ket{g}\!\bra{g} + \sqrt{1 - \gamma\varepsilon}\ket{e}\!\bra{e}.
    \]
    Since $\smallnorm{\mathcal{K} - \exp(\gamma\varepsilon\mathcal{D}_{\sigma})}_\diamond \leq O(\varepsilon^2)$, it follows that $\smallnorm{\hat{\mathcal{E}} - \mathcal{E}(t, s)}_\diamond \leq O(T \varepsilon^2) \leq O(\abs{t - s}\varepsilon)$ and thus the exact continuous-time dynamics can be recovered in the limit of $\varepsilon \to 0$. Next, we establish that the channel $\mathcal{K}$ is strictly contractive --- to show this, we recall a result from Ref.~\cite{ruskai1994beyond} (theorem 2), where it was shown that for any two states $\rho_1, \rho_2$ such that $\rho_1 \neq \rho_2$,
    \begin{align}\label{eq:contr_coeff}
         \frac{\norm{\mathcal{K}(\rho_1 - \rho_2 )}_1}{\norm{\rho_1 - \rho_2}_1} \leq \frac{1}{2}\sup_{\substack{\ket{\psi_1}, \ket{\psi_2} \\ \bra{\psi_1}\psi_2\rangle = 0}}\norm{\mathcal{K}(\ket{\psi_1}\!\bra{\psi_1} - \ket{\psi_2}\!\bra{\psi_2})}.
    \end{align}
    In the other words, the states $\rho_1, \rho_2$ that maximize $\smallnorm{\mathcal{K}\big(\rho_1 - \rho_2 \big)}_1$ are pure states $\ket{\psi_1}, \ket{\psi_2}$ that are orthogonal to each other (i.e.,~$\bra{\psi_1}\psi_2 \rangle = 0$). Setting $\ket{\psi_1} = \cos \theta \ket{e} + e^{i\varphi}  \sin \theta \ket{g}$ and $\ket{\psi_2} = \sin \theta \ket{e} - e^{i\varphi}\cos \theta \ket{g}$ we obtain
    \[
    \mathcal{K}(\ket{\psi_1}\!\bra{\psi_1} - \ket{\psi_2}\!\bra{\psi_2}) = (1 - \gamma \varepsilon) \cos 2\theta \big(\ket{e}\!\bra{e} -  \ket{g}\!\bra{g}\big) + (1 - \gamma \varepsilon)^{1/2}\sin 2\theta \big(e^{i\varphi}  \ket{e}\!\bra{g} + \text{h.c.} \big),
    \]
    and therefore
    \[
    \frac{1}{2}\norm{\mathcal{K}(\ket{\psi_1}\!\bra{\psi_1} - \ket{\psi_2}\!\bra{\psi_2})}_1 = \big((1 - \gamma \varepsilon)^2 \cos^2 2\theta + (1 - \gamma \varepsilon) \sin^2 2\theta\big)^{1/2} \leq \sqrt{1 - \gamma \varepsilon}.
    \]
    Finally, using this together with the fact that unitaries leave the 1-norm invariant, we obtain that for any two states $\rho_1, \rho_2$
    \[
    \smallnorm{\hat{\mathcal{E}}(\rho_1 - \rho_2)} \leq (1 - \gamma \varepsilon)^{T/2}.
    \]
    Finally, taking the limit of $\varepsilon \to 0$ and setting $\hat{\mathcal{E}} \to \mathcal{E}(t, s)$ and $(1 - \gamma \varepsilon)^{T/2} \to e^{-\gamma \abs{t - s}/2}$, we obtain the proposition.
\end{proof}

\subsection{Adjoint-variable method for optimizing QFI}\label{app:adj_vm}
In this subsection, we will provide details of numerically optimizing the QFI with respect to the time-dependent Hamiltonian applied on the source. We will work with a discretized matrix-product state representation of the emitted photons, which we review first. Then, we will show how the computation of the gradient of the QFI with respect to the Hamiltonian parameters can be sped up using the adjoint variable method.

\subsubsection{Review: Matrix Product state representation}\label{app:mps_rep}
We first show how to approximate the state of the emitted photons with a matrix product state. We recall from the main text that the Hamiltonian of the source interacting with $d$ output ports is given by
\begin{align*}
    H(t) = H_S(t) + \sum_{j = 1}^M \big(a_{j, t}^\dagger L_j + \text{h.c.}\big),
\end{align*}
where $[a_{\alpha, t}, a_{\alpha', t'}] = \delta_{\alpha, \alpha'}\delta(t - t')$ and $H_S(t)$ and $L_\alpha$ are operators acting on the source Hilbert space. We will denote by $U(t, s)$ and $U_S(t, s)$ the unitaries generated by the Hamiltonians $H(t)$ and $H_S(t)$ respectively:
\begin{align*}
    U(t, s) =\overrightarrow{\mathcal{T}}\exp\bigg(-\int_s^t H(\tau) d\tau\bigg) \text{ and }U_S(t, s) = \overrightarrow{\mathcal{T}}\exp\bigg(-\int_s^t H_S(\tau) d\tau\bigg).
\end{align*}
Assuming that the initial and final states of the source are $\ket{\phi_{S, i}}, \ket{\phi_{S, f}}$, the state of the photons $\ket{\psi}$ in the output ports can be expressed as
\begin{align}\label{eq:state_def}
    \ket{\psi} = \frac{1}{\mathcal{N}} \langle \phi_{S, f}\ket{\Psi} \text{ where }\ket{\Psi} = U(T, 0) \ket{\text{vac}, \phi_{S, i}},
\end{align}
where $\mathcal{N} = \norm{\langle \phi_{S, f}\ket{\Psi}}$ is a normalization constant. We will approximate $\ket{\Psi}$ by $\smallket{\hat{\Psi}}$, which would be a state in the Hilbert space $(\mathbb{C}^{d + 1})^{\otimes M}$ where we will (i) effectively discretize the ports in space as well as (ii) truncate the Hilbert space dimension of each discretized step to having at-most a single-excitation, i.e.,
\begin{align}\label{eq:def_hat_Psi}
    \smallket{\hat{\Psi}} = U_M^S K_M U_{M - 1}^S K_{M - 1} \dots U_1^S K_1 \ket{\phi_{S, i}},
\end{align}
where $M$ is the number of time steps that $[0, T]$ is discretized into and
\begin{subequations}\label{eq;def_discr}
\begin{align}
    U_k^S = U_S(k\varepsilon, (k - 1)\varepsilon) \text{ and }K_k = (I - \varepsilon Q)^{1/2} \ket{0_k} - i\sqrt{\varepsilon} \sum_{j = 1}^d L_j \ket{j_{k}},
\end{align}
with $\varepsilon = T / M$ and
\begin{align}
Q = \sum_{j = 1}^d L_j^\dagger L_j.
\end{align}
\end{subequations}
Here $\ket{0_k} = \smallket{\text{vac}_{((k - 1)\varepsilon, k \varepsilon]}}$ is the vacuum in the output port corresponding to the time interval $((k - 1)\varepsilon, k \varepsilon]$ and $\ket{j_k}$ is a state with a single photon in the $j^\text{th}$ output port in the time interval $((k - 1)\varepsilon, k \varepsilon]$, i.e.,
\begin{align}
    \ket{j_\alpha} = A_{\alpha, k}^\dagger \smallket{\text{vac}_{((k - 1)\varepsilon, k \varepsilon]}} \text{ with }A_{\alpha, k} = \frac{1}{\sqrt{\varepsilon}} \int_{(k - 1)\varepsilon}^{k \varepsilon} a_{\alpha, s}ds.
\end{align}
It can be noted that $[A_{\alpha, k}, A^\dagger_{\alpha', k'}] = \delta_{\alpha, \alpha'} \delta_{k, k'}$. The approximated state of photons $\smallket{\hat{\psi}} = \bra{\phi_{S, f}} \hat{\Psi}\rangle / \hat{\mathcal{N}}$ where $\hat{\mathcal{N}} = \smallnorm{\smallbra{\phi_{S, f}} \hat{\Psi}\rangle}$.

For our analysis, it will be convenient to use the following lemma.
\begin{lemma}
    For any interval $s < \tau < t$,
    \[
    \smallbra{\textnormal{vac}_{(s, t]}}U(t, s) \smallket{\textnormal{vac}_{(s, t]}} = U_\textnormal{eff}(t, s) \text{ and }\smallbra{\textnormal{vac}_{(s, t]}}a_{\alpha, \tau} U(t, s) \smallket{\textnormal{vac}_{(s, t]}} = -iU_\textnormal{eff}(t, \tau) L_\alpha U_\textnormal{eff}(\tau, s),
    \]
    where 
    \[
    U_\textnormal{eff}(t, s) = \overrightarrow{\mathcal{T}} \exp\bigg(-i\int_{s}^t H_\textnormal{eff}(\tau) d\tau\bigg) \text{ with }H_\textnormal{eff}(t) = H_S(t) - \frac{i}{2}Q.
    \]
\end{lemma}
\begin{proof}
We show this by discretizing the interval $(s, t]$ into $(\tau_0, \tau_1] \cup (\tau_1, \tau_2] \cup \dots \cup (\tau_{N - 1}, \tau_N]$ where $\tau_k = s + k \delta$ with $\delta = (t - s)/N$. We note that an explicit second order Dyson expansion of $U(\tau_{k + 1}, \tau_k)$ yields
\begin{align*}
\smallbra{\text{vac}_{(\tau_k, \tau_{k + 1}]}}U(\tau_{k + 1}, \tau_k) \smallket{\text{vac}_{(\tau_k, \tau_{k + 1}]}}=U_\text{eff}(\tau_{k + 1}, \tau_k) + O(\delta^2).
\end{align*}
Consequently, 
\begin{align*}
    \smallbra{\textnormal{vac}_{(s, t]}}U(t, s) \smallket{\textnormal{vac}_{(s, t]}} = \prod_{k = N}^1 \smallbra{\text{vac}_{(\tau_k, \tau_{k + 1}]}} U(\tau_{k + 1}, \tau_k) \smallket{\text{vac}_{(\tau_k, \tau_{k + 1}]}} = {U}_\text{eff}(t, s) + O(\abs{t - s}\delta).
\end{align*}
Taking the limit of $\delta \to 0$ yields that $\smallbra{\textnormal{vac}_{(s, t]}}U(t, s) \smallket{\textnormal{vac}_{(s, t]}} = U_\textnormal{eff}(t, s)$. 

Next, we consider deriving an expression for $\smallbra{\textnormal{vac}_{(s, t]}}a_{\alpha, \tau} U(t, s) \smallket{\textnormal{vac}_{(s, t]}}$. We begin with the Heisenberg picture operator $a_{\alpha, \tau}(t, s) = U(s, t) a_{\alpha, \tau} U(t, s)$ which satisfies
\begin{align*}
    \frac{d}{dt}a_{\alpha, \tau}(t, s) = -i \delta(\tau - t) U(s, t) L_\alpha U(t, s).
\end{align*}
Noting that $a_{\alpha, \tau}(s, s) = a_{\alpha, \tau}$, we can integrate this equation from $s \to t$ to obtain
\[
a_{\alpha, \tau}(t, s) = a_{\alpha, \tau} -iU(s, \tau) L_\alpha U(\tau, s),
\]
or equivalently
\[
a_{\alpha, \tau} U(t, s) = U(t, s) a_{\alpha, \tau} -i U(t, \tau) L_\alpha U(\tau, s).
\]
Therefore, we obtain that
\begin{align*}
\smallbra{\text{vac}_{(s, t]}}a_{\alpha, \tau} U(t, s) \smallket{\text{vac}_{(s, t]}} &= -i\smallbra{\text{vac}_{(s, t]}}U(t, \tau) L_\alpha U(\tau, s) \smallket{\text{vac}_{(s, t]}} \nonumber \\
&=-i\smallbra{\text{vac}_{(\tau, t]}}U(t, \tau) \smallket{\text{vac}_{(\tau, t]}} L_\alpha \smallbra{\text{vac}_{(s, \tau]}}U(\tau, s) \smallket{\text{vac}_{(s, \tau]}} \nonumber \\
&=-i U_\text{eff}(t, \tau) L_\alpha U_\text{eff}(\tau, s),
\end{align*}
which proves the lemma.
\end{proof}

We next provide an error bound between the exact state and its MPS representation.
\begin{proposition}[Error bound on MPS representation]
    Supppose $J, \ell > 0$ are positive real numbers such that $\forall t \geq 0: \norm{H_S(t)} \leq J$ and $\sum_{\alpha = 1}^d \norm{L_\alpha}^2 \leq \ell$, and $\eta_0 = 2\ell(J + 2\ell)$ then if $\varepsilon < \mathcal \mathcal{N}^2 / \eta_0 T$
\[
\smallnorm{\ket{\psi}\!\bra{\psi} - \smallket{\hat{\psi}}\!\smallbra{\hat{\psi}}} \leq \bigg(1 + \frac{1}{(\mathcal{N}^2 - \eta_0 T\varepsilon)}\bigg) \frac{\eta_0 T \varepsilon}{\mathcal{N}^2} \leq O\bigg(\frac{\eta_0 T\varepsilon}{\mathcal{N}^2}\bigg).
\]
\end{proposition}
\begin{proof}
    It will be convenient to define $J, \ell$ via
    \[
    J = \max_{t \in [0, T]}\norm{H_S(t)} \text{ and }\ell = \sum_{\alpha = 1}^d \norm{L_\alpha}^2.
    \]
    It can be noted that $\norm{Q} \leq \ell$, where $Q$ is defined in Eq.~(\ref{eq;def_discr}b). 
    
    We will first show that $\smallbra{\hat{\Psi}} \Psi \rangle$ is close to 1, where $\ket{\Psi}$ and $\smallket{\hat{\Psi}}$ are defined in Eqs.~(\ref{eq:state_def}) and (\ref{eq:def_hat_Psi}) respectively. We begin by rewriting $\smallbra{\hat{\Psi}} \Psi \rangle$ as
    \begin{align*}
        \smallbra{\hat{\Psi}} \Psi \rangle = \bra{\phi_{S, i}}  {V}_1^\dagger  \dots {V}_{M - 1}^\dagger {V}_M^\dagger U^S_M K_M U^S_{M - 1} K_{M - 1} \dots U^S_1 K_1\ket{\phi_{S, i}},
    \end{align*}
    where $V_m = U(m\varepsilon, (m - 1)\varepsilon) \smallket{0_m}$. It will be convenient to define
    \begin{align*}
        \smallket{\hat{\Psi}_m} = U_m^S K_{m} U_{m - 1}^S K_{m - 1} \dots U_1^S K_1 \ket{\phi_{S, i}} \text{ and }\smallket{{\Psi}_m} = V_m V_{m -1} \dots V_1 \ket{\phi_{S, i}}.
    \end{align*}
    It can be noted that $\ket{\Psi_M} = \ket{\Psi}$ and $\smallket{\hat{\Psi}_M} = \smallket{\hat{\Psi}}$. Furthermore, it is straightfoward to see that $\norm{\ket{\Psi_m}} = 1$. Additionally, since $\sum_{\alpha}K_\alpha^\dagger K_\alpha = I$ (with the identity being on the system operator), it also follows that $\smallnorm{\smallket{\hat{\Psi}_m}} = 1$.
    
    We will also define $T_m = V_m^\dagger U^S_m K_m$: Note that $T_m$ is an operator acting on the Hilbert space of the source. Furthermore, we obtain the recursive inequality
    \begin{align*}
        \smallabs{\smallbra{\hat{\Psi}} \Psi \rangle - \smallbra{\hat{\Psi}_{m - 1}} \Psi_{m - 1}\rangle} \leq \smallabs{\smallbra{\Psi_{m - 1}}(T_m - I) \smallket{\hat{\Psi}_{m - 1}}}  \leq \norm{T_m - I},
    \end{align*}
    and therefore
    \begin{align}\label{eq:telescoping}
        \smallabs{\smallbra{\hat{\Psi}_m} \Psi_m \rangle - 1} \leq \sum_{m = 1}^M \norm{T_m - I}.
    \end{align}
    Thus, it is sufficient to bound $\norm{T_m - I}$. This can be done with an application of time-dependent perturbation theory. Note that
    \begin{align}\label{eq:T_m_expression}
        T_m  &= \smallbra{0_m} U^\dagger(m\varepsilon, (m - 1)\varepsilon) \ket{0_m} U_m^S (I - \varepsilon Q)^{1/2} - i\sqrt{\varepsilon}\sum_{\alpha = 1}^d \smallbra{0_m}U^\dagger(m\varepsilon, (m - 1)\varepsilon) A^\dagger_{\alpha, m}\ket{0_m} U^S_m L_\alpha \nonumber \\
        &=U^\dagger_\text{eff}(m\varepsilon, (m - 1)\varepsilon)U_m^S (I - \varepsilon Q)^{1/2} +\sum_{\alpha = 1}^d \int_{(m - 1)\varepsilon}^{m \varepsilon}U_\text{eff}^\dagger(s, (m - 1)\varepsilon) L_\alpha^\dagger U_\text{eff}^\dagger(m\varepsilon, s) U^S_m L_\alpha.
    \end{align}
    Next, we observe that since $\sqrt{1 - x} = 1 - \frac{x}{2} - \frac{x^2}{4(1 - x)^{1/2} + (2 - x)^2}$, $(I - \varepsilon Q)^{1/2} \approx I - \varepsilon Q/2$ and for $\varepsilon \leq 1/\ell$
    \begin{align}\label{eq:upper_bound_term_sqrt}
        \norm{(I - \varepsilon Q)^{1/2} - \bigg(I - \frac{\varepsilon Q}{2}\bigg)} \leq \varepsilon^2\smallnorm{(4(I - \varepsilon Q)^{1/2} + (2 - \varepsilon Q)^2)^{-1}Q^2} \leq \varepsilon^2  \ell^2.
    \end{align}
    Furthermore, we note that
    \begin{align*}
        &(U_m^S)^\dagger U_\text{eff}(m\varepsilon, (m - 1)\varepsilon) \nonumber\\
        &\qquad= I + \int_{(m - 1)\varepsilon}^{m\varepsilon}\frac{d}{dt'} U_S((m - 1)\varepsilon, t') U_\text{eff}(t', (m - 1)\varepsilon) dt' \nonumber \\
        &\qquad= I -\frac{1}{2} \int_{(m - 1)\varepsilon}^{m\varepsilon} U_S((m -1)\varepsilon, t') Q U_\text{eff}(t', (m - 1)\varepsilon) dt' \nonumber \\
        &\qquad=I - \frac{\varepsilon}{2}Q - \frac{1}{2} \int_{(m - 1)\varepsilon}^{m\varepsilon} \int_{(m - 1)\varepsilon}^{t'} \frac{d}{dt''}U_S((m - 1)\varepsilon, t'') Q U_\text{eff}(t'', (m - 1)\varepsilon) dt'' dt' \nonumber \\
        &\qquad=I - \frac{\varepsilon}{2}Q - \frac{i}{2} \int_{(m -1)\varepsilon}^{m\varepsilon} \int_{(m - 1)\varepsilon}^{t'} U_S((m - 1)\varepsilon, t'')(H_S(t'') Q - Q H_\text{eff}(t'') )U_\text{eff}(t'', (m - 1)\varepsilon)dt'' dt'.
    \end{align*}
    Thus, for small $\varepsilon$, $(U_m^S)^\dagger U_\text{eff}(m\varepsilon, (m - 1)\varepsilon) \approx I - \varepsilon Q/2$ with the equation above implying that
    \begin{align}\label{eq:upper_bound_term_1}
        \norm{(U_m^S)^\dagger U_\text{eff}(m\varepsilon, (m - 1)\varepsilon) - \bigg(I -\frac{\varepsilon Q}{2}\bigg)} \leq \frac{\varepsilon^2}{4}\ell\bigg(2J + \frac{1}{2}\ell\bigg),
    \end{align}
    where we have used $\norm{H_S(t)} \leq J, \norm{H_\text{eff}(t)} \leq J + \norm{Q}/2$, $\norm{U_S(t, s)} = 1$ and $\norm{U_\text{eff}(t, s)} \leq 1$.
    
    Finally, we note that
    \begin{align*}
        &U^\dagger_\text{eff}(s, (m - 1)\varepsilon) L_\alpha^\dagger U_\text{eff}^\dagger(m\varepsilon, s) \nonumber\\
        &= L_\alpha^\dagger U^\dagger_\text{eff}(m\varepsilon, s) +i \bigg(\int_{(m - 1)\varepsilon}^{s}U_\text{eff}^\dagger(s', (m - 1)\varepsilon) H_\text{eff}^\dagger(s')ds'\bigg)L_\alpha^\dagger U_\text{eff}^\dagger(m\varepsilon, s) \nonumber\\
        &=L_\alpha^\dagger -iL_\alpha^\dagger\bigg(\int_{m\varepsilon}^{s} H_\text{eff}^\dagger(s')U_\text{eff}^\dagger(m\varepsilon, s') ds'\bigg) +i \bigg(\int_{(m - 1)\varepsilon}^{s}U_\text{eff}^\dagger(s', (m - 1)\varepsilon) H_\text{eff}^\dagger(s')ds'\bigg)L_\alpha^\dagger U_\text{eff}^\dagger(m\varepsilon, s),
    \end{align*}
    and therefore $U_\text{eff}^\dagger(s, (m - 1)\varepsilon) L_\alpha^\dagger U_\text{eff}^\dagger(m\varepsilon, s) \approx L_\alpha^\dagger$ with the equation above implying that
    \begin{align}\label{eq:upper_bound_term_2}
        \smallnorm{U^\dagger_\text{eff}(s, (m - 1)\varepsilon) L_\alpha^\dagger U_\text{eff}^\dagger(m\varepsilon, s)  - L_\alpha^\dagger} \leq \varepsilon \norm{L_\alpha} \bigg(J + \frac{\ell}{2}\bigg).
    \end{align}
    where we have used $\norm{U_\text{eff}(t, s)} \leq 1$ and $\norm{H_\text{eff}(t)} \leq J + \norm{Q}/2$. Using Eqs.~(\ref{eq:upper_bound_term_sqrt}), (\ref{eq:upper_bound_term_1}) and (\ref{eq:upper_bound_term_2}) together with Eq.~(\ref{eq:T_m_expression}), we obtain that for $\varepsilon \norm{Q} \leq 1$
    \begin{align}\label{eq:T_upper_bound}
        \norm{T_m - I} \leq \frac{5}{4}\varepsilon^2 \ell^2 + \varepsilon^2\bigg(J + \frac{\ell}{2}\bigg)\ell + \frac{\varepsilon^2}{4}\ell\bigg(2J + \frac{1}{2}\ell\bigg) \leq \frac{1}{2}\eta_0 \varepsilon^2,
    \end{align}
    where $\eta_0 = 2\ell(2\ell + J)$. Thus, from Eq.~(\ref{eq:telescoping}) that
    \[
    \smallabs{\bra{\Psi}\hat{\Psi}\rangle - 1} \leq \frac{1}{2}\eta_0 M \varepsilon^2 \text{ or equivalently }\smallnorm{\ket{\Psi} - \smallket{\hat{\Psi}}}^2 = 2\text{Re}\big(1 - \bra{\Psi}\hat{\Psi}\rangle\big) \leq  \eta_0 M \varepsilon^2.
    \]
    Next, we consider the normalized states, $\ket{\psi} = \langle \phi_{S, f}\ket{\Psi}/\mathcal{N} $ and $\smallket{\hat{\psi}} = \langle \phi_{S, f}\smallket{\hat{\Phi}} / \hat{\mathcal{N}}$. We first note that
    \[
    \smallabs{\mathcal{N}^2 - \hat{\mathcal{N}}^2} = \smallabs{\text{Tr}\big(\ket{\phi_{S, f}}\!\bra{\phi_{S, f}}(\ket{\Psi}\!\bra{\Psi} - \smallket{\hat{\Psi}}\!\smallbra{\hat{\Psi}})\big)} \leq \smallnorm{\ket{\Psi}\!\bra{\Psi} - \smallket{\hat{\Psi}}\!\smallbra{\hat{\Psi}}}_1 \leq  \eta_0 M\varepsilon^2.
    \]
    Note that this implies that $\hat{\mathcal{N}}^2 \geq \mathcal{N}^2 - \eta_0 M\varepsilon^2$. Consequently, we then obtain that if $M\varepsilon^2 < \mathcal{N}^2 / \eta_0$,
    \[
    \abs{\frac{1}{\mathcal{N}^2} - \frac{1}{\hat{\mathcal{N}}^2}} = \frac{\smallabs{\mathcal{N}^2 - \hat{\mathcal{N}}^2}}{\mathcal{N}^2\hat{\mathcal{N}^2}} \leq \frac{\eta_0 M \varepsilon^2}{\mathcal{N}^2(\mathcal{N}^2 - \eta_0 M\varepsilon^2)}.
    \]
    Finally,
    \begin{align}
        \smallnorm{\ket{\psi}\!\bra{\psi} - \smallket{\hat{\psi}}\!\smallbra{\hat{\psi}}} \leq \frac{1}{\mathcal{N}^2}\smallnorm{\smallket{\Psi}\!\smallbra{\Psi} - \smallket{\hat{\Psi}}\!\smallbra{\hat{\Psi}} }_1 + \abs{\frac{1}{\mathcal{N}^2} - \frac{1}{\hat{\mathcal{N}}^2}}  \leq \bigg(1 + \frac{1}{(\mathcal{N}^2 - \eta_0 M\varepsilon^2)}\bigg) \frac{\eta_0 M \varepsilon^2}{\mathcal{N}^2},
    \end{align}
    which proves the lemma.
\end{proof}

\subsubsection{Gradient computation}
We next show how to compute the QFI and the QFI gradient with respect to the system unitaries applied in every time-step. To make the computations easy to visualize, we will use the tensor network diagrammatic notation. 

\begin{align}
 \ket{\psi} = \frac{1}{\mathcal{N}}
    \begin{array}{c}
        \begin{tikzpicture}[scale=0.6]
        \bTensor{(-\dx*0, 0)}{\small $\phi_i$}{1}
        \KTensor{(-\dx*2,0)}{\small $K$}
        \KTensor{(-\dx*4,0)}{\small $K$}
        \SingleDots{-\dx*6,0}{\dx/2}
        \KTensor{(-\dx*7,0)}{\small $K$}
        \bTensor{(-\dx*9, 0)}{\small $\phi_f^*$}{-1}
        \gate{(-\dx*3,0)}{\small $U_1^S$}
        \gate{(-\dx*1,0)}{\small $U_0^S$}
        \gate{(-\dx*8,0)}{\small $U_M^S$}
        \gate{(-\dx*5,0)}{\small $U_2^S$}
        \end{tikzpicture}
    \end{array},
\end{align}
where
\begin{align}
    \begin{array}{c}
        \begin{tikzpicture}[scale=0.6]
        \draw(-0.35, 0.9) node {$\alpha$};
        \KTensor{(0, 0)}{\small $K$}{0} 
        \end{tikzpicture}
    \end{array} = \begin{cases}
        (I - \varepsilon Q)^{1/2}  & \text{ if }\alpha = 0, \\
        -i\sqrt{\varepsilon} L & \text{ if } \alpha = 1.
    \end{cases}
\end{align}
We will find it convenient to introduce $\text{E}_{m}^n, \text{R}_m^n, \text{L}_m^n$ for $0 \leq m < n \leq M$ given by
\begin{subequations}
\begin{align}
    \begin{array}{c}
    \begin{tikzpicture}[scale=0.6]
        \CTensor{(0, 0)}{\small $\text{E}_{m}^{n}$}
    \end{tikzpicture}
    \end{array} &= 
    \begin{array}{c}
    \begin{tikzpicture}[scale=0.6]
        \UTensor{(0, 0)}{\small $U_m^S$}{\small $U_m^{S*}$};
        \TTensor{(-\singledx, 0)}{\small $K$}{\small $K^*$};
        \UTensor{(-\singledx*2, 0)}{\small $U_{m + 1}^S$}{\small $U_{m + 1}^{S*}$};
        \TTensor{(-\singledx*3, 0)}{\small $K$}{\small $K^*$};
        \DoubleDots{(-\singledx*4, 0)}{\singledx/2}
        \UTensor{(-\singledx*5, 0)}{\small $U_{n- 1}^S$}{\small $U_{n -1}^{S*}$};
        \TTensor{(-\singledx*6, 0)}{\small $K$}{\small $K^*$};
        \UTensor{(-\singledx*7, 0)}{\small $U_{n}^S$}{\small $U_{n}^{S*}$};
    \end{tikzpicture}
    \end{array}, \nonumber \\
        \begin{array}{c}
    \begin{tikzpicture}[scale=0.6]
        \CTensor{(0, 0)}{\small $\text{R}_{m}^{n}$}
    \end{tikzpicture}
    \end{array} &= 
    \begin{array}{c}
    \begin{tikzpicture}[scale=0.6]
        \UTensor{(0, 0)}{\small $U_m^S$}{\small $U_m^{S*}$};
        \TTensor{(-\singledx, 0)}{\small $K$}{\small $K^*$};
        \UTensor{(-\singledx*2, 0)}{\small $U_{m + 1}^S$}{\small $U_{m + 1}^{S*}$};
        \TTensor{(-\singledx*3, 0)}{\small $K$}{\small $K^*$};
        \DoubleDots{(-\singledx*4, 0)}{\singledx/2}
        \UTensor{(-\singledx*5, 0)}{\small $U_{n}^S$}{\small $U_{n}^{S*}$};
        \TTensor{(-\singledx*6, 0)}{\small $K$}{\small $K^*$};
    \end{tikzpicture}
    \end{array}, \nonumber \\
        \begin{array}{c}
    \begin{tikzpicture}[scale=0.6]
        \CTensor{(0, 0)}{\small $\text{L}_{m}^{n}$}
    \end{tikzpicture}
    \end{array} &= 
    \begin{array}{c}
    \begin{tikzpicture}[scale=0.6]
        \TTensor{(0, 0)}{\small $K$}{\small $K^*$};
        \UTensor{(-\singledx, 0)}{\small $U_m^S$}{\small $U_m^{S*}$};
        \TTensor{(-\singledx*2, 0)}{\small $K$}{\small $K^*$};
        \UTensor{(-\singledx*3, 0)}{\small $U_{m + 1}^S$}{\small $U_{m + 1}^{S*}$};
        \DoubleDots{(-\singledx*4, 0)}{\singledx/2}
        \TTensor{(-\singledx*5, 0)}{\small $K$}{\small $K^*$};
        \UTensor{(-\singledx*6, 0)}{\small $U_{n}^S$}{\small $U_{n}^{S*}$};
    \end{tikzpicture}
    \end{array},
\end{align}
\end{subequations}
where for $n < m$ we interpret $\text{E}_m^n, \text{R}_m^n , \text{L}_m^n  = I\otimes I $. We can also note that for any $m \leq p \leq n$,
\begin{align}\label{eq:decomp_E_L_R}
     \begin{array}{c}
    \begin{tikzpicture}[scale=0.6]
        \CTensor{(0, 0)}{\small $\text{E}_{m}^{n}$}
    \end{tikzpicture}
    \end{array} = \begin{array}{c}
    \begin{tikzpicture}[scale=0.6]
        \CTensor{(0, 0)}{\small $\text{R}_{m}^{p - 1}$};
        \UTensor{(-\singledx, 0)}{\small $U_p^{S*}$}{\small $U_p^S$};
        \CTensor{(-\singledx * 2, 0}{\small $\text{L}_{p + 1}^{n}$}
    \end{tikzpicture}
    \end{array} 
\end{align}
The normalization constant $\mathcal{N}$ is given by the diagrammatic expression
\begin{align}
    \mathcal{N}^2 = \begin{array}{c}
    \begin{tikzpicture}[scale=0.6]
    \BTensor{(0, 0)}{\small $\phi_i$}{\small $\phi_i^*$}{1};
    \CTensor{(-\singledx, 0)}{\small $\text{E}_{M, 0}$};
    \BTensor{(-\singledx*2, 0)}{\small $\phi_f^*$}{\small $\phi_f$}{-1}
    \end{tikzpicture}
    \end{array}
\end{align}
We will consider the case of identical and independent sources --- in terms of the MPS approximation of the state, the QFI from Eq.~(\ref{eq:QFI_identical_sources}) can be approximated by
\begin{subequations}\label{eq:qfi_identical_numerical}
\begin{align}
\text{QFI} \approx \frac{16 }{\mathcal{N}^4}\sum_{n = 1}^M \sum_{m = 1}^{n - 1} \bigg(\smallabs{C^{(g)}_{n, m}}^2 - \smallabs{C^{(\chi)}_{n, m}}^2\bigg) + \frac{8}{\mathcal{N}^2} \sum_{m = 1}^M {n_m},
\end{align}
where 
\begin{align}
    C_{n, m}^{(g)} &= 
    \begin{array}{c}
    \begin{tikzpicture}[scale=0.6]
    \BTensor{(0, 0)}{\small $\phi_i$}{\small $\phi_i^*$}{1};
    \BTensor{(-\singledx*6, 0)}{\small $\phi_f^*$}{\small $\phi_f$}{-1};
    \KGate{(-\singledx*2, 0)}{\small $K_0$}{\small $K_1^*$};
    \KGate{(-\singledx*4, 0)}{\small $K_1$}{\small $K_0^*$};
    \CTensor{(-\singledx, 0)}{\small $\text{E}_{0}^{m - 1}$};
    \CTensor{(-\singledx*3, 0)}{\small $\text{E}_{m}^{n - 1}$};
    \CTensor{(-\singledx*5, 0)}{\small $\text{E}_{n}^{M}$};
    \end{tikzpicture}
    \end{array}, \\
    C_{n, m}^{(\chi)} &= \begin{array}{c}
    \begin{tikzpicture}[scale=0.6]
    \BTensor{(0, 0)}{\small $\phi_i$}{\small $\phi_i^*$}{1};
    \BTensor{(-\singledx*6, 0)}{\small $\phi_f^*$}{\small $\phi_f$}{-1};
    \KGate{(-\singledx*2, 0)}{\small $K_0$}{\small $K_1^*$};
    \KGate{(-\singledx*4, 0)}{\small $K_0$}{\small $K_1^*$};
    \CTensor{(-\singledx, 0)}{\small $\text{E}_{0}^{m - 1}$};
    \CTensor{(-\singledx*3, 0)}{\small $\text{E}_{m}^{n - 1}$};
    \CTensor{(-\singledx*5, 0)}{\small $\text{E}_{n}^{M}$};
    \end{tikzpicture}
    \end{array},\\
    n_m &= \begin{array}{c}
    \begin{tikzpicture}[scale=0.6]
    \BTensor{(0, 0)}{\small $\phi_i$}{\small $\phi_i^*$}{1};
    \BTensor{(-\singledx*4, 0)}{\small $\phi_f^*$}{\small $\phi_f$}{-1};
    \KGate{(-\singledx*2, 0)}{\small $K_1$}{\small $K_1^*$};
    \CTensor{(-\singledx, 0)}{\small $\text{E}_{0}^{m - 1}$};
    \CTensor{(-\singledx*3, 0)}{\small $\text{E}_{m}^{M}$};
    \end{tikzpicture}
    \end{array},
\end{align}
\end{subequations}
We note that Eq.~(\ref{eq:qfi_identical_numerical}) can be used to numerically compute the QFI in $O(M^2)$ time. To see this, note that $\{\text{E}_{n, m}\}_{0\leq m\leq n\leq M}$ can be pre-computed in $O(M^2)$ time. After this pre-computation, $\mathcal{N}$, $C^{(g)}_{n, m}, C^{(\chi)}_{n, m}$ and $n_m$ can be computed in $O(1)$ time. Therefore, the total time required to compute the QFI scales as $O(M^2)$.

We next consider the computation of the gradient of the QFI with respect to the unitaries $U_p^S$ --- if the gradient is computed by first computing $\nabla_{U_p^S} C_{n, m}^{(g)}$ and $\nabla_{U_p^S} C_{n, m}^{(\chi)}$, then the total time required to compute the gradient of QFI with respect to all the unitaries $U_p^S$ would be $O(M^3)$. In practice, this can become prohibitively expensive when $M \gtrsim 100$. We next outline an application of the adjoint variable method \cite{giles2000introduction, givoli2021tutorial, white2022enhancing} that can reduce the cost of computing $\nabla_{U_p^S}\text{QFI}$ to time $O(M^2)$. We will formulate the adjoint variable method for computing the gradient of both the double summand and the single summand in Eq.~(\ref{eq:qfi_identical_numerical}), both of which require slightly different treatment.

\emph{Cost functions that are single sums}. We first consider a cost function $\Gamma^{(1)}(\theta)$ which depends on the parameters $\{\theta_0, \theta_1 \dots \theta_M\}$ through the unitaries $U^S = \{U^S_0(\theta_0), U^S_1(\theta_1) \dots U^S_M(\theta_M)\}$ and can be expressed as a single summation:
\begin{align}
    \Gamma^{(1)}(\theta) = \sum_{m = 1}^N f(\gamma_m), 
\end{align}
where $f:\mathbb{C} \to \mathbb{R}$ is Wirtinger differentiable (i.e.,~differentiable with respect to the real and imaginary parts of the complex-value argument), and $\gamma_m \in \mathbb{C}$ is given by the following contraction:
\begin{align}\label{eq:gamma_m}
    \gamma_m = 
    \begin{array}{c}
    \begin{tikzpicture}[scale=0.6]
    \BTensor{(0, 0)}{\small $\phi_i$}{\small $\phi_i^*$}{1};
    \BTensor{(-\singledx*4, 0)}{\small $\phi_f^*$}{\small $\phi_f$}{-1};
    \KGate{(-\singledx*2, 0)}{\small $Q$}{\small $P^*$};
    \CTensor{(-\singledx, 0)}{\small $\text{E}_{0}^{m - 1}$};
    \CTensor{(-\singledx*3, 0)}{\small $\text{E}_{m}^{M}$};
    \end{tikzpicture}
    \end{array},
\end{align}
for some $P$ and $Q$. Next, we note that
\begin{subequations}\label{eq:deriv_gamma_n}
\begin{align}
    \frac{\partial}{\partial \theta_p} \Gamma(\theta) =2 \sum_{m = 1}^N \text{Re}\bigg(f'(\gamma_m)\frac{\partial \gamma_m}{\partial \theta_p} \bigg),
\end{align}
where $f' \equiv \partial f(z, z^*)/ \partial z$. Note that $\partial \gamma_m / \partial \theta_p$ will be given by
\begin{align}
     \frac{\partial \gamma_m}{\partial \theta_p} &= \begin{array}{c}
    \begin{tikzpicture}[scale=0.6]
    \BTensor{(0, 0)}{\small $\phi_i$}{\small $\phi_i^*$}{1};
    \CTensor{(-\singledx, 0)}{\small $\text{R}_{0}^{p - 1}$};
    \UDTensor{(-\singledx*2, 0}{\large $\frac{\partial \mathcal{U}^S_p}{\partial \theta_p}$};
    \CTensor{(-\singledx*3, 0)}{\small $\text{L}_{p+1}^{m - 1}$};
    \KGate{(-\singledx*4, 0)}{\small $Q$}{\small $P^*$};
        \CTensor{(-\singledx*5, 0)}{\small $\text{E}_{m}^{M}$};
    \BTensor{(-\singledx*6, 0)}{\small $\phi_f^*$}{\small $\phi_f$}{-1};
    \end{tikzpicture}
    \end{array} \text{     if } 0\leq p < m \text{ and}, \nonumber \\
     \frac{\partial \gamma_m}{\partial \theta_p} &= \begin{array}{c}
    \begin{tikzpicture}[scale=0.6]
    \BTensor{(0, 0)}{\small $\phi_i$}{\small $\phi_i^*$}{1};
    \CTensor{(-\singledx, 0)}{\small $\text{E}_{0}^{m - 1}$};
    \KGate{(-\singledx*2, 0)}{\small $Q$}{\small $P^*$};
    \CTensor{(-\singledx*3, 0)}{\small $\text{R}_{m}^{p - 1}$};
    \UDTensor{(-\singledx*4, 0}{\large $\frac{\partial \mathcal{U}^S_p}{\partial \theta_p}$};
    \CTensor{(-\singledx*5, 0)}{\small $\text{L}_{p + 1}^{M}$};
    \BTensor{(-\singledx*6, 0)}{\small $\phi_f^*$}{\small $\phi_f$}{-1};
    \end{tikzpicture}
    \end{array}
    \text{     if }m\leq p \leq M,
\end{align}
where
\begin{align}
    \begin{array}{c}
        \begin{tikzpicture}[scale=0.6]
             \UDTensor{(0, 0}{\large $\frac{\partial \mathcal{U}^S_p}{\partial \theta_p}$};
        \end{tikzpicture}
    \end{array} = \frac{\partial}{\partial \theta_p}\left(\begin{array}{c}
        \begin{tikzpicture}[scale=0.6]
             \UTensor{(0, 0}{\small $U_p^S$}{\small $U_p^{S*}$};
        \end{tikzpicture}
    \end{array}
    \right).
\end{align}
\end{subequations}
While Eq.~(\ref{eq:deriv_gamma_n}) can be used to compute $\partial \Gamma^{(1)}(\theta) / \partial \theta_p$, the total cost to compute this derivative for all $p$ would be $O(M^2)$. To compute these derivatives in time $O(M)$, we first obtain from Eq.~(\ref{eq:deriv_gamma_n}a) that
\begin{subequations}\label{eq:deriv_single_sum_adj}
\begin{align}
    \frac{\partial}{\partial \theta_p}\Gamma^{(1)}(\theta) = 2\text{Re}\bigg(\sum_{m = p + 1}^M f'(\gamma_m)\frac{\partial \gamma_m}{\partial \theta_p} \bigg) + 2\text{Re}\bigg(\sum_{m = 1}^p f'(\gamma_m)\frac{\partial \gamma_m}{\partial \theta_p} \bigg),
\end{align}
and using Eq.~(\ref{eq:deriv_gamma_n}b), it follows that
\begin{align}
    \sum_{m = p + 1}^M f'(\gamma_m)\frac{\partial \gamma_m}{\partial \theta_p} = 
    \begin{array}{c}
        \begin{tikzpicture}[scale=0.6]
            \BTensor{(0, 0)}{\small $\phi_i$}{\small $\phi_i^*$}{1};
            \CTensor{(-\singledx, 0)}{\small $\text{R}_{0}^{p - 1}$};
            \UDTensor{(-\singledx*2, 0}{\large $\frac{\partial \mathcal{U}^S_p}{\partial \theta_p}$};
            \sidetensorleft{(-\singledx*3, 0)}{\small $l_p$};
        \end{tikzpicture}
    \end{array}, 
\end{align}
where
\begin{align}
    \begin{array}{c}
        \begin{tikzpicture}[scale=0.6]
            \sidetensorleft{(0, 0)}{\small $l_p$};
        \end{tikzpicture}
    \end{array} = \sum_{m = p + 1}^M f'(\gamma_m) \left(\begin{array}{c}
    \begin{tikzpicture}[scale=0.6]
    \CTensor{(0, 0)}{\small $\text{L}_{p+1}^{m - 1}$};
    \KGate{(-\singledx*1, 0)}{\small $Q$}{\small $P^*$};
        \CTensor{(-\singledx*2, 0)}{\small $\text{E}_{m}^{M}$};
    \BTensor{(-\singledx*3, 0)}{\small $\phi_f^*$}{\small $\phi_f$}{-1};
    \end{tikzpicture}
    \end{array}\right),
\end{align}
and 
\begin{align}
    \sum_{m = 1}^p f'(\gamma_m)\frac{\partial \gamma_m}{\partial \theta_p} = 
    \begin{array}{c}
        \begin{tikzpicture}[scale=0.6]
            \sidetensorright{(0, 0)}{\small $r_p$};
            \UDTensor{(-\singledx, 0}{\large $\frac{\partial \mathcal{U}^S_p}{\partial \theta_p}$};
            \CTensor{(-\singledx*2, 0)}{\small $\text{L}_{p + 1}^{M}$};
            \BTensor{(-\singledx*3, 0)}{\small $\phi_f^*$}{\small $\phi_f$}{-1};            
        \end{tikzpicture}
    \end{array}, 
\end{align}
where
\begin{align}
    \begin{array}{c}
        \begin{tikzpicture}[scale=0.6]
            \sidetensorright{(0, 0)}{\small $r_p$};
        \end{tikzpicture}
    \end{array} = \sum_{m = 1}^p f'(\gamma_m) \left(\begin{array}{c}
    \begin{tikzpicture}[scale=0.6]
    \BTensor{(0, 0)}{\small $\phi_i$}{\small $\phi_i^*$}{1};
    \CTensor{(-\singledx, 0)}{\small $E_0^{m - 1}$};
    \KGate{(-\singledx*2, 0)}{\small $Q$}{\small $P^*$};
    \CTensor{(-\singledx*3, 0)}{\small $R_m^{p - 1}$};
    \end{tikzpicture}
    \end{array}\right).
\end{align}
\end{subequations}
If the vectors $l_p, r_p$ are pre-computed, then Eqs.~(\ref{eq:deriv_single_sum_adj}a), (\ref{eq:deriv_single_sum_adj}b) and (\ref{eq:deriv_single_sum_adj}d) can be used to compute $\partial \Gamma(\theta)/ \partial \theta_p$ for all $p$ in $O(M)$ time. Furthermore, $l_p, r_p$ can be computed, for all $p$, in a total of $O(M)$ time by noting that they satisfy the following recursion:
\begin{subequations}\label{eq:adjoint_single_sum_recursion}
\begin{align}
    \begin{array}{c}
        \begin{tikzpicture}[scale=0.6]
            \sidetensorleft{(0, 0)}{\small $l_{p - 1}$};
        \end{tikzpicture}
    \end{array} &= 
   \begin{array}{c}
        \begin{tikzpicture}[scale=0.6]
            \TTensor{(0, 0)}{\small $K$}{\small $K^*$};
            \UTensor{(-\singledx, 0)}{\small $U_p^S$}{\small $U_p^{S*}$}
            \sidetensorleft{(-\singledx*2, 0)}{\small $l_{p}$};
        \end{tikzpicture}
    \end{array} +  f'(\gamma_p)\left(
        \begin{array}{c}
        \begin{tikzpicture}[scale=0.6]
            \KGate{(0, 0)}{\small $Q$}{\small $P^*$};
            \CTensor{(-\singledx, 0)}{\small $E_p^M$};
            \BTensor{(-\singledx*2, 0)}{\small $\phi_f^*$}{\small $\phi_f$}{-1};
        \end{tikzpicture}
    \end{array} 
    \right), \\
        \begin{array}{c}
        \begin{tikzpicture}[scale=0.6]
            \sidetensorright{(0, 0)}{\small $r_{p}$};
        \end{tikzpicture}
    \end{array} &=     \begin{array}{c}
         \begin{tikzpicture}[scale=0.6]
             \sidetensorright{(0, 0)}{\small $r_{p }$};
             \UTensor{(-\singledx, 0)}{\small $U_{p}^S$}{\small $U_{p - 1}^{S*}$};
             \TTensor{(-\singledx*2, 0)}{\small $K$}{\small $K^*$};
         \end{tikzpicture} 
    \end{array} + f'(\gamma_p) \left(
    \begin{array}{c}
        \begin{tikzpicture}[scale=0.6]
            \BTensor{(0, 0)}{\small $\phi_i$}{\small $\phi_i^*$}{1}
            \CTensor{(-\singledx, 0)}{\small $E_0^{p - 1}$};
            \KGate{(-\singledx*2, 0)}{\small $Q$}{\small $P^*$};
        \end{tikzpicture}
    \end{array}
    \right),
\end{align}
\end{subequations}
together with the boundary conditions $l_M = r_0 = 0$. Therefore, to compute $\{\partial \Gamma(\theta) / \partial \theta_p\}_{p \in \{0, 1, 2 \dots M\}}$, an algorithm with $O(M)$ run-time is as follows:
\begin{enumerate}
    \item[(1)] First compute $\{\gamma_m\}_{m\in \{1, 2 \dots M\}}$ --- this can be done in $O(M)$ time by recursively computing $\text{E}_0^{m - 1}$ and $\text{E}_m^M$ for $m \in \{1, 2 \dots M\}$ using
    \begin{align}
        \begin{array}{c}
        \begin{tikzpicture}[scale=0.6]
            \CTensor{(0, 0)}{\small $E_{m - 1}^M$};
        \end{tikzpicture}
    \end{array} = \begin{array}{c}
        \begin{tikzpicture}[scale=0.6]
            \CTensor{(0, 0)}{\small $E_{m}^M$};
            \TTensor{(\singledx, 0)}{\small $K$}{\small $K^*$};
            \UTensor{(\singledx*2, 0)}{\small $U_{m - 1}^S$}{\small $U_{m - 1}^{S*}$}
        \end{tikzpicture}
    \end{array} \text{ and }
        \begin{array}{c}
        \begin{tikzpicture}[scale=0.6]
            \CTensor{(0, 0)}{\small $E_{0}^{m}$};
        \end{tikzpicture}
    \end{array} = \begin{array}{c}
        \begin{tikzpicture}[scale=0.6]
            \CTensor{(0, 0)}{\small $E_{0}^{m - 1}$};
            \TTensor{(-\singledx, 0)}{\small $K$}{\small $K^*$};
            \UTensor{(-\singledx*2, 0)}{\small $U_{m }^S$}{\small $U_{m}^{S*}$}
        \end{tikzpicture}
    \end{array} ,
    \end{align}
    together with $E^{M}_{M + 1} = E^{-1}_0 = I$ and then using Eq.~(\ref{eq:gamma_m}).
    \item[(2)] Next, we compute $\{l_p\}_{p \in \{0, 1, 2 \dots M\}}, \{r_p\}_{p \in \{0, 1, 2 \dots M\}}$ using the recursion in Eq.~(\ref{eq:adjoint_single_sum_recursion}) --- this can be done in $O(M)$ time. 
    \item[(3)] Finally, we use Eqs.~(\ref{eq:deriv_single_sum_adj}a), (\ref{eq:deriv_single_sum_adj}b) and (\ref{eq:deriv_single_sum_adj}d) to compute $\{\partial \Gamma^{(1)}(\theta)/ \partial \theta_p\}_{p \in \{0, 1, 2 \dots M\}}$ --- this can again be done in $O(M)$ time.
\end{enumerate}

\emph{Cost functions that are double sums}. Next, we consider cost functions $\Gamma^{(2)}(\theta)$ that are of the form
\begin{align}
    \Gamma^{(2)}(\theta) = \sum_{n = 1}^N \sum_{m = 1}^{n - 1}f(\gamma_{n, m}),
\end{align}
where $f:\mathbb{C} \to \mathbb{R}$ is Wirtinger differentiable and $\gamma_{n, m} \in \mathbb{C}$ is given by:
\begin{align}\label{eq:gamma_m_n}
    \gamma_{n, m} = 
    \begin{array}{c}
        \begin{tikzpicture}[scale=0.6]
            \BTensor{(0, 0)}{\small $\phi_i$}{\small $\phi_i^*$}{1};
            \BTensor{(-\singledx*6, 0)}{\small $\phi_f^*$}{\small $\phi_f$}{-1};
            \KGate{(-\singledx*2, 0)}{\small $Q_2$}{\small $P_2^*$};
            \KGate{(-\singledx*4, 0)}{\small $Q_1$}{\small $P_1^*$};
            \CTensor{(-\singledx, 0)}{\small $\text{E}_{0}^{m - 1}$};
            \CTensor{(-\singledx*3, 0)}{\small $\text{E}_{m}^{n - 1}$};
            \CTensor{(-\singledx*5, 0)}{\small $\text{E}_{n}^{M}$};
        \end{tikzpicture}
    \end{array},
\end{align}
for some fixed tensors $P_1, P_2, Q_1, Q_2$. Again, the derivative of $\Gamma^{(2)}(\theta)$, $\partial \Gamma^{(2)}(\theta) / \partial \theta_p$ is given by
\begin{subequations}\label{eq:deriv_gamma_m_n}
\begin{align}
    \frac{\partial}{\partial \theta_p}\Gamma^{(2)}(\theta) = 2\sum_{n = 1}^M \sum_{m = 1}^{n - 1}\text{Re}\bigg(f'(\gamma_{n, m}) \frac{\partial \gamma_{n, m}}{\partial \theta_p}\bigg),
\end{align}
where $f' \equiv \partial f(z, z^*)/\partial z$. Note that $\partial \gamma_{n, m}/\partial \theta_p$ will be given by
\begin{align}\label{eq:gamma_n_m_deriv}
    \frac{\partial \gamma_{n, m}}{\partial \theta_p} &= 
    \begin{array}{c}
        \begin{tikzpicture}[scale=0.6]
        \BTensor{(0, 0)}{\small $\phi_f^*$}{\small $\phi_f$}{-1};
        \CTensor{(\singledx, 0)}{\small $\text{E}_n^M$};
        \KGate{(\singledx*2, 0)}{\small $Q_1$}{\small $P_1^*$};
        \CTensor{(\singledx*3, 0)}{\small $\text{E}_m^{n - 1}$};
        \KGate{(\singledx*4, 0)}{\small $Q_2$}{\small $P_2^*$};
        \CTensor{(\singledx*5, 0)}{\small $\text{L}_{p + 1}^{m - 1}$};
        \UDTensor{(\singledx*6, 0)}{\large $\frac{\partial \mathcal{U}_p^S}{\partial \theta_p}$};
        \CTensor{(\singledx*7, 0)}{\small $\text{R}_{0}^{p - 1}$};
        \BTensor{(\singledx*8, 0)}{\small $\phi_i$}{\small $\phi_i^*$}{1};
        \end{tikzpicture}
    \end{array} \ \text{ if }0 \leq p < m, \nonumber\\
    \frac{\partial \gamma_{n, m}}{\partial \theta_p} &= 
    \begin{array}{c}
        \begin{tikzpicture}[scale=0.6]
        \BTensor{(0, 0)}{\small $\phi_f^*$}{\small $\phi_f$}{-1};
        \CTensor{(\singledx, 0)}{\small $\text{E}_n^M$};
        \KGate{(\singledx*2, 0)}{\small $Q_1$}{\small $P_1^*$};
        \CTensor{(\singledx*3, 0)}{\small $\text{L}_{p + 1}^{n - 1}$};
        \UDTensor{(\singledx*4, 0)}{\large $\frac{\partial \mathcal{U}_p^S}{\partial \theta_p}$};
        \CTensor{(\singledx*5, 0)}{\small $\text{R}_{m}^{p - 1}$};
        \KGate{(\singledx*6, 0)}{\small $Q_2$}{\small $P_2^*$};
        \CTensor{(\singledx*7, 0)}{\small $\text{E}_{0}^{m - 1}$};
        \BTensor{(\singledx*8, 0)}{\small $\phi_i$}{\small $\phi_i^*$}{1};
        \end{tikzpicture}
    \end{array} \ \text{ if }m \leq p < n \text{ and},\nonumber \\
     \frac{\partial \gamma_{n, m}}{\partial \theta_p} &= 
    \begin{array}{c}
        \begin{tikzpicture}[scale=0.6]
        \BTensor{(0, 0)}{\small $\phi_f^*$}{\small $\phi_f$}{-1};
        \CTensor{(\singledx, 0)}{\small $\text{L}_{p + 1}^M$};
        \UDTensor{(\singledx*2, 0)}{\large $\frac{\partial \mathcal{U}_p^S}{\partial \theta_p}$};
        \CTensor{(\singledx*3, 0)}{\small $\text{R}_{n}^{p - 1}$};
        \KGate{(\singledx*4, 0)}{\small $Q_1$}{\small $P_1^*$};
        \CTensor{(\singledx*5, 0)}{\small $\text{E}_{m}^{n - 1}$};
        \KGate{(\singledx*6, 0)}{\small $Q_2$}{\small $P_2^*$};
        \CTensor{(\singledx*7, 0)}{\small $\text{E}_{0}^{m - 1}$};
        \BTensor{(\singledx*8, 0)}{\small $\phi_i$}{\small $\phi_i^*$}{1};
        \end{tikzpicture}
    \end{array} \ \text{ if }n \leq p \leq M.
\end{align}
\end{subequations}
Note that using Eq.~(\ref{eq:deriv_gamma_m_n}) to compute $\{\partial \Gamma^{(2)}(\theta)/\partial \theta_p\}_{p \in \{0,1, 2 \dots M\}}$ would require $O(M^3)$ time. However, similar to the case of a cost function with a single summation, this computation can be re-organized to require $O(M^2)$ time. To do so, we note that
\begin{subequations}\label{eq:deriv_double_sum_adj}
\begin{align}
    \frac{\partial}{\partial \theta_p}\Gamma^{(2)}(\theta) =2\text{Re}\bigg(\sum_{p < m < n \leq M} f'(\gamma_{n, m}) \frac{\partial \gamma_{n, m}}{\partial \theta_p}\bigg) + 2\text{Re}\bigg(\sum_{\substack{1\leq m \leq p \\ p < n \leq M}} f'(\gamma_{n, m}) \frac{\partial \gamma_{n, m}}{\partial \theta_p}\bigg) + 2\text{Re}\bigg(\sum_{ 1\leq m < n \leq p} f'(\gamma_{n, m}) \frac{\partial \gamma_{n, m}}{\partial \theta_p}\bigg).
\end{align}
Using Eq.~(\ref{eq:gamma_n_m_deriv}), it follows that
\begin{align}
       \sum_{p < m < n \leq M} f'(\gamma_{n, m})\frac{\partial \gamma_{n, m}}{\partial \theta_p} = 
    \begin{array}{c}
        \begin{tikzpicture}[scale=0.6]
            \BTensor{(0, 0)}{\small $\phi_i$}{\small $\phi_i^*$}{1};
            \CTensor{(-\singledx, 0)}{\small $\text{R}_{0}^{p - 1}$};
            \UDTensor{(-\singledx*2, 0}{\large $\frac{\partial \mathcal{U}^S_p}{\partial \theta_p}$};
            \sidetensorleft{(-\singledx*3, 0)}{\small $l_p$};
        \end{tikzpicture}
    \end{array}, 
\end{align}
where
\begin{align}
    \begin{array}{c}
         \begin{tikzpicture}[scale=0.6]
         \sidetensorleft{(0, 0)}{\small $l_p$};
         \end{tikzpicture} 
    \end{array} = \sum_{p < m < n \leq M}f'(\gamma_{n, m})\left(
    \begin{array}{c}
         \begin{tikzpicture}[scale=0.6]
             \BTensor{(0, 0)}{\small $\phi_f^*$}{\small $\phi_f$}{-1};
             \CTensor{(\singledx, 0)}{\small $\text{E}_n^M$};
             \KGate{(\singledx*2, 0}{\small $Q_1$}{\small $P_1^*$};
             \CTensor{(\singledx*3, 0)}{\small $\text{E}_m^{n - 1}$};
             \KGate{(\singledx*4, 0}{\small $Q_2$}{\small $P_2^*$};
             \CTensor{(\singledx*5, 0)}{\small $\text{L}_{p + 1}^{m - 1}$};
         \end{tikzpicture} 
    \end{array}\right).
\end{align}
Furthermore,
\begin{align}
 \sum_{\substack{1 \leq m \leq p \\ p < n \leq M}} f'(\gamma_{n, m})\frac{\partial \gamma_{n, m}}{\partial \theta_p} = 
    \sum_{1 \leq m \leq p}\left(\begin{array}{c}
        \begin{tikzpicture}[scale=0.6]
            \sidetensorleft{(0, 0)}{\small $c_{p, m}$};
            \UDTensor{(\singledx, 0}{\large $\frac{\partial \mathcal{U}^S_p}{\partial \theta_p}$};
            \CTensor{(\singledx*2, 0)}{\small $\text{R}_{m}^{p - 1}$};
            \KGate{(\singledx*3, 0)}{\small $Q_2$}{\small $P_2^*$};
            \CTensor{(\singledx*4, 0)}{\small $\text{E}_0^{m - 1}$};
            \BTensor{(\singledx*5, 0)}{\small $\phi_i$}{\small $\phi_i^*$}{1};
        \end{tikzpicture}
    \end{array}\right), 
\end{align}
where
\begin{align}
    \begin{array}{c}
         \begin{tikzpicture}[scale=0.6]
         \sidetensorleft{(0, 0)}{\small $c_{p, m}$};
         \end{tikzpicture} 
    \end{array} = \sum_{p  < n \leq M}f'(\gamma_{n, m})\left(
    \begin{array}{c}
         \begin{tikzpicture}[scale=0.6]
             \BTensor{(0, 0)}{\small $\phi_f^*$}{\small $\phi_f$}{-1};
             \CTensor{(\singledx, 0)}{\small $\text{E}_n^M$};
             \KGate{(\singledx*2, 0}{\small $Q_1$}{\small $P_1^*$};
             \CTensor{(\singledx*3, 0)}{\small $\text{L}_{p + 1}^{n - 1}$};
         \end{tikzpicture} 
    \end{array}\right).
\end{align}
Finally,
\begin{align}
 \sum_{\substack{1 \leq m < n \leq p}} f'(\gamma_{n, m})\frac{\partial \gamma_{n, m}}{\partial \theta_p} = 
    \begin{array}{c}
        \begin{tikzpicture}[scale=0.6]
        \BTensor{(0, 0)}{\small $\phi_f^*$}{\small $\phi_f$}{-1};
        \CTensor{(\singledx, 0)}{\small $\text{L}_{p + 1}^M$};
        \UDTensor{(\singledx*2, 0)}{\large $\frac{\partial \mathcal{U}_p^S}{\partial \theta_p}$};
        \sidetensorright{(\singledx*3, 0)}{\small $r_{p}$};
        \end{tikzpicture}
    \end{array}, 
\end{align}
where
\begin{align}
\begin{array}{c}
    \begin{tikzpicture}[scale=0.6]
        \sidetensorright{(0, 0)}{\small $r_{p}$};
    \end{tikzpicture}
\end{array} = 
    \sum_{1 \leq m < n \leq p}\left(
    \begin{array}{c}
        \begin{tikzpicture}[scale=0.6]
        \CTensor{(0, 0)}{\small $\text{R}_{n}^{p - 1}$};
        \KGate{(\singledx, 0)}{\small $Q_1$}{\small $P_1^*$};
        \CTensor{(\singledx*2, 0)}{\small $\text{E}_{m}^{n - 1}$};
        \KGate{(\singledx*3, 0)}{\small $Q_2$}{\small $P_2^*$};
        \CTensor{(\singledx*4, 0)}{\small $\text{E}_{0}^{m - 1}$};
        \BTensor{(\singledx*5, 0)}{\small $\phi_i$}{\small $\phi_i^*$}{1};
        \end{tikzpicture}
    \end{array}\right)
\end{align}
\end{subequations}
If the vectors $l_p, c_{p, m}, r_p$ are pre-computed, then from Eqs.~(\ref{eq:deriv_double_sum_adj}a), (\ref{eq:deriv_double_sum_adj}b), (\ref{eq:deriv_double_sum_adj}d) and (\ref{eq:deriv_double_sum_adj}f), we can compute $\{\partial \Gamma^{(2)}(\theta) / \partial \theta_p\}_{p \in \{0, 1, 2 \dots M\}}$ in $O(M^2)$ time. Furthermore, $l_p, c_{p, m}, r_p$ can also be computed in $O(M^2)$ time by noting that they satisfy the following recursions:
\begin{subequations}\label{eq:adjoint_double_sum_recursion}
\begin{align}
    \begin{array}{c}
         \begin{tikzpicture}[scale=0.6]
         \sidetensorleft{(0, 0)}{\small $l_{p- 1}$};
         \end{tikzpicture} 
    \end{array} &=
    \begin{array}{c}
        \begin{tikzpicture}[scale=0.6]
            \TTensor{(0, 0)}{\small $K$}{\small $K^*$};
            \UTensor{(-\singledx, 0)}{\small $U_p^S$}{\small $U_p^{S*}$}
            \sidetensorleft{(-\singledx*2, 0)}{\small $l_{p}$};
        \end{tikzpicture}
    \end{array} + \sum_{p  < n \leq M}f'(\gamma_{n, p})\left(
    \begin{array}{c}
         \begin{tikzpicture}[scale=0.6]
             \BTensor{(0, 0)}{\small $\phi_f^*$}{\small $\phi_f$}{-1};
             \CTensor{(\singledx, 0)}{\small $\text{E}_n^M$};
             \KGate{(\singledx*2, 0}{\small $Q_1$}{\small $P_1^*$};
             \CTensor{(\singledx*3, 0)}{\small $\text{E}_p^{n - 1}$};
             \KGate{(\singledx*4, 0}{\small $Q_2$}{\small $P_2^*$};
         \end{tikzpicture} 
    \end{array}\right), \\
    \begin{array}{c}
         \begin{tikzpicture}[scale=0.6]
         \sidetensorleftlarger{(0, 0)}{\small $c_{p - 1, m}$};
         \end{tikzpicture} 
    \end{array} &=
    \begin{array}{c}
        \begin{tikzpicture}[scale=0.6]
            \TTensor{(0, 0)}{\small $K$}{\small $K^*$};
            \UTensor{(-\singledx, 0)}{\small $U_p^S$}{\small $U_p^{S*}$}
            \sidetensorleft{(-\singledx*2, 0)}{\small $c_{p,m}$};
        \end{tikzpicture}
    \end{array} + f'(\gamma_{p, m})\left(\begin{array}{c}
         \begin{tikzpicture}[scale=0.6]
             \BTensor{(0, 0)}{\small $\phi_f^*$}{\small $\phi_f$}{-1};
             \CTensor{(\singledx, 0)}{\small $\text{E}_n^M$};
             \KGate{(\singledx*2, 0}{\small $Q_1$}{\small $P_1^*$};
         \end{tikzpicture} 
    \end{array}\right), \\
    \begin{array}{c}
    \begin{tikzpicture}[scale=0.6]
        \sidetensorright{(0, 0)}{\small $r_{p}$};
    \end{tikzpicture}
\end{array} &=      \begin{array}{c}
         \begin{tikzpicture}[scale=0.6]
             \sidetensorright{(0, 0)}{\small $r_{p - 1}$};
             \UTensor{(-\singledx, 0)}{\small $U_{p-1}^S$}{\small $U_{p - 1}^{S*}$};
             \TTensor{(-\singledx*2, 0)}{\small $K$}{\small $K^*$};
         \end{tikzpicture} 
    \end{array} + 
    \sum_{1 \leq m <  p}\left(
    \begin{array}{c}
        \begin{tikzpicture}[scale=0.6]
        \KGate{(\singledx, 0)}{\small $Q_1$}{\small $P_1^*$};
        \CTensor{(\singledx*2, 0)}{\small $\text{E}_{m}^{p - 1}$};
        \KGate{(\singledx*3, 0)}{\small $Q_2$}{\small $P_2^*$};
        \CTensor{(\singledx*4, 0)}{\small $\text{E}_{0}^{m - 1}$};
        \BTensor{(\singledx*5, 0)}{\small $\phi_i$}{\small $\phi_i^*$}{1};
        \end{tikzpicture}
    \end{array}\right).
\end{align}
\end{subequations}
together with the boundary conditions $l_M = c_{M, m} = r_0 = 0$.  Therefore, to compute $\{\partial \Gamma(\theta) / \partial \theta_p\}_{p \in \{0, 1, 2 \dots M\}}$, an algorithm with $O(M)$ run-time is as follows:
\begin{enumerate}
    \item[(1)] First compute $\{\text{E}_{n,m}\}_{0\leq m \leq n \leq M}$ --- this can be done in $O(M^2)$ by using, for every $m$, the following recurrence in $n$
    \begin{align}
        \begin{array}{c}
        \begin{tikzpicture}[scale=0.6]
            \CTensor{(0, 0)}{\small $E_{m}^{n}$};
        \end{tikzpicture}
    \end{array} = \begin{array}{c}
        \begin{tikzpicture}[scale=0.6]
            \CTensor{(0, 0)}{\small $E_{m}^{n - 1}$};
            \TTensor{(-\singledx, 0)}{\small $K$}{\small $K^*$};
            \UTensor{(-\singledx*2, 0)}{\small $U_{n }^S$}{\small $U_{n}^{S*}$}
        \end{tikzpicture}
    \end{array} ,
    \end{align}
    together with $E^{m- 1}_{m} =  I$ and then using Eq.~(\ref{eq:gamma_m}).
    \item[(2)] Next, we compute $l_p, c_{p, m}, r_p$ using the recursion in Eq.~(\ref{eq:adjoint_double_sum_recursion}) --- this can be done in $O(M^2)$ time. 
    \item[(3)] Finally, we use Eqs.~(\ref{eq:deriv_double_sum_adj}a), (\ref{eq:deriv_double_sum_adj}b) and (\ref{eq:deriv_double_sum_adj}d) to compute $\{\partial \Gamma^{(1)}(\theta)/ \partial \theta_p\}_{p \in \{0, 1, 2 \dots M\}}$ --- this can again be done in $O(M^2)$ time.
\end{enumerate}
\section{Optimal measurement}This Appendix provides details on the optimal measurements discussed in the main text. First, we provide details of the optimal measurement protocol which utilizes tunable non-linear optical systems. Next, we complete the details of optimality of photodetection and tunable linear optics from Section \ref{sec:pdlo}.

\subsection{Optimal measurement with non-linear optics}\label{app:opt_measurement_nlo}
We first recall a basic fact: A measurement to extract $\varphi$ from $\ket{\psi_\varphi}$ when $\varphi = \varphi_0$ that saturates the CRB is the projective measurement $P_0 = \ket{\psi_{\varphi_0}}\! \bra{\psi_{\varphi_0}}$, $P_1, P_2 \dots$, where $P_0 + P_1 + P_2 + \dots = I$ and $P_i P_j = \delta_{i, j} P_i$. To see this, note that the probability of the $i^\text{th}$ measurement outcome is given by $p_i(\varphi) = \bra{\psi_\varphi}  P_i \ket{\psi_\varphi}$. The classical Fisher Information (CFI) for this probability distribution at $\varphi = \varphi_0$ is given by
\begin{equation} 
    \text{CFI} = \lim_{\varphi \to \varphi_0}\sum_i \frac{(p_i'(\varphi))^2}{p_i(\varphi)}.
\end{equation}
We notice that $p_i(\varphi_0) = \delta_{i, 0}$ and $p_i'(\varphi_0) = 0$. Thus, CFI($\varphi_0$) can be expressed as
\begin{equation}
\text{CFI} = \lim_{\varphi \to \varphi_0} \sum_i \frac{(p_i'(\varphi))^2}{p_i(\varphi)} = 2 \sum_{i \neq 0}p_i''(\varphi_0).
\end{equation}
Since for $i \neq 0$, $P_i \ket{\psi_{\varphi_0}} = 0$, it follows that 
\[
p_i''(\varphi_0) = \bra{\psi_{\varphi_0}}P_i \ket{\psi''_{\varphi_0}} + \bra{\psi''_{\varphi_0}}P_i \ket{\psi_{\varphi_0}} + 2\bra{\psi'_{\varphi_0}}P_i \ket{\psi'_{\varphi_0}} = 2\bra{\psi'_{\varphi_0}}P_i \ket{\psi'_{\varphi_0}},
\]
and consequently,
\[
\text{CFI} = 4\sum_{i\neq 0}\bra{\psi'_{\varphi_0}} P_i \ket{\psi'_{\varphi_0}} = 4\big(\smallnorm{\ket{\psi_{\varphi_0}'}}^2 - \smallabs{\bra{\psi_{\varphi_0}} \psi'_{\varphi_0}\rangle}^2\big) = \text{QFI},
\]
thus establishing that this measurement is optimal.

Next, we continue the analysis in Section~\ref{sec:pd_nlo} of the main text, and provide details on how to implement re-absorption of the photons emitted by the source using a non-linear optical cavity. As was discussed in Appendix \ref{app:mps_rep}, the photons emitted in time $T$ into the two input ports of the MZI can be approximately expressed as an MPS on $M = T/\varepsilon$ qudits with 3 levels, where $\varepsilon$ is a small discretization parameter. It is a well known fact that a sequence of unitaries can be designed between an ancilla qudit and the MPS $\ket{\psi}$ to map $\ket{\psi}\to \ket{0}^{\otimes M}$ \cite{schon2005sequential}. More specifically, given an MPS on $M$ qudits with $d$-levels and bond dimension $D$ with tensors $A_1, A_2 \dots A_M$, we can always alternatively expressed in terms of unitaries $U_1, U_2 \dots U_M$, with $U_i:\mathbb{C}^D \otimes \mathbb{C}^d \to \mathbb{C}^D \otimes \mathbb{C}^d$ \cite{perez2006matrix}:
\begin{align}
    \begin{array}{c}
        \begin{tikzpicture}[scale=0.6]
            \bTensor{(0, 0)}{$\phi_f^*$}{-1};
            \KTensor{(\dx,0)}{\small $A_M$};
            \SingleDots{\dx*2,0}{\dx/2};
            \KTensor{(\dx*3,0)}{\small $A_2$};
            \KTensor{(\dx*4,0)}{\small $A_1$};
            \bTensor{(\dx*5, 0)}{$\phi_i$}{1};
        \end{tikzpicture}
    \end{array} = 
        \begin{array}{c}
        \begin{tikzpicture}[scale=0.6, baseline={([yshift=1.3cm] current bounding box.center)}]
            \UT{(0,0)}{\small $U_M$};
            \SingleDots{(\dx,0)}{\dx/2};
            \UT{(\dx*2,0)}{\small $U_2$};
            \UT{(\dx*3,0)}{\small $U_1$};
            \draw[>-, thick] (-0.9, 0) -- (-0.8, 0);
            \draw[>-, thick] (0.8, 0) -- (0.9, 0);
            \draw[>-,,thick] (-0.9 + \dx*2, 0) -- (-0.8 + \dx*2, 0);
            \draw[>-, thick] (-0.9 + \dx*3, 0) -- (-0.8 + \dx*3, 0);
            \draw[>-, thick] (0.8 + \dx*3, 0) -- (0.9 + \dx*3, 0);
            \draw[>-, thick] (0, 0.8) -- (0, 0.9);
            \draw[>-, thick] (\dx*2, 0.8) -- (\dx*2, 0.9);
            \draw[>-, thick] (\dx*3, 0.8) -- (\dx*3, 0.9);
            \draw[>-, thick] (0, -0.9) -- (0, -0.8);
            \draw[>-, thick] (\dx*2, -0.9) -- (\dx*2, -0.8);
            \draw[>-, thick] (\dx*3, -0.9) -- (\dx*3, -0.8);
            \draw(0, -1.7) node {$\ket{0}$};
            \draw(\dx*2, -1.7) node {$\ket{0}$};
            \draw(\dx*3, -1.7) node {$\ket{0}$};
            \draw(-1.7, 0) node {$\ket{0}$};
            \draw(\dx*3 + 1.9, 0) node {$(\ket{0})$};
        \end{tikzpicture}
    \end{array}
\end{align}
In other words, the MPS $\ket{\psi}$ can be generated by initializing a $D-$level ancilla and the $M$ qudits in $\ket{0}$ and then applying $U_M$ on the ancilla and the $M^\text{th}$ qudit, then $U_{M - 1}$ on the ancilla and the $(M - 1)^\text{th}$ qudit and so on. Furthermore, since $\ket{\psi}$ is a normalized state, it is always guaranteed that on applying the last unitary in this sequence, $U_1$, the ancilla will disentangle from the remaining qudits and can be assumed to be in the state $\ket{0}$. To undo the state preparation, we can thus start from an ancilla in $\ket{0}$, apply $U_1^\dagger$ to the ancilla and the right-most qudit, then apply $U_2^\dagger$ to the ancilla and the next qudit and so on --- at the end of this sequence of unitaries, the ancilla as well as the remaining qudits will be in $\ket{0}$ state.

To physically implement these unitaries, our proposal is to use a multi-mode optical cavity with $\chi^{(3)}$ nonlinearity and coherent drives. We assume that we can implement the following Hamiltonian $H_R(t)$ between the three cavities
\begin{align}\label{eq:Hamiltonian_kerr}
    H_R(t) = \frac{\chi}{2}\sum_{k, k' \in \{0,A,B\}}c_k^\dagger c_{k'}^\dagger c_k c_{k'} + \sum_{k \in \{0,A,B\}} \big(\lambda_k^*(t) e^{-i\phi_k(t)} b_k +\text{h.c.}\big),
\end{align}
where, as shown in Fig.~\ref{fig:mzi_reabsorb}, $c_A$ and $c_B$ are the annihilation operators of the modes coupling to the output ports of the MZI and $c_0$ is the annihilation operator of an additional mode that will play the role of the ancilla. Furthermore, we also assume that we can tunably and linearly couple the modes $c_A$ and $c_B$ to the MZI output ports via the Hamiltonian
\begin{align}
    H_{R, P}(t) = V_A(t) a_t^\dagger c_A + V_B(t) b_t^\dagger c_B + \text{h.c.}
\end{align}
Now, to implement the unitary that coherently re-absorbs the photons emitted by the source, at every timestep $[m\varepsilon, (m + 1)\varepsilon)$, we perform two steps: \emph{first}, transfer the photons from the segment $[m\varepsilon, (m + 1)\varepsilon)$ in the ports to the cavity modes $c_A, c_B$ and \emph{second}, apply the unitary in between $U_m^\dagger$ in between the three modes. To perform the first step, we set all the tunable parameters ($\{\lambda_k(t)\}$) to 0, and choose $V_A(t), V_B(t)$ in the interval $t \in [m\varepsilon, (m + 1)\varepsilon)$ so as to map 
    \[
    \big(1 + \alpha_0 A_m^\dagger + \beta_0 B_m^\dagger)\ket{\text{vac}} \to \big(1 + \alpha_0 c_A^\dagger + \beta_0 c_B^\dagger) \ket{0_A, 0_B},
    \]
    where, following Appendix \ref{app:mps_rep}, $A_m = \int_{m\varepsilon}^{(m + 1)\varepsilon}a_t dt / \sqrt{\varepsilon}, B_m = \int_{m\varepsilon}^{(m + 1)\varepsilon}b_t dt / \sqrt{\varepsilon}$ and $\ket{0_A}, \ket{0_B}$ are respectively the vacuum states of the oscillators $c_A$ and $c_B$. Recall from Appendix \ref{app:mps_rep} that for sufficiently small $\varepsilon$, the segment of the ports corresponding to $[m\varepsilon, (m + 1)\varepsilon)$ has at most $1$-photon upto $O(\varepsilon^2)$ error, so this unitary effectively transfers the photons in the segment $[m\varepsilon, (m + 1)\varepsilon)$ of the two ports to the oscillators $c_A$ and $c_B$. It is also easy to check that the explicit choice $V_A(t), V_B(t) = ((m + 1)\varepsilon - t)^{-1/2}$ for $t \in [m\varepsilon, (m + 1)\varepsilon)$ accomplishes this transfer of excitations from the ports to the two coupled cavities.
    
Next, we switch off the couplings to the MZI ports (i.e.,~set $V_A(t), V_B(t) = 0$) and apply the unitary $U_m^\dagger$: Note that the unitary $U_m^\dagger$ acts on the space $\mathbb{C}^{3}\otimes \mathbb{C}^D$, where the qutrit is in between the states $\ket{0_A, 0_B}, c_A^\dagger \ket{0_A, 0_B}, c_B^\dagger \ket{0_A, 0_B}$ and the $D$-level qudit corresponds to the state of the oscillator $c_0$. We also point out that this gate needs to applied on time-scales much faster than $\varepsilon$, i.e.,~before the oscillators interact with the next time-bin. However, it would appear from Eq.~(\ref{eq:Hamiltonian_kerr}) that the gate-time is limited to $\sim 1/\chi$. Experimentally it is typically only feasible to have very large coherent drives $\lambda_k(t)$, but $\chi$ remains a fixed (and often small) constant. However, we show that by adapting the strategy of Ref.~\cite{yuan2023universal}, we can use a large drive to apply any target unitary on this subspace within a time limited only by the strength of drive $\lambda_k(t)$.

Suppose $U_R(t,s)=\mathcal{T}\text{exp}(-\int_s^t H_R(\tau) d\tau)$ is the unitary group corresponding to $H_R(t)$: We first consider this unitary group in the frame rotated with respect to the Hamiltonian $\sum_k\phi_k'(t)c_k^\dagger c_k $ and displaced with respect to the Hamiltonian $ \sum_k (\dot{\gamma}_k(t) c_k^\dagger - \text{h.c.})$. Furthermore, given $\gamma_k(t)$ (which we will choose later) we will make the choices
\begin{align*}
&\dot{\phi}_k(t) =- \chi\sum_{k'}\smallabs{\gamma_{k'}(t)}^2 - \frac{\chi}{2}\smallabs{\gamma_k(t)}^2,\nonumber\\
&\lambda_k(t) = i\dot{\gamma}_k(t)-\dot{\phi}_k(t) \gamma_k(t)-\chi D \gamma_k(t)-\chi \sum_{k'}\smallabs{\gamma_k(t)}^2.
\end{align*}
The effective Hamiltonian $\tilde{H}_R(t)$ in this new frame is then given by
\begin{align}
\tilde{H}_R(t) = \frac{\chi}{2}\sum_{k, k'} c_k^\dagger c_{k'}^\dagger c_k c_{k'} + \chi \sum_{k}\big(\gamma_k(t) c_k^\dagger(N-D)+\text{h.c.}\big)+ \frac{\chi}{2}\sum_{k\neq k'} \big(c_k^\dagger c_{k'}^\dagger \gamma_k(t) \gamma_{k'}(t) + c_k^\dagger c_{k'}\gamma_k(t) \gamma_{k'}^*(t)+\text{h.c.}\big),
\end{align}
where $N=\sum_kc_k^\dagger c_k$ is the total number of photons in the three cavities. To further simplify $\tilde{H}_R(t)$, we also choose $\gamma_k(t) = \Gamma_k(t) f_k(t)$ where $f_k(t)$ are rapidly oscillating functions that further satisfy $\langle f_k(t)\rangle = 1, \langle f_k^2(t)\rangle = 0$, and for $k \neq k'$ $\langle f_k(t)f_{k'}(t)\rangle=\langle f_k(t) f_{k'}^*(t)\rangle= 0$. A set of possible explicit choices for $f_A(t), f_B(t),f_0(t)$ that satisfy this condition are
\begin{align}
    f_k(t) = 1 + \frac{1}{\sqrt{2}} \sum_{n=1}^6 e^{i2 \pi m_k n / 7} \cos(n\Omega t) \text{ where }m_0= 1, m_A = 2 \text{ and }m_B =3.
\end{align}
In the limit of large $\Omega$, the dynamics due to $\tilde{H}_R(t)$ is well approximated by $\hat{H}_R(t)$
\begin{subequations}\label{eq:final_blockaded_hamiltonian}
   \begin{align}
\hat{H}_R(t) = \frac{\chi}{2}\sum_{k, k'} c_k^\dagger c_{k'}^\dagger c_k c_{k'} + \hat{H}^c_R(t)
\end{align}
where
\begin{align}
    \hat{H}^c_R(t) = \chi\sum_k\big( \text{Re}[\Gamma_k(t)] H_k^{(+)} + \text{Im}[\Gamma_k(t)] H_k^{(-)}\big),
\end{align}
with
\begin{align}
   H_k^{(+)} =  c_k^\dagger (N- D)+\text{h.c.} \text{ and }H_k^{(-)} = i(c_k^\dagger (N- D)-\text{h.c.}).
\end{align} 
\end{subequations}
Here $N = \sum_k c_k^\dagger c_k$ is the total number of photons in the three oscillators. Given $\vec{n} = (n_0, n_A, n_B)$ and denoting by $\ket{\vec{n}} \propto (c_0^\dagger)^{n_0}(c_A^\dagger)^{n_A}(c_B^\dagger)^{n_B}\ket{0}$, we first observe that $\hat{H}_R(t)$ and $\hat{H}_R^{c}(t)$ are block diagonal in the subspace $\mathcal{S}^{(D)} = \text{span}(\{\ket{\vec{n}}:n_0+n_A+n_B\leq D\})$ (the subspace with no more than $D$ photons in the three cavities) and the subspace orthogonal to it, i.e.,~for all $\ket{\psi}\in \mathcal{S}^{(D)}$, $\hat{H}_R(t)\ket{\psi}, \hat{H}_R^{c}(t)\ket{\psi} \in \mathcal{S}^{(D)}$. We now concentrate on the finite-dimensional subspace $\mathcal{S}^{(D)}$ and show that the Hamiltonian $\hat{H}_R^c(t)$ is universal in this subspace, i.e.,~any desired unitary $U$ can be implemented up to a global phase with an appropriate choice of $\Gamma_k(t)$. From the standard theory of quantum control \cite{dong2010quantum, d2021introduction}, to establish universality of $H_R^{c}(t)$, we need to establish that the algebra generated by the operators $\{H_k^{(+)}, H_k^{(-)}\}_{k \in \{0, A, B\}}$ under nested commutator is the full $\mathfrak{su}(\text{dim}(\mathcal{S}^{(D)}))$ algebra, i.e.,~any traceless Hermitian operator on the subspace $\mathcal{S}^{(D)}$ can be expressed as a linear combination of $\{H_k^{(+)}, H_k^{(-)}\}_{k \in \{0, A, B\}}$ and their nested commutators. Reference~\cite{yuan2023universal} considered this problem for the case of a single oscillator, and established universality of $\hat{H}_R^c(t)$ using a result from Ref.~\cite{schirmer2001complete} that described a set of sufficient checkable conditions for a set of operators to generate the full algebra of traceless Hermitian operators. However, the result of Ref.~\cite{schirmer2001complete} does not apply directly to $\hat{H}_R^c(t)$ when more than one oscillator is present and their argument has to be modified. We do so in the following lemma.

\begin{lemma}\label{lemma:controllability}
    For a positive integers $k, d>1$, consider a finite-dimensional Hilbert space $\mathcal{H}$ with orthonormal basis elements $\{\ket{\vec{n}} \text{ where }\vec{n} = \{n_1,n_2 \dots n_k\} \text{ with } n_1 + n_2 \dots n_k \leq d\}$. Consider the Hermitian operators $H_0$, $\{H_p^{(+)}, H_p^{(-)}\}_{p \in \{1, 2 \dots k\}}$ where
    \begin{align*}
    &H_0 = \sum_{\vec{n}: \norm{\vec{n}}_1 \leq d} \mu({\vec{n}}) \ket{\vec{n}}\!\bra{\vec{n}},\nonumber\\
    &H_p^{(+)} = \sum_{\vec{n}:\smallnorm{\vec{n}}_1 \leq d - 1}d_p^{(+)}(\vec{n}) \big(\ket{\vec{n}}\!\bra{\vec{n} + \vec{e}_p} +\text{h.c.}\big), \text{ and }\nonumber\\
    &H_p^{(-)} = i\sum_{\vec{n}:\smallnorm{\vec{n}}_1 \leq d - 1}d_p^{(-)}(\vec{n}) \big(\ket{\vec{n}}\!\bra{\vec{n} +\vec{e}_p} -\text{h.c.}\big),
    \end{align*}
    where $\vec{e}_p = (\underbrace{0, \dots 0}_{p - 1 \text{ times}}, 1, 0 \dots 0)$.
    If the coefficients $d_p^{(\pm)}(\vec{n})$, $\mu(\vec{n})$ satisfy the conditions
    \begin{enumerate}
        \item[C1.] {Non-zero off diagonal elements in $H_p^{(\pm)}$}: For all $p \in \{1, 2 \dots k\}$ and $\vec{n}$, $d_p^{(\pm)}(\vec{n}) \neq 0$.
        \item[C2.] Non-degeneracy in level-spacings of $H_0$: For all $p \in \{1, 2 \dots k\}$, 
        \[
        (\mu(\vec{n} + \vec{e}_p) - \mu(\vec{n}))^2 = (\mu(\vec{m} + \vec{e}_p) - \mu(\vec{m}))^2 \text{ if and only if }\vec{n} = \vec{m},
        \]
    \end{enumerate}
    then the algebra generated by the set of operators $\{H_0, H^{(\pm)}_1, H^{(\pm)}_2 \dots H^{(\pm)}_k\}$ contains the full $\mathfrak{su}(\textnormal{dim}(\mathcal{H}))$ algebra.
\end{lemma}
\begin{proof}
Unless otherwise mentioned, throughout this proof $\vec{n} = (n_1, n_2 \dots n_k)$ will denote a vector of non-negative integers with $n_1 + n_2+ \dots + n_k \leq d$. For notational convenience, we will define the Hermitian operators $E^{(\pm)}(\vec{n}, \vec{m})$ via
\[
E^{(0)}(\vec{n}, \vec{m}) = \ket{\vec{n}}\!\bra{\vec{n}} - \ket{\vec{m}}\!\bra{\vec{m}}, E^{(+)}(\vec{n}, \vec{m}) = \ket{\vec{n}}\!\bra{\vec{m}} + \text{h.c.} \text{ and }E^{(-)}(\vec{n}, \vec{m}) = i(\ket{\vec{n}}\!\bra{\vec{m}} - \text{h.c.}).
\]
Furthermore, for an operator $E$, we will denote by $\mathcal{C}_E$ its commutator, i.e., $\mathcal{C}_E(X) =[E, X]$. For a set of operators $\mathcal{Q} = \{E_1, E_2 \dots E_n\}$, we will denote by $\text{alg}(\mathcal{Q})$ to be the algebra generated by $\mathcal{Q}$ under nested commutators.

We first observe that the operators $\mathcal{M} = \{E^{(+)}(\vec{n}, \vec{n} + \vec{e}_p), E^{(-)}(\vec{n}, \vec{n} + \vec{e}_p)\}_{\vec{n}: \norm{\vec{n}}_1 \leq d - 1}$ generate the full $\mathfrak{su}(\text{dim}(\mathcal{H}))$ algebra (i.e.,~the algebra of traceless Hermitian operators on $\mathcal{H}$). To see this explicitly, note that for any $\vec{n} \neq \vec{m}$, we can always find $\vec{k}_1, \vec{k}_2 \dots \vec{k}_P$ such that $\smallnorm{\vec{n} - \vec{k}_1}_1, \smallnorm{\vec{k}_1 - \vec{k}_2}_1, \smallnorm{\vec{k}_2 - \vec{k}_3}_1 \dots \smallnorm{\vec{k}_P - \vec{m}}_1 = 1$ and consequently $E^{(\pm)}(\vec{n}, \vec{k}_1), E^{(\pm)}(\vec{k}_1, \vec{k}_2), E^{(\pm)}(\vec{k}_2, \vec{k}_3) \dots E^{(\pm)}(\vec{k}_P, \vec{m}) \in \mathcal{M}$. Furthermore, since for any $\vec{n} \neq \vec{m} \neq \vec{k}$, $[E^{(+)}(\vec{n}, \vec{m}), E^{(+)}(\vec{m}, \vec{k})]= -i E^{(-)}(\vec{m}, \vec{k})$, $[E^{(+)}(\vec{n}, \vec{m}), E^{(-)}(\vec{m}, \vec{k})] = iE^{(+)}(\vec{n}, \vec{k})$, $[E^{(-)}(\vec{n}, \vec{m}), E^{(-)}(\vec{m}, \vec{k})] = -E^{(+)}(\vec{n}, \vec{k})$, it follows that for both $\sigma \in \{+, -\}$
\[
E^{(\sigma)}(\vec{n}, \vec{m}) = \eta\mathcal{C}_{E^{(\sigma_1)}(\vec{n}, \vec{k}_1)}, \mathcal{C}_{E^{(\sigma_2)}(\vec{k}_1, \vec{k}_2)} \dots \mathcal{C}_{E^{(\sigma_P)}(\vec{k}_{P-1}, \vec{k}_P)}, E^{(\sigma_{P + 1})}(\vec{k}_P, \vec{m}),
\]
for some $\sigma_1, \sigma_2 \dots \sigma_{P + 1} \in \{+, -\}$ and $ \eta \in \{\pm 1, \pm i\}$
and therefore $E^{(\pm)}(\vec{n}, \vec{m}) \in \text{alg}(\mathcal{M})$. Furthermore, since $[E^{(+)}(\vec{n}, \vec{m}), E^{(-)}(\vec{n}, \vec{m})] = -iE^{(0)}(\vec{n}, \vec{m})$, we conclude that for all $\vec{n}, \vec{m}$, $E^{(0)}(\vec{n}, \vec{m}) \in \text{alg}(\mathcal{M})$. Since the operator set $\{E^{(0)}(\vec{n}, \vec{m}), E^{(+)}(\vec{n}, \vec{m}), E^{(-)}(\vec{n}, \vec{m})\}_{\vec{n}\neq \vec{m}}$ spans the whole set of Hermitian traceless operators on $\mathcal{H}$, we obtain that $\text{alg}(\mathcal{M}) = \mathfrak{su}(\text{dim}(\mathcal{H}))$.

Next, we show that $\mathcal{M} = \{E^{(+)}(\vec{n}, \vec{n} + \vec{e}_p), E^{(-)}(\vec{n}, \vec{n} + \vec{e}_p)\}_{\vec{n}: \norm{\vec{n}}_1 \leq d - 1}$ is contained in $\text{alg}(\{H_0, H_1^{(\pm)}, H_2^{(\pm)} \dots H_k^{(\pm)}\})$. For this, we observe that $\mathcal{C}_{H_0}^{2k}(H_p^{(+)})  \in \text{alg}(\{H_0, H_1^{(\pm)}, H_2^{(\pm)} \dots H_k^{(\pm)}\})$ is given by
\[
\mathcal{C}_{H_0}^{2k}(H_p^{(+)}) = \sum_{\vec{n}:\norm{\vec{n}}_1 \leq d - 1} (\mu(\vec{n} + \vec{e}_p) - \mu(\vec{n}))^{2k} d_p^{(+)}(\vec{n})E^{(+)}(\vec{n}, \vec{n} + \vec{e}_p).
\]
We note that this equation can be re-written as
\[
\begin{bmatrix}
    H_p^{(+)} \\
    \mathcal{C}_{H_0}^2(H_p^{(+)}) \\
    \mathcal{C}_{H_0}^4(H_p^{(+)})  \\
    \vdots
\end{bmatrix} = \text{V}
\begin{bmatrix}
    E^{(+)}(\vec{n}_1, \vec{n}_1 +  \vec{e}_p) \\
    E^{(+)}(\vec{n}_2, \vec{n}_1 +  \vec{e}_p) \\
    E^{(+)}(\vec{n}_3, \vec{n}_1 +  \vec{e}_p) \\
    \vdots
\end{bmatrix},
\]
where
\[
\textbf{V} = \begin{bmatrix}
   1 & 1 & 1 & \dots \\
   (\mu(\vec{n}_1 + \vec{e}_p) - \mu(\vec{n}_1))^2 & (\mu(\vec{n}_2 + \vec{e}_p) - \mu(\vec{n}_2))^2 & (\mu(\vec{n}_3 + \vec{e}_p) - \mu(\vec{n}_3))^2 & \dots \\
    (\mu(\vec{n}_1 + \vec{e}_p) - \mu(\vec{n}_1))^4 & (\mu(\vec{n}_2 + \vec{e}_p) - \mu(\vec{n}_2))^4 & (\mu(\vec{n}_3 + \vec{e}_p) - \mu(\vec{n}_3))^4 & \dots \\
    \vdots & \vdots & \vdots & \ddots
\end{bmatrix}.
\]
We note that $\textbf{V}$ is a Vandermonde matrix and since for $\vec{n} \neq \vec{m}$, $(\mu(\vec{n} + \vec{e}_p)^2 - \mu(\vec{n}))^2 \neq (\mu(\vec{m} + \vec{e}_p)^2 - \mu(\vec{m}))^2$ by assumption, $\textbf{V}$ is invertible. Therefore, we obtain that $E^{(+)}(\vec{n}, \vec{n} + \vec{e}_p)$ can be expressed as linear combinations of $H_p^{(+)}, \mathcal{C}^2_{H_0}(H_p^{(+)}), \mathcal{C}^4_{H_0}(H_p^{(+)}) \dots$ and thus $E^{(+)}(\vec{n}, \vec{n} + \vec{e}_p) \in \text{alg}(\{H_0, H_1^{(\pm)}, H_2^{(\pm)} \dots H_k^{(\pm)}\})$. A similar argument can be repeated to show that $E^{(-)}(\vec{n}, \vec{n} + e_p) \in \text{alg}(\{H_0, H_1^{(\pm)}, H_2^{(\pm)} \dots H_k^{(\pm)}\})$ for all $\vec{n}$. Since we have already previously established that the algebra generated by $\{E^{(+)}(\vec{n}, \vec{n} + \vec{e}_p), E^{(-)}(\vec{n}, \vec{n} + \vec{e}_p)\}$ is the full $\mathfrak{su}(\text{dim}(\mathcal{H}))$, we obtain the lemma statement.
\end{proof}

Next, we use lemma \ref{lemma:controllability} to show that the Hamiltonian $H_R^c(t)$ in Eq.~(\ref{eq:final_blockaded_hamiltonian}) is completely controllable within the blockaded subspace $\mathcal{S}^{(D)} = \text{span}(\{\ket{n_0, n_A, n_B}: n_0 + n_A + n_B \leq D\}$. We first re-write it within the subspace $\mathcal{S}^{(D)}$ as
\[
\hat{H}_R^c(t) = \chi \sum_{k \in \{0, A, B\}}\big(\text{Re}[\Gamma_k(t)] H_k^{(+)} + \text{Im}[\Gamma_k(t)] H_k^{(-)}\big),
\]
where
\begin{align*}
&H_k^{(+)} = c_k^\dagger (N - D) + \text{h.c.} = \sum_{\vec{n}:\smallnorm{\vec{n}}_1 \leq D - 1} (\smallnorm{\vec{n}}_1 - D) \sqrt{n_k + 1} \big(\ket{\vec{n}}\!\bra{\vec{n} + \vec{e}_k} + \text{h.c.}\big),\nonumber\\
&H_k^{(-)} = i(c_k^\dagger (N - D) - \text{h.c.}) = i\sum_{\vec{n}: \smallnorm{\vec{n}}_1 \leq D - 1} (\smallnorm{\vec{n}}_1 - D) \sqrt{n_k + 1} \big(\ket{\vec{n}}\!\bra{\vec{n} + \vec{e}_k} - \text{h.c.}\big).
\end{align*}
Note that the operators $H_k^{(\pm)}$ already satisfy the conditions of lemma \ref{lemma:controllability}.
Next, for any real constants $\vec{\xi} = \{\xi_0, \xi_A, \xi_B\}$, the algebra generated by $\{H_k^{(+)}, H_k^{(-)}\}_{k\in \{0, A, B\}}$ contains the operator $H_0$ given by
\begin{subequations}\label{eq:H_0_blockade}
\begin{align}
H_0 = i \sum_k \xi_k [H_k^{(+)}, H_k^{(-)}] = \sum_{\vec{n}: \norm{\vec{n}}_1 \leq D} \mu(\vec{n}) \ket{\vec{n}}\!\bra{\vec{n}},
\end{align}
where
\begin{align}
    \mu(\vec{n}) = \sum_k\big(\xi_k n_k \big(2\smallnorm{\vec{n}}_1 - 2D - 1\big)  + \xi_k(\smallnorm{n}_1 - D)^2\big).
\end{align}
\end{subequations}
We next show that the constants $\xi_0, \xi_A, \xi_B$ can be chosen to satisfy the non-degeneracy condition in lemma $\ref{lemma:controllability}$. We will choose $\xi_k$ to be irrational numbers such that their products $\{\xi_k \xi_{k'}\}_{k, k' \in \{0, A, B\}}$ are irrational and incommensurate --- a concrete choice could be $\xi_0 = 2^{1/3}, \xi_A = 3^{1/3}, \xi_B =  5^{1/3}$. Next, we note that for any integers $\vec{p} = (p_0, p_A, p_B)$ and $\vec{m} = (q_0, q_A, q_B)$,
\begin{align}\label{eq:rational}
\bigg(\sum_{k \in \{0, A, B\}} p_k \xi_k\bigg)^2 = \bigg(\sum_{k \in \{0, A, B\}} q_k \xi_k\bigg)^2 \implies \text{ either }\vec{p} = \vec{q} \text{ or }\vec{p} = -\vec{q}.
\end{align}
We use this fact to now show that $H_0$ in Eq.~(\ref{eq:H_0_blockade}) satisfies the non-degeneracy condition in lemma \ref{lemma:controllability}, i.e., for all $k$, the set of real numbers $\{(\mu(\vec{n} + \vec{e}_k) - \mu(\vec{n}))^2\}_{\vec{n}: \smallnorm{\vec{n}}_1 \leq D}$ are all distinct. We begin by noting that
\[
\mu(\vec{n} + \vec{e}_k) - \mu(\vec{n}) = \big(2n_k + 2(2n - 2D + 1)) \xi_k + \sum_{k' \neq k} \xi_{k'}(2n_{k'} + 2n - 2D + 1).
\]
Next, suppose that there were $\vec{n}, \vec{m}$ such that $(\mu(\vec{n} + \vec{e_k}) - \mu(\vec{n}))^2 = (\mu(\vec{m} + \vec{e_k}) - \mu(\vec{m}))^2$ --- applying Eq.~(\ref{eq:rational}), there are two possibilities. \emph{First},
\begin{align*}
    &\big(2n_k + 2(2n - 2D + 1)) =  \big(2m_k + 2(2m - 2D + 1)), \text{ and }\nonumber\\
    &(2n_{k'} + 2n - 2D + 1) = (2m_{k'} + 2m - 2D + 1) \text{ for }k' \neq k,
\end{align*}
which implies that $\vec{n} = \vec{m}$. \emph{Second}, 
\begin{align*}
    &\big(2n_k + 2(2n - 2D + 1)) =  -\big(2m_k + 2(2m - 2D + 1)), \text{ and }\nonumber\\
    &(2n_{k'} + 2n - 2D + 1) = -(2m_{k'} + 2m - 2D + 1) \text{ for }k' \neq k,
\end{align*}
which implies that $5(\smallnorm{\vec{n}}_1 + \smallnorm{\vec{m}}_1) = 4(2D - 1)$ --- since we can always choose $D$ such that $2D - 1$ is not divisible by 5, this equation will have no solutions. Therefore, we conclude that as long as $\vec{n} \neq \vec{m}$, $(\mu(\vec{n} + \vec{e_k}) - \mu(\vec{n}))^2 \neq (\mu(\vec{m} + \vec{e_k}) - \mu(\vec{m}))^2$, thus confirming that $H_0$ as constructed in Eq.~(\ref{eq:H_0_blockade}) satisfies the non-degeneracy condition in lemma \ref{lemma:controllability}.

Thus, we have shown that $\hat{H}_R^c(t)$ as in Eq.~(\ref{eq:final_blockaded_hamiltonian}) can be used to implement any unitary within the blockaded subspace, thus allowing us to implement the re-absorption unitary required in the optimal measurement setup. Furthermore, since the magnitude of $\Gamma_k(t)$ in Eq.~(\ref{eq:final_blockaded_hamiltonian}) are determined only by the coherent drive applied on the oscillators, the speed of applying a unitary on the blockaded subspace is not limited by the non-linear strength $\chi$.

\subsection{Optimality of photodetection and linear optics}\label{app:opt_measurement_pdlo}
\subsubsection{Independent sources}
We recall the setup introduced in the main text Section \ref{sec:pdlo} --- a linear optical element is applied at the output of the MZI and is tuned depending on the output of the photodetectors [Fig.~\ref{fig:beam_splitter}(b)]. The result of the photodetection would be a sequence of times $0 \leq \tau_1 < \tau_2 \dots < \tau_n < \dots$ at which the photons are detected as well as $\sigma_1, \sigma_2 \dots \sigma_n \dots \in \{a, b\}$ indicating the port in which a photon has been detected. Having obtained this photodetection record until the $n^\text{th}$ detection event and before the next detection event, we apply a linear optical element described by the unitary $U^{\vec{\sigma}^n}_{\vec{\tau}^n}(t)$, where $\vec{\sigma}^n = \{\sigma_n, \sigma_{n - 1} \dots \sigma_1\}, \vec{\tau}^n = \{\tau_n, \tau_{n - 1}\dots \tau_1\}$, for $t > \tau_n$. Suppose that the next detection happens at $\tau_{n + 1} > \tau_n$: The annihilation operators at the two output ports at this time, $a_{\tau_{n + 1}}$ and $b_{\tau_{n + 1}}$, after the application of the unitary $U^{\vec{\sigma}^n}_{\vec{\tau}^n}(t)$ will be given by
\begin{align}\label{eq:time_dep_bs}
U^{\vec{\sigma}^n \dagger}_{\vec{\tau}^n}(\tau_{n + 1})\begin{bmatrix}
a_{\tau_{n + 1}} \\
b_{\tau_{n + 1}}
\end{bmatrix} U^{\vec{\sigma}^n}_{\vec{\tau}^n}(\tau_{n + 1}) = \underbrace{\begin{bmatrix}
    (V_{a, a}(\tau_{n + 1}))^{\vec{\sigma}^n}_{\vec{\tau}^n} & (V_{a, b}(\tau_{n + 1}))^{\vec{\sigma}^n}_{\vec{\tau}^n} \\
    (V_{b, a}(\tau_{n + 1}))^{\vec{\sigma}^n}_{\vec{\tau}^n} & (V_{b, b}(\tau_{n + 1}))^{\vec{\sigma}^n}_{\vec{\tau}^n}
\end{bmatrix}}_{V^{\vec{\sigma}^n}_{\vec{\tau}^n}(\tau_{n + 1})}
\begin{bmatrix}
a_s \\
b_s
\end{bmatrix},
\end{align}
where $V^{\vec{\sigma}^n}_{\vec{\tau}^n}(\tau_{n + 1})$ is a $2\times 2$ unitary matrix. We can streamline the notation a lot more by introducing $p_{\vec{\tau}^k}^{\vec{\sigma}^k}, q_{\vec{\tau}^k}^{\vec{\sigma}^k}$ via
\begin{align}
    p_{\vec{\tau}^{k}}^{\{\vec{\sigma}^{k - 1}, a\}} = (V_{a, a}(\tau_k))^{\vec{\sigma}^{k - 1}}_{\vec{\tau}^{k - 1}},\   q_{\vec{\tau}^{k}}^{\{\vec{\sigma}^{k - 1}, a\}} = (V_{a, b}(\tau_k))^{\vec{\sigma}^{k - 1}}_{\vec{\tau}^{k - 1}},\ p_{\vec{\tau}^{k}}^{\{\vec{\sigma}^{k - 1}, b\}} = (V_{b, a}(\tau_k))^{\vec{\sigma}^{k - 1}}_{\vec{\tau}^{k - 1}} \ \text{ and }q_{\vec{\tau}^{k}}^{\{\vec{\sigma}^{k - 1}, b\}} = (V_{b, b}(\tau_k))^{\vec{\sigma}^{k - 1}}_{\vec{\tau}^{k - 1}}.
\end{align}
Using this notation, the measurement of a sequence of photodetection times $\vec{\tau}$ and a sequence of the ports $\vec{\sigma}$ at which these detections happen is equivalent to a projection on the state $\smallket{E_{\vec{\tau}}^{\vec{\sigma}}}$ where
\begin{align}\label{eq:linear_optics_measurement}
    \smallket{E_{\vec{\tau}}^{\vec{\sigma}}} = \prod_{k = 1}^{n}\big(\big(p^{\vec{\sigma}^k}_{\vec{\tau}^k}\big)^* a_{\tau_k}^\dagger + \big(q^{\vec{\sigma}^k}_{\vec{\tau}^k}\big)^* b_{\tau_k}^\dagger\big)  \ket{0},
\end{align}
where $\vec{\tau}^k = \{\tau_k, \tau_{k - 1} \dots \tau_1\}, \vec{\sigma}^k = \{\sigma_k, \sigma_{k - 1} \dots \sigma_1\}$. Using Eq.~(\ref{eq:time_dep_bs}), it also follows that $\smallbra{E^{\vec{\sigma}'}_{\vec{\tau}'}} E^{\vec{\sigma}}_{\vec{\tau}}\rangle = \delta_{\vec{\sigma}, \vec{\sigma}'}\delta(\vec{\tau} - \vec{\tau}')$. Finally, we remark that if $q^{\{\vec{\sigma}^{k - 1}, a\}}_{\vec{\tau}^k}, p^{\{\vec{\sigma}^{k - 1}, b\}}_{\vec{\tau}^k} = 0$ or $q^{\{\vec{\sigma}^{k - 1}, b\}}_{\vec{\tau}^k}, p^{\{\vec{\sigma}^{k - 1}, a\}}_{\vec{\tau}^k} = 0$ for all $\vec{\sigma}^{k - 1}, \vec{\tau}^k$, then this measurement protocol reduces to simply performing photodetection on the two output ports without any linear optics.

To understand the optimality of this measurement protocol, we will use the optimality criteria from Ref.~\cite{braunstein1994statistical}: A rank 1 projective measurement described by orthogonal states $\ket{E_x}$ optimally senses the parameter $\varphi$ from the state $\ket{\psi_\varphi}$ if and only if
\begin{equation}
    \text{Im}(\smallbra{E_x} \psi_\varphi \rangle \smallbra{\psi_\varphi^\perp} E_x\rangle) = 0 \qquad \forall x ,
\end{equation}
where 
\[
\ket{\psi_\varphi^\bot} = \left(1- \ket{\psi_\varphi}\!\bra{\psi_\varphi}\right)\frac{d}{d\varphi}\ket{\psi_\varphi}.
\]
Using the derivative of $\ket{\psi_\varphi}$ from Eq.~(\ref{eq:deriv_psi_phi}), we obtain that $\{\smallket{E^{\vec{\sigma}}_{\vec{\tau}}} \ \forall \ \vec{\sigma}, \vec{\tau}\}$ is optimal at $\varphi = 0$ if
\begin{align}\label{eq:optimality_condition_pd}
\smallbra{E_{\vec{\tau}}^{\vec{\sigma}}} \{\ket{\psi}\!\bra{\psi}, H_d\} \smallket{E_{\vec{\tau}}^{\vec{\sigma}}} = 2\smallabs{\smallbra{\psi} E_{\vec{\tau}}^{\vec{\sigma}}\rangle}^2 \smallbra{\psi} H_d \smallket{\psi}.
\end{align}
The rest of our analysis will be based on this optimality condition.

We first consider the case where two independent sources emit photons into the two input ports, i.e.,~the state $\ket{\psi} = \ket{\psi_A}\otimes \ket{\psi_B}$. It will be convenient to define the wave-functions corresponding to the states $\ket{\psi_A}, \ket{\psi_B}$: Given  $\vec{\tau} = \{\tau_1, \tau_2 \dots \tau_n\}$, we will define for $X \in \{A, B\}$
\[
\Psi_X(\vec{\tau}) = \bra{\text{vac}}\bigg(\prod_{i = 1}^n x_{\tau_i} \bigg)\ket{\psi_X},
\]
with $\Psi_X(\emptyset) = \bra{\text{vac}}\psi_X\rangle$. We will additionally assume that
\begin{align}
\Psi_X(\vec{\tau}) \neq 0 \ \forall \ \vec{\tau},
\end{align}
which is equivalent to assume that the wave-functions corresponding to the states $\ket{\psi_A}, \ket{\psi_B}$ are never exactly zero. Without loss of generality, we will assume $\Psi_X(\emptyset) > 0$ for both $X \in \{A, B\}$, and define phase $\Theta_X(\vec{\tau})$ via
\[
\Psi_X(\vec{\tau}) = \smallabs{\Psi_X(\vec{\tau})} e^{i \Theta_X(\vec{\tau})}.
\]

Our main result for the case of independent sources is that if photodetection with time and measurement-record dependent linear optics is optimal, then photodetection itself is an optimal measurement. To show this, we begin by applying Eq.~(\ref{eq:optimality_condition_pd}) for $\vec{\tau}, \vec{\sigma} =\emptyset $, which corresponds to the measurement outcome of no photons being detected. Together with $\Psi_X(\emptyset) \neq 0$, this yields
\[
\bra{\psi} H_d \ket{\psi} = 0.
\]
Next, we consider the case of a single photodetection event and apply Eq.~(\ref{eq:optimality_condition_pd}) for $\vec{\tau}^1 = \{\tau_1\}, \vec{\sigma}^1 = \{\sigma_1\}$. We obtain that
\begin{subequations}\label{eq:single_photon_optimality}
\begin{align}
\begin{bmatrix}
    p_{\tau_1}^{\sigma_1} & q_{\tau_1}^{\sigma_1}
\end{bmatrix} \Lambda^{(1)}(\vec{\tau}^1) \begin{bmatrix}
p_{\tau_1}^{\sigma_1 *}\\
q_{\tau_1}^{\sigma_1 *}
\end{bmatrix} = 0,
\end{align}
where $\Lambda^{(1)}(\vec{\tau}^1)$ is a $2 \times 2$ Hermitian matrix depending on $\tau_1$ and whose diagonal elements given by
\begin{align}
[\Lambda^{(1)}(\vec{\tau}^1)]_{1,1} = [\Lambda^{(1)}(\vec{\tau}^1)]_{2,2} =  \text{Im}(\Psi_A(\tau_1) \Psi_A(\emptyset) \Psi_B^*(\emptyset) \Psi_B^*(\tau_1)).
\end{align}
\end{subequations}
We recall that Eq.~(\ref{eq:single_photon_optimality}a) needs to have two orthogonal solutions corresponding to $\sigma_1 = a$ and $b$. This is only possible if $\text{Tr}[\Lambda^{(1)}(\vec{\tau}^1)] = 0$ \cite{horn2012matrix}, but from Eq.~(\ref{eq:single_photon_optimality}b), this would also imply that $[\Lambda^{(1)}(\vec{\tau}^1)]_{1,1} = [\Lambda^{(1)}(\vec{\tau}^1)]_{2,2} = 0$, i.e.,~$\Lambda^{(1)}(\vec{\tau}^1)$ is a matrix with 0 diagonal elements. Then, it follows that the possible solutions to Eq.~(\ref{eq:single_photon_optimality}a) are (up to a global phase that can be dropped)
\begin{align}\label{eq:single_photon_solution}
\begin{bmatrix}
p^{a}_{\tau_1} \\
q^{a}_{\tau_1}
\end{bmatrix} = \begin{bmatrix}
1 \\
0
\end{bmatrix},
\begin{bmatrix}
p^{b}_{\tau_1} \\
q^{b}_{\tau_1}
\end{bmatrix} = \begin{bmatrix}
0 \\
1
\end{bmatrix},
\end{align}
which would imply that photodetection, without any linear optics, remains optimal until the first photodetection event. We point out that $[p^a_{\tau_1}; q^a_{\tau_1}] = [0; 1]$ and $[p^b_{\tau_1}; q^b_{\tau_1}] = [1; 0]$ are also possible solutions to Eq.~(\ref{eq:single_photon_optimality}a), but this simply corresponds to swapping the two photodetectors before the first photodetection [see Eq.~(\ref{eq:time_dep_bs})] and is thus equivalent to the solution in Eq.~(\ref{eq:single_photon_solution}). Finally, it is useful to note that it follows from Eq.~(\ref{eq:single_photon_optimality}b) and $[\Lambda^{(1)}(\vec{\tau}^1)]_{1,1} = [\Lambda^{(1)}(\vec{\tau}^1)]_{2, 2} = 0$ that
\begin{align}\label{eq:phase_condition_single_ph}
\Theta_A(\tau_1) = \Theta_B(\tau_1),
\end{align}
which we will use subsequently.

Proceeding similarly, we next consider the case of two photodetection events and apply Eq.~(\ref{eq:optimality_condition_pd}) for $\vec{\tau}^2 = \{\tau_1, \tau_2\}$ and $\vec{\sigma}^2 = \{\sigma_1, \sigma_2\}$. We obtain 
\begin{align}\label{eq:optimality_two_photon}
    \begin{bmatrix}
        p_{\tau_1, \tau_2}^{\sigma_1, \sigma_2} & q_{\tau_1, \tau_2}^{\sigma_1, \sigma_2}
    \end{bmatrix}
    \Lambda^{(2)}_{\sigma_1}(\vec{\tau}^2)
    \begin{bmatrix}
        p_{\tau_1,\tau_2}^{\sigma_1, \sigma_2*} \\ q_{\tau_1,\tau_2}^{\sigma_1, \sigma_2*}
    \end{bmatrix}
    =0,
\end{align}
where $\Lambda^{(2)}_{\sigma_1}(\vec{\tau})$ is again a 2 by 2 Hermitian matrix which depends on $\tau_1, \tau_2$ (the times at which the photon is detected) and $\sigma_1$ (the port at which the first photon is detected). For $\sigma_1 = 0$,
\begin{subequations}\label{eq:diagonal_elem_two_ph}
\begin{align}
    [\Lambda_{a}^{(2)}(\vec{\tau}^2)]_{1, 1} &= \text{Im}\big(\big(\Psi_A(\tau_1) \Psi_B(\tau_2) + \Psi_A(\tau_2) \Psi_B(\tau_1)\big)  \Psi_A^*(\tau_1, \tau_2) \Psi_B^*(\emptyset)\big), \nonumber \\
    [\Lambda_a^{(2)}(\vec{\tau}^2)]_{2, 2} &= \text{Im}\big(\big(\Psi_A(\emptyset) \Psi_B(\tau_1, \tau_2) - \Psi_A(\tau_1, \tau_2) \Psi_B(\emptyset)\big)\Psi_A^*(\tau_1) \Psi_B^*(\tau_2)\big),
\end{align}
and for $\sigma_1 = 1$,
\begin{align}
    [\Lambda_{b}^{(2)}(\vec{\tau}^2)]_{1, 1} &= \text{Im}\big(\big(\Psi_A(\emptyset) \Psi_B(\tau_1, \tau_2) - \Psi_A(\tau_1, \tau_2) \Psi_B(\emptyset)\big)  \Psi_A^*(\tau_2) \Psi_B^*(\tau_1)\big), \nonumber \\
    [\Lambda_b^{(2)}(\vec{\tau}^2)]_{2, 2} &=- \text{Im}\big(\big(\Psi_A(\tau_1) \Psi_B( \tau_2) +  \Psi_A(\tau_2) \Psi_B(\tau_1)\big)\Psi_A^*(\emptyset) \Psi_B^*(\tau_1, \tau_2)\big).
\end{align}
\end{subequations}
Again, since Eq.~(\ref{eq:optimality_two_photon}) needs to have two orthogonal solutions (corresponding to $\sigma_2 = 0$ and 1), we obtain that $\text{Tr}(\Lambda^{(2)}_0(\vec{\tau})) = \text{Tr}(\Lambda^{(2)}_1(\vec{\tau})) = 0$. Using Eq.~(\ref{eq:phase_condition_single_ph}), this condition can be rewritten as 
\begin{equation}
    \text{M}^{(2)}(\tau_1, \tau_2)
    \begin{bmatrix}
        \text{Im}\big(\Psi_A^*(\tau_1, \tau_2) \Psi_B^*(\emptyset) e^{i(\chi(t_1) + \chi(t_2))}\big) \\ 
        \text{Im}\big(\Psi_B^*(\tau_1, \tau_2) \Psi_A^*(\emptyset) e^{i(\chi(t_1) + \chi(t_2))}\big)
    \end{bmatrix} 
    = 0 ,
\end{equation}
where $\chi(t) = \Theta_A(t) = \Theta_B(t)$ and
\begin{equation}
    \text{M}^{(2)}(\tau_1, \tau_2) = 
    \begin{pmatrix}
        2\abs{\Psi_A(\tau_1)}\abs{\Psi_B(\tau_2)} + \abs{\Psi_A(\tau_2)}\abs{\Psi_B(\tau_1)} & -\abs{\Psi_A(\tau_1)}\abs{\Psi_B(\tau_2)} \\ 
        -\abs{\Psi_A(\tau_2)}\abs{\Psi_B(\tau_1)} & 2\abs{\Psi_A(\tau_2)}\abs{\Psi_B(\tau_1)} + \abs{\Psi_A(\tau_1)}\abs{\Psi_B(\tau_2)} 
    \end{pmatrix}.
\end{equation}
Since $\Psi_A(\tau), \Psi_B(\tau) \neq 0$ by assumption,  by computing the determinant of $\text{M}^{(2)}(\tau_1, \tau_2)$ it can be verified that it is invertible. Therefore, it follows that $\text{Im}\big(\Psi_A^*(\tau_1, \tau_2) \Psi_B^*(\emptyset) e^{i(\chi(t_1) + \chi(t_2))}\big) = 0$ and $\text{Im}\big(\Psi_B^*(\tau_1, \tau_2) \Psi_A^*(\emptyset) e^{i(\chi(\tau_1) + \chi(\tau_2))}\big) = 0$. Since we have fixed $\Psi_A(\emptyset), \Psi_B(\emptyset) > 0$, this is equivalent to
\begin{equation}\label{eq:phase_condition_two_photon}
    \Theta_A(\tau_1, \tau_2) = \Theta_B(\tau_1, \tau_2) = \chi(\tau_1) + \chi(\tau_2).
\end{equation}
Returning to Eqs.~(\ref{eq:optimality_two_photon}) and (\ref{eq:diagonal_elem_two_ph}), we see that Eq.~(\ref{eq:phase_condition_two_photon}) implies that $\Lambda^{(2)}(\vec{\tau})$ has $0$ diagonal elements, and consequently similar to the case of a single photodetection even, photodetection without any linear optics remains an optimal measurement.

This argument can be continued for any number of photodetection events via induction: We assume that, for $\vec{\tau}^n = \{\tau_1, \tau_2 \dots \tau_n\}$ (which corresponds to $n$ detection events), 
\begin{align}\label{eq:phase_condition_ind}
\Theta_A(\vec{\tau}^n) = \Theta_B(\vec{\tau}^n) = \sum_{j = 1}^n\chi(\tau_j),
\end{align}
and apply the optimality condition [Eq.~(\ref{eq:optimality_condition_pd})] for $\vec{\sigma}^{n + 1} = \{\sigma_1, \sigma_2 \cdots \sigma_{n + 1} \}$, and $\vec{\tau}^{n + 1} = \{ \tau_1, \tau_2 \cdots\tau_{n + 1} \} $. Doing so, we obtain that
\begin{align}\label{eq:optimality_n_photon}
    \begin{bmatrix}
        p_{\vec{\tau}^{n + 1}}^{\vec{\sigma}^{n + 1}} & q_{\vec{\tau}^{n + 1}}^{\vec{\sigma}^{n + 1}}
    \end{bmatrix} \Lambda_{\vec{\sigma}^n}^{(n + 1)}(\vec{\tau}^{n + 1}) \begin{bmatrix}
    p_{\vec{\tau}^{n + 1}}^{\vec{\sigma}^{n + 1} *}\\
    q_{\vec{\tau}^{n + 1}}^{\vec{\sigma}^{n + 1} *}
    \end{bmatrix} = 0,
\end{align}
where $\Lambda_{\sigma^n}^{(n + 1)}(\vec{\tau}^{n + 1})$ is a 2 by 2 matrix. Since Eq.~(\ref{eq:optimality_n_photon}) needs to have two orthogonal solutions (corresponding to $\sigma_{n+ 1} =a$ and $\sigma_{n + 1}=b$), it follows that $\text{Tr}(\Lambda_{\sigma^n}^{(n + 1)}(\vec{\tau}^{n + 1})) = 0$. Next, given $l \in \{1, 2 \dots n\}$ we choose $\vec{\sigma}^n = \{\sigma_1, \sigma_2 \dots \sigma_n\}$ where $\sigma_k = a$ for $k \neq l$ and $\sigma_l = b$, i.e., $\vec{\sigma}^n$ corresponds to a photodetection record where all but one of the first $n$ photons are detected at port $a$, with a single photon detected at port $b$ at time $\tau_l$. For this choice of $\vec{\sigma}^n$, the condition for $\text{Tr}(\Lambda^{(n + 1)}_{\vec{\sigma}^n}(\vec{\tau}))$ can be written as
    \begin{equation}
        \Im \bigg( \Psi_A(\vec{\tau}^{n + 1} ) \Psi_B(\emptyset) \Psi_A^*(\vec{\tau}^{n + 1} \setminus \tau_l) \Psi_B^*(\{\tau_l\}) \bigg) = 0.
    \end{equation}
It thus follows that
\begin{equation}
        \Theta_A(\vec{\tau}^{n + 1})  = \Theta_A(\vec{\tau}^{n + 1}\setminus \tau_l) + \Theta_B(\{\tau_l\}) = \sum_{j = 1}^{n + 1} \chi(\tau_j),
\end{equation}
where, in the last step, we have used the induction hypothesis [Eq.~(\ref{eq:phase_condition_ind})]. Performing a similar analysis while choosing $\vec{\sigma}^n$ via $\sigma_k = b$ if $k \neq l$ and $\sigma_l = a$, we obtain a similar result for $\Theta_B(\vec{\tau}^{n + 1})$. By induction, we thus conclude that Eq.~(\ref{eq:phase_condition_ind}) is satisfied for all $n$.

Finally, we show that Eq.~(\ref{eq:phase_condition_ind}) implies that photodetection alone, without any linear optical elements, is optimal. We recall that photo-detection corresponds to a projective measurement on the states $\smallket{E_{\vec{\tau}_a, \vec{\tau}_b}}$ where
\begin{equation}
    \smallket{E_{\vec{\tau}_a,\vec{\tau}_b}}  = \prod_{\tau\in \vec{\tau}_a} a_{\tau}^\dagger \prod_{\tau' \in \vec{\tau}_b} b_{\tau'}^\dagger  \ket{0},
\end{equation}
where $\vec{\tau}_x$ denotes the set of detection times recorded at port $x \in \{a, b\}$. Then, checking the optimality condition in Eq.~(\ref{eq:optimality_condition_pd}) is equivalent to checking if
\begin{equation}
    \Im\bigg(\bigg[\sum_{\tau\in\vec{\tau}_a }\Psi_A(\vec{\tau}_a \setminus \tau)\Psi_B( \vec{\tau}_b \cup \tau)- \sum_{t\in\vec{\tau}_b}\Psi_A(\vec{\tau}_a \cup \tau)\Psi_B( \vec{\tau}_b \setminus \tau)\bigg]
    \Psi_A^*(\vec{\tau}_a)\Psi_B^*(\vec{\tau}_b)\bigg) = 0.
\end{equation}
It can now easily be verified that this equation holds if Eq.~(\ref{eq:phase_condition_ind}) is true, thus showing that photodetection alone is an optimal measurement. 

\subsubsection{Counterexample: entangled state}\label{app:linear_optics_entangled}
Here, we construct an example where linear optics and photodetection is optimal but photodetection alone is not. We begin by defining
\begin{equation}
    \Psi(\vec{\tau}_a;\vec{\tau}_b) = \bra{\text{vac}}\prod_{\tau \in \vec{\tau}_a} a_{\tau}\prod_{\tau \in \vec{\tau}_b} b_{\tau} \ket{\psi}.
\end{equation}
We will consider a photonic state where a single source emits into both the input ports of the MZI. Consequently, the photons in the two different ports can be entangled and $\Psi(\vec{\tau}_a; \vec{\tau}_b)$ does not necessarily factorize into a product of the form $\Psi_A(\vec{\tau}_a)\Psi_B(\vec{\tau}_b)$.

As an explicit example, consider a two-photon wavepacket given by
\begin{equation}
    \Psi(\tau_1,\tau_2;\emptyset) = f(\tau_1,\tau_2), \quad 
    \Psi(\tau_1;\tau_2) = -f(\tau_1,\tau_2), \quad
    \Psi(\emptyset; \tau_1,\tau_2) = (1+i)f(\tau_1,\tau_2), \quad
    \Psi(\tau_2;\tau_1) = (1+i)f(\tau_1,\tau_2),
\end{equation}
where $f(\tau_1,\tau_2)$ is a symmetric function chosen such that the wave function is normalized. First we note that for this specific example, $\bra{\psi} H_d \ket{\psi}$ is equal to zero, and the condition in Eq.~(\ref{eq:single_photon_optimality}) is satisfied since it's a two photon wavepacket. The diagonal entries of the matrix $\Lambda_{\sigma_1}^{(2)}(\vec{\tau})$ in Eq.~(\ref{eq:optimality_two_photon}) (corresponding to the optimality of two photodetection events) for $\sigma_1 = a$ and $b$ are 
\begin{subequations}
\begin{align}
    (\Lambda_{a}^{(2)}(\vec{\tau}))_{1, 1} &= \text{Im}\big(\big(\Psi(\tau_1;\tau_2) + \Psi(\tau_2;\tau_1)\big)  \Psi^*(\tau_1, \tau_2;\emptyset)\big) = \abs{f(\tau_1,\tau_2)}^2, \nonumber \\
    (\Lambda_a^{(2)}(\vec{\tau}))_{2, 2} &= \text{Im}\big(\big(\Psi(\emptyset;\tau_1, \tau_2) - \Psi(\tau_1, \tau_2;\emptyset)\big)\Psi^*(\tau_1;\tau_2)\big)= - \abs{f(\tau_1,\tau_2)}^2,
\end{align}
and,
\begin{align}
    (\Lambda_{b}^{(2)}(\vec{\tau}))_{1, 1} &= \text{Im}\big(\big(\Psi(\emptyset;\tau_1, \tau_2) - \nonumber \Psi(\tau_1, \tau_2;\emptyset)\big)  \Psi^*(\tau_2;\tau_1)\big) =  \abs{f(\tau_1,\tau_2)}^2, \\
    (\Lambda_b^{(2)}(\vec{\tau}))_{2, 2} &=- \text{Im}\big(\big(\Psi(\tau_1; \tau_2) + \Psi(\tau_2;\tau_1)\big)\Psi^*(\emptyset;\tau_1, \tau_2)\big) =  - \abs{f(\tau_1,\tau_2)}^2.
\end{align}
\end{subequations}
Clearly, we have $\text{Tr}(\Lambda_{\sigma_1}^{(2)}(\vec{\tau})) = 0$,  indicating that a measurement using linear optics and photodetection is optimal. However, the diagonal elements of $\text{Tr}(\Lambda_{\sigma_1}^{(2)}(\vec{\tau}))$ are nonzero. This implies that the two orthonormal basis solutions of Eq.~(\ref{eq:optimality_two_photon}) do not correspond to photodetection. Therefore, photodetection is sub-optimal as a measurement but becomes optimal when supplemented with linear optical elements. 

\end{document}